\documentclass[submission]{dmtcs-episciences}

\usepackage[utf8]{inputenc}
\usepackage{subfigure}
\usepackage[round]{natbib}
\usepackage{amsmath}
\usepackage{amssymb}

\newcommand{\bis}{\protect\underline{\leftrightarrow}}

\newcommand{\doublera}{\rightarrow\!\!\!\!\!\rightarrow}
\newcommand{\cho}{[]}
\newcommand{\rs}{\ {\sf rs}\ }
\newcommand{\sy}{\ {\sf sy}\ }
\newcommand{\sr}{\ {\sf sr}\ }

\newtheorem{definition}{Definition}[section]
\newtheorem{example}{Example}[section]
\newtheorem{proposition}{Proposition}[section]
\newtheorem{theorem}{Theorem}[section]

\author{Igor V. Tarasyuk\affiliationmark{1}\thanks{Partially supported by Deutsche Forschungsgemeinschaft under Grant
BE 1267/14-1.} \and Hermenegilda Maci\`a\affiliationmark{2} \and Valent\'in Valero\affiliationmark{2}}

\title[Stochastic equivalence for performance analysis in dtsiPBC]{Stochastic equivalence for performance analysis of
concurrent systems in dtsiPBC\thanks{This work was supported in part by the Spanish Ministry of Science and Innovation
and the European Union FEDER Funds with the coordinated Project DArDOS entitled ``Formal development and analysis of
complex systems in distributed contexts: foundations, tools and applications'', UCLM subproject ``Formal analysis and
applications of Web services and electronic contracts'', under Grant TIN2015-65845-C3-2-R.}}

\affiliation{A.P. Ershov Institute of Informatics Systems, SB RAS, Novosibirsk, Russian Federation\\
High School of Informatics Engineering, University of Castilla - La Mancha, Albacete, Spain}


\received{YYYY-MM-DD}
\revised{YYYY-MM-DD}
\accepted{YYYY-MM-DD}

\begin{document}

\publicationdetails{VOL}{2017}{ISS}{NUM}{SUBM}

\maketitle

\begin{abstract}
We propose an extension with immediate multiactions of discrete time stochastic Petri Box Calculus (dtsPBC), presented
by I.V. Tarasyuk. The resulting algebra dtsiPBC is a discrete time analogue of stochastic Petri Box Calculus (sPBC)
with immediate multiactions, designed by H. Maci\`a, V. Valero et al. within a continuous time domain. The step
operational semantics is constructed via labeled probabilistic transition systems. The denotational semantics is based
on labeled discrete time stochastic Petri nets with immediate transitions. To evaluate performance, the corresponding
semi-Markov chains are analyzed. We define step stochastic bisimulation equivalence of expressions that is applied to
reduce their transition systems and underlying semi-Markov chains while preserving the functionality and performance
characteristics. We explain how this equivalence can be used to simplify performance analysis of the algebraic
processes. In a case study, a method of modeling, performance evaluation and behaviour reduction for concurrent systems
is outlined and applied to the shared memory system.
\end{abstract}

\section{Introduction}
\label{introduction.sec}

Algebraic process calculi like CSP \cite{Hoa85}, ACP \cite{BK85} and CCS \cite{Mil89} are well-known formal models for
specification of computing systems and analysis of their behaviour. In such process algebras (PAs), systems and
processes are specified by formulas, and verification of their properties is accomplished at a syntactic level via
equivalences, axioms and inference rules. In recent decades, stochastic extensions of PAs were proposed, such as MTIPP
\cite{HR94}, PEPA \cite{Hil96} and EMPA \cite{BGo98,BDGo98,Bern99}. Unlike standard PAs, stochastic process algebras
(SPAs) do not just specify actions which can occur (qualitative features), but they associate with the actions the
distribution parameters of their random time delays (quantitative characteristics).

\subsection{Petri Box Calculus}

PAs specify concurrent systems in a compositional way via an expressive formal syntax. On the other hand, Petri nets
(PNs) provide a graphical representation of such systems and capture explicit asynchrony in their behaviour. To combine
the advantages of both models, a semantics of algebraic formulas in terms of PNs has been defined. Petri Box Calculus
(PBC) \cite{BDH92,BKo95,BDK01} is a flexible and expressive process algebra developed as a tool for specification of
the PNs structure and their interrelations. Its goal was also to propose a compositional semantics for high level
constructs of concurrent programming languages in terms of elementary PNs. Formulas of PBC are combined not from single
(visible or invisible) actions and variables, like in CCS, but from multisets of elementary actions and their
conjugates, called multiactions ({\em basic formulas}). The empty multiset of actions is interpreted as the silent
multiaction specifying some invisible activity. In contrast to CCS, synchronization is separated from parallelism ({\em
concurrent constructs}). Synchronization is a unary multi-way stepwise operation based on communication of actions and
their conjugates. This extends the CCS approach with conjugate matching labels. Synchronization in PBC is asynchronous,
unlike that in Synchronous CCS (SCCS) \cite{Mil89}. Other operations are sequence and choice ({\em sequential
constructs}). The calculus includes also restriction and relabeling ({\em abstraction constructs}). To specify infinite
processes, refinement, recursion and iteration operations were added ({\em hierarchical constructs}). Thus, unlike CCS,
PBC has an additional iteration operation to specify infinite behaviour when the semantic interpretation in finite PNs
is possible. PBC has a step operational semantics in terms of labeled transition systems, based on the rules of
structural operational semantics (SOS) \cite{Plo81}. The operational semantics of PBC is of step type, since its SOS
rules have transitions with (multi)sets of activities, corresponding to simultaneous executions of activities (steps).
Note that we do not reason in terms of a big-step (natural) \cite{Kah87} or small-step (structural) \cite{Plo81}
operational semantics here, and that PBC (and all its extensions to be mentioned further) have a small-step operational
semantics, in that terminology. A denotational semantics of PBC was proposed via a subclass of PNs equipped with an
interface and considered up to isomorphism, called Petri boxes. For more detailed comparison of PBC with other process
algebras and the reasoning about importance of non-interleaving semantics see \cite{BDH92,BDK01}.

\subsection{Stochastic extensions of Petri Box Calculus}

A stochastic extension of PBC, called stochastic Petri Box Calculus (sPBC), was proposed in \cite{MVF01}. In sPBC,
multiactions have stochastic delays that follow (negative) exponential distribution. Each multiaction is equipped with
a rate that is a parameter of the corresponding exponential distribution. The instantaneous execution of a stochastic
multiaction is possible only after the corresponding stochastic time delay. Just a finite part of PBC was initially
used for the stochastic enrichment, i.e. in its former version sPBC does not have refinement, recursion or iteration
operations. The calculus has an interleaving operational semantics defined via transition systems labeled with
multiactions and their rates. Its denotational semantics was defined in terms of a subclass of labeled continuous time
stochastic PNs, based on CTSPNs \cite{Mar90,Bal01} and called stochastic Petri boxes (s-boxes). In \cite{MVCC04}, the
iteration operator was added to sPBC. In sPBC with iteration, performance of the processes is evaluated by analyzing
their underlying continuous time Markov chains (CTMCs). In \cite{MVCF08}, a number of new equivalence relations were
proposed for regular terms of sPBC with iteration to choose later a suitable candidate for a congruence. sPBC with
iteration was enriched with immediate multiactions having zero delay in \cite{MVCR08}. We call such an sPBC extension
generalized sPBC or gsPBC. An interleaving operational semantics of gsPBC was constructed via transition systems
labeled with stochastic or immediate multiactions together with their rates or probabilities. A denotational semantics
of gsPBC was defined via a subclass of labeled generalized stochastic PNs, based on GSPNs \cite{Mar90,Bal01,Bal07} and
called generalized stochastic Petri boxes (gs-boxes). The performance analysis in gsPBC is based on the underlying
semi-Markov chains (SMCs).

PBC has a step operational semantics, whereas sPBC has an interleaving one. Remember that in step semantics, parallel
executions of activities (steps) are permitted while in interleaving semantics, we can execute only single activities.
Hence, a stochastic extension of PBC with a step semantics is needed to keep the concurrency degree of behavioural
analysis at the same level as in PBC. As mentioned in \cite{Mol81,Mol85}, in contrast to continuous time approach (used
in sPBC), discrete time approach allows for constructing models of common clock systems and clocked devices. In such
models, multiple transition firings (or executions of multiple activities) at time moments (ticks of the central clock)
are possible, resulting in a step semantics. Moreover, employment of discrete stochastic time fills the gap between the
models with deterministic (fixed) time delays and those with continuous stochastic time delays. As argued in
\cite{AHR00}, arbitrary delay distributions are much easier to handle in a discrete time domain. In
\cite{MVi08,MVi09,MABV12}, discrete stochastic time was preferred to enable simultaneous expiration of multiple delays.
In \cite{Tar05,Tar07}, a discrete time stochastic extension dtsPBC of finite PBC was presented. In dtsPBC, the
residence time in the process states is geometrically distributed. A step operational semantics of dtsPBC was
constructed via labeled probabilistic transition systems. Its denotational semantics was defined in terms of a subclass
of labeled discrete time stochastic PNs (LDTSPNs), based on DTSPNs \cite{Mol81,Mol85} and called discrete time
stochastic Petri boxes (dts-boxes). A variety of stochastic equivalences were proposed to identify stochastic processes
with similar behaviour which are differentiated by the semantic equivalence. The interrelations of all the introduced
equivalences were studied. In \cite{Tar06,Tar14}, we constructed an enrichment of dtsPBC with the iteration operator
used to specify infinite processes. The performance evaluation in dtsPBC with iteration is accomplished via the
underlying discrete time Markov chains (DTMCs) of the algebraic processes. Since dtsPBC has a discrete time semantics
and geometrically distributed sojourn time in the process states, unlike sPBC with continuous time semantics and
exponentially distributed delays, the calculi apply two different approaches to the stochastic extension of PBC, in
spite of some similarity of their syntax and semantics inherited from PBC. The main advantage of dtsPBC is that
concurrency is treated like in PBC having step semantics, whereas in sPBC parallelism is simulated by interleaving,
obliging one to collect the information on causal independence of activities before constructing the semantics. In
\cite{TMV13,TMV14,TMV15}, we presented the extension dtsiPBC of the latter calculus with immediate multiactions.
Immediate multiactions increase the specification capability: they can model logical conditions, probabilistic
branching, instantaneous probabilistic choices and activities whose durations are negligible in comparison with those
of others. They are also used to specify urgent activities and the ones that are not relevant for performance
evaluation. Thus, immediate multiactions can be considered as a kind of instantaneous dynamic state adjustment and, in
many cases, they result in a simpler and more clear system representation.

\subsection{Equivalence relations}

A notion of equivalence is important in theory of computing systems. Equivalences are applied both to compare behaviour
of systems and reduce their structure. There is a wide diversity of behavioural equivalences, and their interrelations
are well explored in the literature. The best-known and widely used one is bisimulation. Typically, the mentioned
equivalences take into account only functional (qualitative) but not performance (quantitative) aspects. Additionally,
the equivalences are usually interleaving ones, i.e. they interpret concurrency as a sequential nondeterminism.
Interleaving equivalences permit to imitate parallel execution of actions via all possible occurrence sequences
(interleavings) of them. Step equivalences require instead simulating such a parallel execution by simultaneous
occurrence (step) of all the involved actions. To respect quantitative features of behaviour, probabilistic
equivalences have additional requirement on execution probabilities. Two equivalent processes must be able to execute
the same sequences of actions, and for every such sequence, its execution probabilities within both processes should
coincide. In case of probabilistic bisimulation equivalence, the states from which similar future behaviours start are
grouped into equivalence classes that form elements of the aggregated state space. From every two bisimilar states, the
same actions can be executed, and the subsequent states resulting from execution of an action belong to the same
equivalence class. In addition, for both states, the cumulative probabilities to move to the same equivalence class by
executing the same action coincide. A different kind of quantitative relations is called Markovian equivalences, which
take rate (the parameter of exponential distribution that governs time delays) instead of probability. The
probabilistic equivalences can be seen as discrete time analogues of the Markovian ones, since the latter are defined
as the continuous time relations.

Interleaving probabilistic weak trace equivalence was introduced in \cite{Chr90} on labeled probabilistic transition
systems. Interleaving probabilistic strong bisimulation equivalence was proposed in \cite{LS91} on the same model.
Interleaving probabilistic equivalences were defined for probabilistic processes in \cite{JS90,GSS95}. Interleaving
Markovian weak bisimulation equivalences were considered in \cite{Buc94a} on Markovian process algebras, in
\cite{Buc95} on labeled CTSPNs and in \cite{Buc98} on labeled GSPNs. Interleaving Markovian strong bisimulation
equivalence was constructed in \cite{HR94} for MTIPP, in \cite{Hil96} for PEPA and in \cite{BGo98,BDGo98,Bern99} for
EMPA. In \cite{Bern07,Bern15}, interleaving Markovian trace, test, strong and weak bisimulation equivalences were
compared on sequential and concurrent Markovian process calculi. However, no appropriate equivalence was defined for
concurrent SPAs. The non-interleaving bisimulation equivalence in GSMPA \cite{BBGo98,Bra02} uses ST-semantics for
action particles while in S$\pi$ \cite{Pri02} it is based on a sophisticated labeling.

\subsection{Our contributions}

We present dtsPBC with iteration extended with immediate multiactions, called {\em discrete time stochastic and
immediate Petri Box Calculus} (dtsiPBC), which is a discrete time analog of sPBC. The latter calculus has iteration and
immediate multiactions within the context of a continuous time domain. The step operational semantics is constructed
with the use of labeled probabilistic transition systems. The denotational semantics is defined in terms of a subclass
of labeled discrete time stochastic and immediate PNs (LDTSPNs with immediate transitions, LDTSIPNs), based on the
extension of DTSPNs with transition labeling and immediate transitions, called dtsi-boxes. The consistency of both
semantics is demonstrated. The corresponding stochastic process, the underlying SMC, is constructed and investigated,
with the purpose of performance evaluation, which is the same for both semantics. In addition, the alternative solution
methods are developed, based on the underlying DTMC. Further, we propose step stochastic bisimulation equivalence
allowing one to identify algebraic processes with similar behaviour that are however differentiated by the semantics of
the calculus. We examine the interrelations of the proposed relation with other equivalences of the algebra. We
describe how step stochastic bisimulation equivalence can be used to reduce transition systems of expressions and their
underlying SMCs while preserving the qualitative and the quantitative characteristics. We prove that the mentioned
equivalence guarantees identity of the stationary behaviour and the residence time properties in the equivalence
classes. This implies coincidence of performance indices based on steady-state probabilities of the modeled stochastic
systems. The equivalences possessing the property can be used to reduce the state space of a system and thus simplify
its performance evaluation,
which is usually a complex problem due to the state space explosion. We present a case study of a system with two
processors and a common shared memory explaining how to model concurrent systems within the calculus and analyze their
performance, as well as how to reduce the systems behaviour while preserving their performance indices and making
easier the performance evaluation. Finally, we consider differences and similarities between dtsiPBC and other SPAs to
determine the advantages of our calculus. The salient point of dtsiPBC is a combination of immediate multiactions,
discrete stochastic time and step semantics in an SPA.

Concerning differences from our previous
papers about dtsiPBC \cite{TMV13,TMV14,TMV15}, the present text is much more detailed and many new important results
have been added. In particular, immediate multiactions now have positive real-valued weights (instead of previously
used positive integer weights), all the used notions (such as numbering, functions collecting executable activities,
probability functions) are formally defined and completely explained with examples; the operational and denotational
semantics are given in full detail (the inaction, action rules, LDTSPNs and dtsi-boxes are extensively described and
discussed); compact illustrative examples (of standard and alternative solution methods) are presented; keeping
properties of original Markov chains (irreducibility, positive recurrence and aperiodicity) in their embedded and
state-aggregated versions is studied. The main new contribution of the paper, step stochastic bisimulation equivalence
of the process expressions, is introduced and checked for stationary behaviour preservation in the equivalence classes;
quotienting the transition systems, SMCs and DTMCs by the equivalence, as well as the resulting simplification of
performance evaluation, are considered; generalized variant of the shared memory system and quotients of its behaviour
by the equivalence are constructed. In the enhanced related work overview, strong points of dtsiPBC with respect to
other SPAs are detected; in the discussion, analytical solution, application area, concurrency interpretation and
general advantages of dtsiPBC are explained.
Thus, the main contributions of the paper are the following.
\begin{itemize}

\item Flexible and expressive discrete time SPA with immediate activities called dtsiPBC.

\item Step operational semantics in terms of labeled probabilistic transition systems.

\item Net denotational semantics via discrete time stochastic and immediate Petri nets.

\item Performance analysis based on the underlying SMCs and DTMCs of expressions.

\item Stochastic equivalence used for functionality- and performance-preserving reduction.

\item Extended case study showing how to apply the theoretical results in practice.

\end{itemize}

\subsection{Structure of the paper}

In Section \ref{syntax.sec}, the syntax of the calculus dtsiPBC is presented. In Section \ref{opersem.sec}, we
construct the operational semantics of the algebra in terms of labeled probabilistic transition systems. In Section
\ref{denosem.sec}, we propose the denotational semantics based on a subclass of LDTSIPNs. In Section
\ref{perfeval.sec}, the corresponding stochastic process is derived and analyzed. Step stochastic bisimulation
equivalence is defined and investigated in Section \ref{stocheqs.sec}. In Section \ref{reduction.sec}, we explain how
to reduce transition systems and underlying SMCs of process expressions modulo the equivalence. In Section
\ref{stationary.sec}, this equivalence is applied to the stationary behaviour comparison in the equivalence classes to
verify the performance preservation. In Section \ref{gshmsysim.sec}, the generalized shared memory system is presented
as a case study. The difference between dtsiPBC and other well-known SPAs is considered in Section \ref{relwork.sec}.
The advantages of dtsiPBC with respect to other SPAs are described in Section \ref{discussion.sec}. Section
\ref{conclusion.sec} summarizes the results obtained and outlines the research perspectives.

\section{Syntax}
\label{syntax.sec}

In this section, we propose the syntax of dtsiPBC. First, we recall a definition of multiset that is an extension of
the set notion by allowing several identical elements.

\begin{definition}
A finite {\em multiset (bag)} $M$ over a set $X$ is a mapping $M:X\rightarrow\naturals$ such that $|\{x\in X\mid
M(x)>0\}|<\infty$, i.e. it contains a finite number of elements ($\naturals$ is the set of all nonnegative integers).
\end{definition}

We denote the {\em set of all finite multisets} over a set $X$ by $\naturals_{\rm fin}^X$. Let $M,M'\in\naturals_{\rm
fin}^X$. The {\em cardinality} of $M$ is $|M|=\sum_{x\in X}M(x)$. We write $x\in M$ if $M(x)>0$ and $M\subseteq M'$ if
$\forall x\in X,\ M(x)\leq M'(x)$. We define $(M+M')(x)=M(x)+M'(x)$ and $(M-M')(x)=\max\{0,M(x)-M'(x)\}$. When $\forall
x\in X,\ M(x)\leq 1,\ M$ can be interpreted as a proper set $M\subseteq X$. The {\em set of all subsets (powerset)} of
$X$ is denoted by $2^X$.

Let $Act=\{a,b,\ldots\}$ be the set of {\em elementary actions}. Then $\widehat{Act}=\{\hat{a},\hat{b},\ldots\}$ is the
set of {\em conjugated actions (conjugates)} such that $\hat{a}\neq a$ and $\hat{\hat{a}}=a$. Let
$\mathcal{A}=Act\cup\widehat{Act}$ be the set of {\em all actions}, and $\mathcal{L}=\naturals_{\rm fin}^\mathcal{A}$
be the set of {\em all multiactions}. Note that $\emptyset\in\mathcal{L}$, this corresponds to an internal move, i.e.
the execution of a multiaction with no visible actions. The {\em alphabet} of $\alpha\in\mathcal{L}$ is defined as
$\mathcal{A}(\alpha )=\{x\in\mathcal{A}\mid\alpha (x)>0\}$.

A {\em stochastic multiaction} is a pair $(\alpha ,\rho )$, where $\alpha\in\mathcal{L}$ and $\rho\in (0;1)$ is the
{\em probability} of the multiaction $\alpha$. This probability is interpreted as that of independent execution of the
stochastic multiaction at the next discrete time moment. Such probabilities are used to calculate those to execute
(possibly empty) sets of stochastic multiactions after one time unit delay. The probabilities of stochastic
multiactions are required not to be equal to $1$ to avoid extra model complexity, since in this case
weights would be required
to make a choice when several stochastic multiactions with probability $1$ can be executed
from a state.
Furthermore,
stochastic multiactions with probability $1$ would occur in a step (parallel execution) and all other with the less
probabilities do not. In this case, some problems appear with conflicts resolving. See \cite{Mol81,Mol85} for the
discussion on SPNs. On the other hand, there is no sense to allow zero probabilities of multiactions, since they would
never be performed in this case. Let $\mathcal{SL}$ be the set of {\em all stochastic multiactions}.

An {\em immediate multiaction} is a pair $(\alpha ,\natural_l)$, where $\alpha\in\mathcal{L}$ and
$l\in\reals_{>0}=(0;+\infty )$ is the positive real-valued {\em weight} of the multiaction $\alpha$. This weight is
interpreted as a measure of importance (urgency, interest) or a bonus reward associated with execution of the immediate
multiaction at the current discrete time moment. Such weights are used to calculate the probabilities to execute sets
of immediate multiactions instantly. Immediate multiactions have a priority over stochastic ones. Thus, in a state
where both kinds of multiactions can occur, immediate multiactions always occur before stochastic ones. Stochastic and
immediate multiactions cannot participate together in some step (concurrent execution), i.e. the steps consisting only
of immediate multiactions or those including only stochastic multiactions are allowed. Let $\mathcal{IL}$ be the set of
{\em all immediate multiactions}.

Note that the same multiaction $\alpha\in\mathcal{L}$ may have different probabilities and weights in the same
specification. An {\em activity} is a stochastic or
immediate multiaction. Let $\mathcal{SIL}=\mathcal{SL}\cup\mathcal{IL}$ be the set of {\em all activities}. The {\em
alphabet} of a multiset of activities $\Upsilon\in\naturals_{\rm fin}^\mathcal{SIL}$ is defined as
$\mathcal{A}(\Upsilon )=\cup_{(\alpha ,\kappa )\in\Upsilon}\mathcal{A}(\alpha )$. For an activity $(\alpha ,\kappa
)\in\mathcal{SIL}$, we define its {\em multiaction part} as $\mathcal{L}(\alpha ,\kappa )=\alpha$ and its {\em
probability} or {\em weight part} as $\Omega (\alpha ,\kappa )=\kappa$ if $\kappa\in (0;1)$; or $\Omega (\alpha ,\kappa
)=l$ if $\kappa =\natural_l,\ l\in\reals_{>0}$. The {\em multiaction part} of a multiset of activities
$\Upsilon\in\naturals_{\rm fin}^\mathcal{SIL}$ is defined as $\mathcal{L}(\Upsilon )=\sum_{(\alpha ,\kappa
)\in\Upsilon}\alpha$.

Activities are combined into formulas (process expressions) by the
operations: {\em sequential execution} $;$, {\em choice} $\cho$, {\em parallelism} $\|$, {\em relabeling} $[f]$ of
actions, {\em restriction} $\!\!\rs\!\!$ over a single action, {\em synchronization} $\!\!\sy\!\!$ on an action and its
conjugate, and {\em iteration} $[\,*\,*\,]$ with three arguments: initialization, body and termination.

Sequential execution and choice have a standard interpretation, like in other process algebras, but parallelism does
not include synchronization, unlike the operation in CCS \cite{Mil89}.

Relabeling functions $f:\mathcal{A}\rightarrow\mathcal{A}$ are bijections preserving conjugates, i.e. $\forall
x\in\mathcal{A},\ f(\hat{x})=\widehat{f(x)}$. Relabeling is extended to multiactions
as usual: for $\alpha\in\mathcal{L}$, we define $f(\alpha )=\sum_{x\in\alpha}f(x)$. Relabeling is extended to the
multisets of activities as follows: for $\Upsilon\in\naturals_{\rm fin}^\mathcal{SIL}$, we define $f(\Upsilon
)=\sum_{(\alpha ,\kappa )\in\Upsilon}(f(\alpha ),\kappa )$.

Restriction over an elementary action $a\in Act$ means that, for a given expression, any process behaviour containing
$a$ or its conjugate $\hat{a}$ is not allowed.

Let $\alpha ,\beta\in\mathcal{L}$ be two multiactions such that for some elementary action $a\in Act$ we have
$a\in\alpha$ and $\hat{a}\in\beta$, or $\hat{a}\in\alpha$ and $a\in\beta$. Then, synchronization of $\alpha$ and
$\beta$ by $a$ is defined as\\
$(\alpha\oplus_a\beta )(x)=\left\{
\begin{array}{ll}
\alpha (x)+\beta (x)-1, & \mbox{if }x=a\mbox{ or }x=\hat{a};\\
\alpha (x)+\beta (x), & \mbox{otherwise}.
\end{array}
\right.$\\
In other words, we require that $\alpha\oplus_a\beta =\alpha +\beta -\{a,\hat{a}\}$, since the synchronization of $a$
and $\hat{a}$ produces $\emptyset$. Activities are synchronized by their multiaction parts, i.e. the synchronization
on $a$ of two activities, whose multiaction parts $\alpha$ and $\beta$ possess the above properties, results in the
activity with the multiaction part $\alpha\oplus_a\beta$. We may synchronize activities of the same type only: either
both stochastic multiactions or both immediate ones, since immediate multiactions have a priority over stochastic ones,
hence, stochastic and immediate multiactions cannot be executed together (note also that the execution of immediate
multiactions takes no time, unlike that of stochastic ones). Synchronization
on $a$ means that, for a given expression with a process behaviour containing two concurrent activities that can be
synchronized
on $a$, there exists also the process behaviour that differs from the former only in that the two activities are
replaced by the result of their synchronization.

In the iteration, the initialization subprocess is executed first, then the body is performed zero or more times, and,
finally, the termination subprocess is executed.

Static expressions specify the structure of processes. As we shall see, the expressions correspond to unmarked LDTSIPNs
(LDTSIPNs are marked by definition).

\begin{definition}
Let $(\alpha ,\kappa )\in\mathcal{SIL}$ and $a\in Act$. A {\em static expression} of dtsiPBC is
$$E::=\ (\alpha ,\kappa )\mid E;E\mid E\cho E\mid E\| E\mid E[f]\mid E\rs a\mid E\sy a\mid [E*E*E].$$
\end{definition}

Let $StatExpr$ denote the set of {\em all static expressions} of dtsiPBC.

To avoid technical difficulties with the iteration operator, we should not allow any concurrency at the highest level
of the second argument of iteration. This is not a severe restriction, since we can always prefix parallel expressions
by an activity with the empty multiaction part. In \cite{Tar14}, we have demonstrated that relaxing the restriction can
result in nets which are not safe. Alternatively, we can use a different, safe, version of the iteration operator, but
its net translation has six arguments \cite{BDK01}.

\begin{definition}
Let $(\alpha ,\kappa )\in\mathcal{SIL}$ and $a\in Act$. A {\em regular static expression} of dtsiPBC is
$$\begin{array}{c}
E::=\ (\alpha ,\kappa )\mid E;E\mid E\cho E\mid E\| E\mid E[f]\mid E\rs a\mid E\sy a\mid [E*D*E],\\
\mbox{where }D::=\ (\alpha ,\kappa )\mid D;E\mid D\cho D\mid D[f]\mid D\rs a\mid D\sy a\mid [D*D*E].
\end{array}$$
\end{definition}

Let $RegStatExpr$ denote the set of {\em all regular static expressions} of dtsiPBC.

Dynamic expressions specify the states of processes. As we shall see, the expressions correspond to LDTSIPNs (marked by
default). Dynamic expressions are obtained from static ones, by annotating them with upper or lower bars which specify
the active components of the system at the current moment. The dynamic expression with upper bar (the overlined one)
$\overline{E}$ denotes the {\em initial}, and that with lower bar (the underlined one) $\underline{E}$ denotes the {\em
final} state of the process specified by a static expression $E$. The {\em underlying static expression} of a dynamic
one is obtained by removing all upper and lower bars from it.

\begin{definition}
Let $E\in StatExpr$ and $a\in Act$. A {\em dynamic expression} of dtsiPBC is
$$\begin{array}{c}
G::=\ \overline{E}\mid\underline{E}\mid G;E\mid E;G\mid G\cho E\mid E\cho G\mid G\| G\mid G[f]\mid G\rs a\mid G\sy
a\mid\\

[G*E*E]\mid [E*G*E]\mid [E*E*G].
\end{array}$$
\end{definition}

Let $DynExpr$ denote the set of {\em all dynamic expressions} of dtsiPBC.

If the underlying static expression of a dynamic one is not regular, the corresponding LDTSIPN can be non-safe (but it
is $2$-bounded in the worst case \cite{BDK01}).

\begin{definition}
A dynamic expression is {\em regular} if its underlying static one is so.
\end{definition}

Let $RegDynExpr$ denote the set of {\em all regular dynamic expressions} of dtsiPBC.

\section{Operational semantics}
\label{opersem.sec}

In this section, we define the operational semantics via labeled transition systems.

\subsection{Inaction rules}

The inaction rules for dynamic expressions describe their structural transformations in the form of
$G\Rightarrow\widetilde{G}$ which do not change the states of the specified processes. The goal of these syntactic
transformations is to obtain the well-structured resulting expressions called operative ones to which no inaction rules
can be further applied. As we shall see, the application of an inaction rule to a dynamic expression does not lead to
any discrete time tick or any transition firing in the corresponding LDTSIPN, hence, its current marking remains
unchanged. An application of every inaction rule does not need a discrete time delay, i.e. the dynamic expression
transformation described by the rule is accomplished instantly.

Table \ref{inactrulesim1.tab} defines inaction rules for regular dynamic expressions in the form of overlined and
underlined static ones, where $E,F,K\in RegStatExpr$ and $a\in Act$.

\begin{table}[h]
\caption{Inaction rules for overlined and underlined regular static expressions.}
\vspace{-1mm}
\label{inactrulesim1.tab}
\begin{center}
$\begin{array}{|lll|}
\hline
\rule{0mm}{4mm}
\overline{E;F}\Rightarrow\overline{E};F &
\underline{E};F\Rightarrow E;\overline{F} &
E;\underline{F}\Rightarrow\underline{E;F}\\[1mm]

\overline{E\cho F}\Rightarrow\overline{E}\cho F &
\overline{E\cho F}\Rightarrow E\cho\overline{F} &
\underline{E}\cho F\Rightarrow\underline{E\cho F}\\[1mm]

E\cho\underline{F}\Rightarrow\underline{E\cho F} &
\overline{E\| F}\Rightarrow\overline{E}\|\overline{F} &
\underline{E}\|\underline{F}\Rightarrow\underline{E\| F}\\[1mm]

\overline{E[f]}\Rightarrow\overline{E}[f] &
\underline{E}[f]\Rightarrow\underline{E[f]} &
\overline{E\rs a}\Rightarrow\overline{E}\rs a\\[1mm]

\underline{E}\rs a\Rightarrow\underline{E\rs a} &
\overline{E\sy a}\Rightarrow\overline{E}\sy a &
\underline{E}\sy a\Rightarrow\underline{E\sy a}\\[1mm]

\overline{[E*F*K]}\Rightarrow[\overline{E}*F*K]\hspace{1mm} &
[\underline{E}*F*K]\Rightarrow[E*\overline{F}*K]\hspace{1mm} &
[E*\underline{F}*K]\Rightarrow[E*\overline{F}*K]\\[1mm]

[E*\underline{F}*K]\Rightarrow[E*F*\overline{K}]\hspace{1mm} &
[E*F*\underline{K}]\Rightarrow\underline{[E*F*K]} & \\[1mm]
\hline
\end{array}$
\end{center}
\end{table}

Table \ref{inactrulesim2.tab} presents inaction rules for regular dynamic expressions in the arbitrary form, where
$E,F\in RegStatExpr,\ G,H,\widetilde{G},\widetilde{H}\in RegDynExpr$ and $a\in Act$.

\begin{table}[h]
\caption{Inaction rules for arbitrary regular dynamic expressions.}
\vspace{-1mm}
\label{inactrulesim2.tab}
\begin{center}
$\begin{array}{|lll|}
\hline
\rule{0mm}{7mm}
\dfrac{G\Rightarrow\widetilde{G},\ \circ\in\{;,\cho\}}{G\circ E\Rightarrow\widetilde{G}\circ E} &
\dfrac{G\Rightarrow\widetilde{G},\ \circ\in\{;,\cho\}}{E\circ G\Rightarrow E\circ\widetilde{G}} &
\dfrac{G\Rightarrow\widetilde{G}}{G\| H\Rightarrow\widetilde{G}\| H}\\[4mm]

\dfrac{H\Rightarrow\widetilde{H}}{G\| H\Rightarrow G\|\widetilde{H}} &
\dfrac{G\Rightarrow\widetilde{G}}{G[f]\Rightarrow\widetilde{G}[f]} &
\dfrac{G\Rightarrow\widetilde{G},\ \circ\in\{\!\!\rs\!\!,\!\!\sy\!\!\}}{G\circ a\Rightarrow\widetilde{G}\circ a}\\[4mm]

\dfrac{G\Rightarrow\widetilde{G}}{[G*E*F]\Rightarrow [\widetilde{G}*E*F]} &
\dfrac{G\Rightarrow\widetilde{G}}{[E*G*F]\Rightarrow [E*\widetilde{G}*F]} &
\dfrac{G\Rightarrow\widetilde{G}}{[E*F*G]\Rightarrow [E*F*\widetilde{G}]}\\[4mm]
\hline
\end{array}$
\end{center}
\end{table}

\begin{definition}
A regular dynamic expression $G$ is {\em operative} if no inaction rule can be applied to it.
\end{definition}

Let $OpRegDynExpr$ denote the set of {\em all operative regular dynamic expressions} of dtsiPBC. Note that any dynamic
expression can be always transformed into a (not necessarily unique) operative one by using the inaction rules. In the
following, we only consider regular expressions
and omit the word ``regular''.

\begin{definition}
The relation $\approx\ =(\Rightarrow\cup\Leftarrow )^*$ is a {\em structural equivalence} of dynamic expressions in
dtsiPBC. Thus, two dynamic expressions $G$ and $G'$ are {\em structurally equivalent}, denoted by $G\approx G'$, if
they can be reached from
one another by applying the inaction rules in a forward or backward direction.
\end{definition}

\subsection{Action and empty loop rules}

The action rules are applied when some activities are executed. With these rules we capture the prioritization of
immediate multiactions w.r.t. stochastic ones. We also have the empty loop rule which is used to capture a delay of one
discrete time unit in the same state when no immediate multiactions are executable. In this case, the empty multiset of
activities is executed. The action and empty loop rules will be used later to determine all multisets of activities
which can be executed from the structural equivalence class of every dynamic expression (i.e. from the state of the
corresponding process). This information together with that about probabilities or weights of the activities to be
executed from the current process state will be used to calculate the probabilities of such executions.

The action rules with stochastic (or immediate, otherwise) multiactions describe dynamic expression transformations in
the form of $G\stackrel{\Gamma}{\rightarrow}\widetilde{G}$ (or $G\stackrel{I}{\rightarrow}\widetilde{G}$) due to
execution of non-empty multisets $\Gamma$ of stochastic (or $I$ of immediate) multiactions. The rules represent
possible state changes of the specified processes when some non-empty multisets of stochastic (or immediate)
multiactions are executed. As we shall see, the application of an action rule with stochastic (or immediate)
multiactions to a dynamic expression leads in the corresponding LDTSIPN to a discrete time tick at which some
stochastic transitions fire (or to the instantaneous firing of some immediate transitions) and possible change of the
current marking. The current marking remains unchanged only if there is a self-loop produced by the iterative execution
of a non-empty multiset, which must be one-element, i.e. the single stochastic (or immediate) multiaction. The reason
is the regularity requirement that allows no concurrency at the highest level of the second argument of iteration.

The empty loop rule (applicable only when no immediate multiactions can be executed from the current state) describes
dynamic expression transformations in the form of $G\stackrel{\emptyset}{\rightarrow}G$ due to execution of the empty
multiset of activities at a discrete time tick. The rule reflects a non-zero probability to stay in the current state
at the next moment, which is a feature of discrete time stochastic processes. As we shall see, the application of the
empty loop rule to a dynamic expression leads to a discrete time tick in the corresponding LDTSIPN at which no
transitions fire and the current marking is not changed. This is a new rule with no prototype among inaction rules of
PBC, since it represents a time delay, but PBC has no notion of time. The PBC rule
$G\stackrel{\emptyset}{\rightarrow}G$ from \cite{BKo95,BDK01} in our setting would correspond to
a rule $G\Rightarrow G$ that describes staying in the current state when no time elapses. Since we do not need the
latter rule to transform dynamic expressions into operative ones and it can destroy the definition of operative
expressions, we do not have it.

Thus, an application of every action rule with stochastic multiactions or the empty loop rule requires one discrete
time unit delay, i.e. the execution of a (possibly empty) multiset of stochastic multiactions leading to the dynamic
expression transformation described by the rule is accomplished
after one time unit.
However, an application of every action rule with immediate multiactions does not take any time, i.e. the execution of
a (non-empty) multiset of immediate multiactions is accomplished instantly at the current
time.

Note that expressions of dtsiPBC can contain identical activities. To avoid technical difficulties, such as the proper
calculation of the state change probabilities for multiple transitions, we can always enumerate coinciding activities
from left to right in the syntax of expressions. The new activities resulted from synchronization will be annotated
with concatenation of numberings of the activities they come from, hence, the numbering should have a tree structure to
reflect the effect of multiple synchronizations. We now define the numbering which encodes a binary tree with the
leaves labeled by natural numbers.

\begin{definition}
The {\em numbering} of expressions is $\iota ::=\ n\mid (\iota )(\iota )$, where $n\in\naturals$.
\end{definition}

Let $Num$ denote the set of {\em all numberings} of expressions.

\begin{example}
The numbering $1$ encodes the binary tree in Figure \ref{bintrnum.fig}(a) with the root labeled by $1$. The numbering
$(1)(2)$ corresponds to the binary tree in Figure \ref{bintrnum.fig}(b) without internal nodes and with two leaves
labeled by $1$ and $2$. The numbering $(1)((2)(3))$ represents the binary tree in Figure \ref{bintrnum.fig}(c) with one
internal node, which is the root for the subtree $(2)(3)$, and three leaves labeled by $1,2$ and $3$.
\end{example}

\begin{figure}
\begin{center}
\includegraphics[scale=0.9]{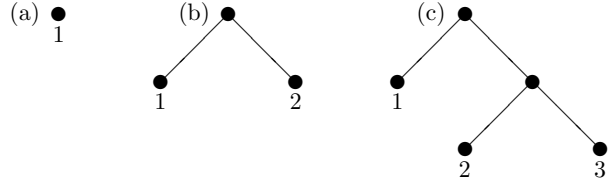}
\end{center}
\vspace{-6mm}
\caption{The binary trees encoded with the numberings $1,\ (1)(2)$ and $(1)((2)(3))$.}
\label{bintrnum.fig}
\end{figure}

The new activities resulting from synchronizations in different orders should be considered up to permutation of their
numbering. In this way, we shall recognize different instances of the same activity. If we compare the contents of
different numberings, i.e. the sets of natural numbers in them, we shall identify the mentioned instances. The {\em
content} of a numbering $\iota\in Num$ is\\
$Cont(\iota )=\left\{
\begin{array}{ll}
\{\iota\}, & \mbox{if }\iota\in\naturals;\\
Cont(\iota_1)\cup Cont(\iota_2), & \mbox{if }\iota =(\iota_1)(\iota_2).
\end{array}
\right.$

After the enumeration, the multisets of activities from the expressions will become the proper sets. Suppose that the
identical activities are enumerated when needed to avoid ambiguity. This enumeration is considered to be implicit.

Let $X$ be some set. We denote the Cartesian product $X\times X$ by $X^2$. Let $\mathcal{E}\subseteq X^2$ be an
equivalence relation on $X$. Then the {\em equivalence class} (w.r.t. $\mathcal{E}$) of an element $x\in X$ is defined
by $[x]_\mathcal{E}=\{y\in X\mid (x,y)\in\mathcal{E}\}$. The equivalence $\mathcal{E}$ partitions $X$ into the {\em set
of equivalence classes} $X/_\mathcal{E}=\{[x]_\mathcal{E}\mid x\in X\}$.

Let $G$ be a dynamic expression. Then $[G]_\approx =\{H\mid G\approx H\}$ is the equivalence class of $G$ w.r.t. the
structural equivalence. $G$ is an {\em initial} dynamic expression, denoted by $init(G)$, if $\exists E\in
RegStatExpr,\ G\in [\overline{E}]_\approx$. $G$ is a {\em final} dynamic expression, denoted by $final(G)$, if\\
$\exists E\in RegStatExpr,\ G\in [\underline{E}]_\approx$.

\begin{definition}
Let $G\in OpRegDynExpr$. We define the {\em set of all non-empty sets of activities which can be potentially executed
from $G$}, denoted by $Can(G)$. Let $(\alpha ,\kappa )\in\mathcal{SIL},\ E,F\in RegStatExpr,\ H\in OpRegDynExpr$ and
$a\in Act$.
\begin{enumerate}

\item If $final(G)$ then $Can(G)=\emptyset$.

\item If $G=\overline{(\alpha ,\kappa )}$ then $Can(G)=\{\{(\alpha ,\kappa )\}\}$.

\item If $\Upsilon\in Can(G)$ then $\Upsilon\in Can(G\circ E),\ \Upsilon\in Can(E\circ G)\ (\circ\in\{;,\cho\}),\
\Upsilon\in Can(G\| H),\\
\Upsilon\in Can(H\| G),\ f(\Upsilon )\in Can(G[f]),\ \Upsilon\in Can(G\rs a)\ (\mbox{when
}a,\hat{a}\not\in\mathcal{A}(\Upsilon )),\\
\Upsilon\in Can(G\sy a),\ \Upsilon\in Can([G*E*F]),\ \Upsilon\in Can([E*G*F]),\ \Upsilon\in Can([E*F*G])$.

\item If $\Upsilon\in Can(G)$ and $\Xi\in Can(H)$ then $\Upsilon +\Xi\in Can(G\| H)$.

\item If $\Upsilon\in Can(G\sy a)$ and $(\alpha ,\kappa ),(\beta ,\lambda )\in\Upsilon$ are different activities,
$a\in\alpha ,\ \hat{a}\in\beta$, then
\begin{enumerate}

\item $(\Upsilon +\{(\alpha\oplus_a\beta ,\kappa\cdot\lambda )\})\setminus\{(\alpha ,\kappa ),(\beta ,\lambda )\}\in
Can(G\sy a)$ if $\kappa ,\lambda\in (0;1)$;

\item $(\Upsilon +\{(\alpha\oplus_a\beta ,\natural_{l+m})\})\setminus\{(\alpha ,\kappa ),(\beta ,\lambda )\}\in
Can(G\sy a)$ if $\kappa =\natural_l,\ \lambda =\natural_m,\ l,m\in\reals_{>0}$.

When we synchronize the same set of activities in different orders, we obtain several activities with the same
multiaction and probability or weight parts, but with different numberings having the same content. Then, we
only consider a single one of the resulting activities.

For example, the synchronization of stochastic multiactions $(\alpha ,\rho )_1$ and $(\beta ,\chi )_2$ in
different orders generates the activities $(\alpha\oplus_a\beta ,\rho\cdot\chi )_{(1)(2)}$ and
$(\beta\oplus_a\alpha ,\chi\cdot\rho )_{(2)(1)}$. Similarly, the synchronization of immediate multiactions
$(\alpha ,\natural_l)_1$ and $(\beta ,\natural_m)_2$ in different orders generates the activities
$(\alpha\oplus_a\beta ,\natural_{l+m})_{(1)(2)}$ and $(\beta\oplus_a\alpha ,\natural_{m+l})_{(2)(1)}$. Since
$Cont((1)(2))=\{1,2\}=Cont((2)(1))$, in both cases, only the first activity (or the second one) resulting from
synchronization will appear in a set from $Can(G\sy a)$.

\end{enumerate}
\end{enumerate}
\end{definition}

Note that if $\Upsilon\in Can(G)$ then by definition of $Can(G)$, for all $\Xi\subseteq\Upsilon ,\ \Xi\neq\emptyset$,
we have $\Xi\in Can(G)$.

Let $G\in OpRegDynExpr$. Obviously, if there are only stochastic (or only immediate) multiactions in the sets from
$Can(G)$ then these stochastic (or immediate) multiactions can be executed from $G$. Otherwise, besides stochastic
ones, there are also immediate multiactions in the sets from $Can(G)$. By the note above, there are non-empty sets of
immediate multiactions in $Can(G)$ as well, i.e. $\exists\Upsilon\in Can(G),\ \Upsilon\in\naturals_{\rm
fin}^\mathcal{IL}\setminus\{\emptyset\}$. Then no stochastic multiactions can be executed from $G$, even if $Can(G)$
contains non-empty sets of stochastic multiactions, since immediate multiactions have a priority over stochastic ones.

\begin{definition}
Let $G\in OpRegDynExpr$. The {\em set of all non-empty sets of activities which can be execu\-ted from $G$} is
$Now(G)=\left\{
\begin{array}{ll}
Can(G), & \mbox{if }(Can(G)\subseteq\naturals_{\rm fin}^\mathcal{SL}\setminus\{\emptyset\})\vee
(Can(G)\subseteq\naturals_{\rm fin}^\mathcal{IL}\setminus\{\emptyset\});\\
Can(G)\cap\naturals_{\rm fin}^\mathcal{IL}, & \mbox{otherwise}.
\end{array}
\right.$

\end{definition}

An expression $G\in OpRegDynExpr$ is {\em tangible}, denoted by $tang(G)$, if $Now(G)\subseteq\naturals_{\rm
fin}^\mathcal{SL}\setminus\{\emptyset\}$. Otherwise, the expression $G$ is {\em vanishing}, denoted by $vanish(G)$, and
in this case $Now(G)\subseteq\naturals_{\rm fin}^\mathcal{IL}\setminus\{\emptyset\}$.

\begin{example}
Let $G=(\overline{(\{a\},\natural_1)}\cho (\{b\},\natural_2))\|\overline{(\{c\},\frac{1}{2})}$ and
$G'=((\{a\},\natural_1)\cho\overline{(\{b\},\natural_2)})\|\overline{(\{c\},\frac{1}{2})}$. Then $G\approx G'$, since
$G\Leftarrow G''\Rightarrow G'$ for $G''=(\overline{(\{a\},\natural_1)\cho(\{b\},\natural_2)}\|
\overline{(\{c\},\frac{1}{2})}$, but $Can(G)=\{\{(\{a\},\natural_1)\},\\
\{(\{c\},\frac{1}{2})\},\ \{(\{a\},\natural_1),(\{c\},\frac{1}{2})\}\},\ Can(G')=\{\{(\{b\},\natural_2)\},
\{(\{c\},\frac{1}{2})\},\{(\{b\},\natural_2),(\{c\},\frac{1}{2})\}\}$ and\\
$Now(G)=\{\{(\{a\},\natural_1)\}\}, Now(G')=\{\{(\{b\},\natural_2)\}\}$. Clearly, we have $vanish(G)$ and $vanish(G')$.
The executions like that of $\{(\{c\},\frac{1}{2})\}$ (and all sets including it) from $G$ and $G'$ must be disabled
using preconditions in the action rules, since immediate multiactions have a priority over stochastic ones, hence, the
former are always executed first.

Let $H=\overline{(\{a\},\natural_1)}\cho (\{b\},\frac{1}{2})$ and
$H'=(\{a\},\natural_1)\cho\overline{(\{b\},\frac{1}{2})}$. Then $H\approx H'$, since $H\Leftarrow H''\Rightarrow H'$
for $H''=\overline{(\{a\},\natural_1)\cho (\{b\},\frac{1}{2})}$, but $Can(H)=Now(H)=\{\{(\{a\},\natural_1)\}\}$ and
$Can(H')=Now(H')=\{\{(\{b\},\frac{1}{2})\}\}$. We have $vanish(H)$, but $tang(H')$. To get the action rules correct
under structural equivalence, the executions like that of $\{(\{b\},\frac{1}{2})\}$ from $H'$ must be disabled using
preconditions in the action rules, since immediate multiactions have a priority over stochastic ones, hence, the
choices are always resolved in favour of the former.
\label{cannow.exm}
\end{example}

In Table \ref{actrulesim.tab}, we define the action and empty loop rules. In this table, $(\alpha ,\rho ),(\beta ,\chi
)\in\mathcal{SL},\ (\alpha ,\natural_l),\\
(\beta ,\natural_m)\in\mathcal{IL}$ and $(\alpha ,\kappa )\in\mathcal{SIL}$. Further, $E,F\in RegStatExpr,\ G,H\in
OpRegDynExpr,\ \widetilde{G},\widetilde{H}\in RegDynExpr$ and $a\in Act$. Moreover, $\Gamma ,\Delta\in\naturals_{\rm
fin}^\mathcal{SL}\setminus\{\emptyset\},\ \Gamma '\in\naturals_{\rm fin}^\mathcal{SL},\ I,J\in\naturals_{\rm
fin}^\mathcal{IL}\setminus\{\emptyset\},\ I'\in\naturals_{\rm fin}^\mathcal{IL}$ and $\Upsilon\in\naturals_{\rm
fin}^\mathcal{SIL}\setminus\{\emptyset\}$. The first rule is the empty loop rule {\bf El}. The other rules are the
action rules, describing transformations of dynamic expressions, which are built using particular algebraic operations.
If we cannot merge a rule with stochastic multiactions and a rule with immediate multiactions for some operation then
we get the coupled action rules. Then the names of the action rules with immediate multiactions have a suffix ``${\bf
i}$''.

\begin{table}
\caption{Action and empty loop rules.}
\vspace{-1mm}
\label{actrulesim.tab}
\begin{center}
$\begin{array}{|ll|}
\hline
\multicolumn{2}{|l|}{\rule{0mm}{7.5mm}{\bf El}\ \dfrac{tang(G)}{G\stackrel{\emptyset}{\rightarrow}G}\hspace{9mm}
{\bf B}\ \overline{(\alpha ,\kappa )}\stackrel{\{(\alpha ,\kappa )\}}
{\longrightarrow}\underline{(\alpha ,\kappa )}\hspace{9mm}
{\bf S}\ \dfrac{G\stackrel{\Upsilon}{\rightarrow}\widetilde{G}}{G;E\stackrel{\Upsilon}{\rightarrow}\widetilde{G};E,\
E;G\stackrel{\Upsilon}{\rightarrow}E;\widetilde{G}}}\\[5mm]

{\bf C}\ \dfrac{G\stackrel{\Gamma}{\rightarrow}\widetilde{G},\ \neg init(G)\vee (init(G)\wedge tang(\overline{E}))}
{G\cho E\stackrel{\Gamma}{\rightarrow}\widetilde{G}\cho E,\ E\cho G\stackrel{\Gamma}{\rightarrow}E\cho\widetilde{G}} &
{\bf Ci}\ \dfrac{G\stackrel{I}{\rightarrow}\widetilde{G}}{G\cho E\stackrel{I}{\rightarrow}\widetilde{G}\cho E,\ E\cho
G\stackrel{I}{\rightarrow}E\cho\widetilde{G}}\\[5mm]

{\bf P1}\ \dfrac{G\stackrel{\Gamma}{\rightarrow}\widetilde{G},\ tang(H)}{G\|
H\stackrel{\Gamma}{\rightarrow}\widetilde{G}\| H,\ H\| G\stackrel{\Gamma}{\rightarrow}H\|\widetilde{G}} &
{\bf P1i}\ \dfrac{G\stackrel{I}{\rightarrow}\widetilde{G}}{G\| H\stackrel{I}{\rightarrow}\widetilde{G}\| H,\
H\| G\stackrel{I}{\rightarrow}H\|\widetilde{G}}\\[5mm]

{\bf P2}\ \dfrac{G\stackrel{\Gamma}{\rightarrow}\widetilde{G},\ H\stackrel{\Delta}{\rightarrow}\widetilde{H}}
{G\| H\stackrel{\Gamma +\Delta}{\longrightarrow}\widetilde{G}\|\widetilde{H}} &
{\bf P2i}\ \dfrac{G\stackrel{I}{\rightarrow}\widetilde{G},\ H\stackrel{J}{\rightarrow}\widetilde{H}}
{G\| H\stackrel{I+J}{\longrightarrow}\widetilde{G}\|\widetilde{H}}\\[5mm]

{\bf L}\ \dfrac{G\stackrel{\Upsilon}{\rightarrow}\widetilde{G}}{G[f]\stackrel{f(\Upsilon
)}{\longrightarrow}\widetilde{G}[f]} &
{\bf Rs}\ \dfrac{G\stackrel{\Upsilon}{\rightarrow}\widetilde{G},\ a,\hat{a}\not\in\mathcal{A}(\Upsilon )}{G\rs
a\stackrel{\Upsilon}{\rightarrow}\widetilde{G}\rs a}\\[5mm]

\multicolumn{2}{|l|}{{\bf I1}\ \dfrac{G\stackrel{\Upsilon}{\rightarrow}\widetilde{G}}
{[G*E*F]\stackrel{\Upsilon}{\rightarrow}[\widetilde{G}*E*F]}\hspace{9mm}
{\bf I2}\ \dfrac{G\stackrel{\Gamma}{\rightarrow}\widetilde{G},\ \neg init(G)\vee (init(G)\wedge
tang(\overline{F}))}{[E*G*F]\stackrel{\Gamma}{\rightarrow}[E*\widetilde{G}*F]}}\\[5mm]

\multicolumn{2}{|l|}{{\bf I2i}\ \dfrac{G\stackrel{I}{\rightarrow}\widetilde{G}}
{[E*G*F]\stackrel{I}{\rightarrow}[E*\widetilde{G}*F]}\hspace{8mm}
{\bf I3}\ \dfrac{G\stackrel{\Gamma}{\rightarrow}\widetilde{G},\ \neg init(G)\vee (init(G)\wedge tang(\overline{F}))}
{[E*F*G]\stackrel{\Gamma}{\rightarrow}[E*F*\widetilde{G}]}}\\[5mm]

{\bf I3i}\ \dfrac{G\stackrel{I}{\rightarrow}\widetilde{G}}{[E*F*G]\stackrel{I}{\rightarrow}[E*F*\widetilde{G}]} &
{\bf Sy1}\ \dfrac{G\stackrel{\Upsilon}{\rightarrow}\widetilde{G}}
{G\sy a\stackrel{\Upsilon}{\rightarrow}\widetilde{G}\sy a}\\[5mm]

\multicolumn{2}{|l|}{{\bf Sy2}\ \dfrac{G\sy a\xrightarrow{\Gamma '+\{(\alpha ,\rho )\}+\{(\beta ,\chi )\}}
\widetilde{G}\sy a,\ a\in\alpha ,\ \hat{a}\in\beta}
{G\sy a\xrightarrow{\Gamma '+\{(\alpha\oplus_a\beta,\rho\cdot\chi )\}}\widetilde{G}\sy a}}\\[5mm]

\multicolumn{2}{|l|}{{\bf Sy2i}\ \dfrac{G\sy a\xrightarrow{I'+\{(\alpha ,\natural_l)\}+
\{(\beta ,\natural_m)\}}\widetilde{G}\sy a,\ a\in\alpha ,\ \hat{a}\in\beta}{G\sy
a\xrightarrow{I'+\{(\alpha\oplus_a\beta ,\natural_{l+m})\}}\widetilde{G}\sy a}}\\[5mm]
\hline
\end{array}$
\end{center}
\end{table}

Almost all the rules in Table \ref{actrulesim.tab} (excepting {\bf El}, {\bf P2}, {\bf P2i}, {\bf Sy2} and {\bf Sy2i})
resemble those of gsPBC \cite{MVCR08}, but the former correspond to execution of sets of activities, not of single
activities, as in the latter, and our rules have simpler preconditions (if any), since all immediate multiactions in
dtsiPBC have the same priority level, unlike those of gsPBC. The preconditions in rules {\bf El}, {\bf C}, {\bf P1},
{\bf I2} and {\bf I3} are needed to ensure that (possibly empty) sets of stochastic multiactions are executed only from
{\em tangible} operative dynamic expressions, such that all operative dynamic expressions structurally equivalent to
them are tangible as well. For example, if $init(G)$ in rule {\bf C} then $G=\overline{F}$ for some static expression
$F$ and $G\cho E=\overline{F}\cho E\approx F\cho\overline{E}$. Hence, it should be guaranteed that
$tang(F\cho\overline{E})$, which holds iff $tang(\overline{E})$. The case $E\cho G$ is treated similarly. Further, in
rule {\bf P1}, assuming that $tang(G)$, it should be guaranteed that $tang(G\| H)$ and $tang(H\| G)$, which holds iff
$tang(H)$. The preconditions in rules {\bf I2} and {\bf I3} are analogous to that in rule {\bf C}.

Rule {\bf El} corresponds to one discrete time unit delay while executing no activities and therefore it has no
analogues among the rules of gsPBC that adopts the continuous time model. Rules {\bf P2} and {\bf P2i} have no similar
rules in gsPBC, since interleaving semantics of the algebra allows no simultaneous execution of activities. {\bf P2}
and {\bf P2i} have in PBC the analogous rule {\bf PAR} that is used to construct step semantics of the calculus, but
the former two rules correspond to execution of sets of activities, unlike that of multisets of multiactions in the
latter rule. Rules {\bf Sy2} and {\bf Sy2i} differ from the corresponding synchronization rules in gsPBC, since the
probability or the weight of synchronization in the former rules and the rate or the weight of synchronization in the
latter rules are calculated in two distinct ways.

Rule {\bf Sy2} establishes that the synchronization of two stochastic multiactions is made by taking the product of
their probabilities, since we are considering that both must occur for the synchronization to happen, so this
corresponds, in some sense, to the probability of the independent event intersection, but the real situation is more
complex, since these stochastic multiactions can be also executed in parallel. Nevertheless, when scoping (the combined
operation consisting of synchronization followed by restriction over the same action \cite{BDK01}) is applied over a
parallel execution, we get as final result just the simple product of the probabilities, since no normalization is
needed there. Multiplication is an associative and commutative binary operation that is distributive over addition,
i.e. it fulfills all practical conditions imposed on the synchronization operator in \cite{Hil94}. Further, if both
arguments of multiplication are from $(0;1)$ then the result belongs to the same interval, hence, multiplication
naturally maintains probabilistic compositionality in our model. Our approach is similar to the multiplication of rates
of the synchronized actions in MTIPP \cite{HR94} in the case when the rates are less than $1$. Moreover, for the
probabilities $\rho$ and $\chi$ of two stochastic multiactions to be synchronized we have $\rho\cdot\chi <\min\{\rho
,\chi\}$, i.e. multiplication meets the performance requirement stating that the probability of the resulting
synchronized stochastic multiaction should be less than the probabilities of the two ones to be synchronized.
In terms of performance evaluation, it is usually supposed that the execution of two components together require more
system resources and time than the execution of each single one. This resembles the {\em bounded capacity} assumption
from \cite{Hil94}. Thus, multiplication is easy to handle with and it satisfies the algebraic, probabilistic, time and
performance requirements. Therefore, we have chosen the product of the probabilities for the synchronization. See also
\cite{BKLL95,BrHe01} for a discussion about binary operations producing the rates of synchronization in the continuous
time setting.

In rule {\bf Sy2i}, we sum the weights of two synchronized immediate multiactions, since the weights can be interpreted
as the rewards \cite{Ros96}, which we collect. Next, we express that the synchronized execution of immediate
multiactions has more importance than that of every single one. The weights of immediate multiactions can be also seen
as bonus rewards associated with transitions \cite{BBr01}. The rewards are summed during synchronized execution of
immediate multiactions, since in this case all the synchronized activities can be seen as
participated in the execution. We prefer to collect more rewards, thus, the transitions providing greater rewards will
have a preference and they will be executed with a greater probability. Since execution of immediate multiactions takes
no time, we prefer to execute in a step as many synchronized immediate multiactions as possible to get more progress in
behaviour. Under behavioural progress we mean an advance in executing activities, which does not always imply a
progress in time, as when the activities are immediate multiactions. This aspect will be used later, while evaluating
performance via the embedded discrete time Markov chains (EDTMCs) of expressions. Since every state change in EDTMC
takes one unit of (local) time, greater advance in operation of the EDTMC allows one to calculate quicker performance
indices.

We do not have a self-synchronization, i.e. a synchronization of an activity with itself, since all the (enumerated)
activities executed together are considered to be different. This permits to avoid unexpected behaviour and technical
difficulties \cite{BDK01}.

In Table \ref{rulescompim.tab}, inaction rules, action rules (with stochastic or immediate multiactions) and empty loop
rule are compared according to the three aspects of their application: whether it changes the current state, whether it
leads to a time progress, and whether it results in execution of some activities. Positive answers to the questions are
denoted by the plus sign while negative ones are specified by the minus sign. If both positive and negative answers can
be given to some of the questions in different cases then the plus-minus sign is written. The process states are
considered up to structural equivalence of the corresponding expressions, and time progress is not regarded as a state
change.

\begin{table}
\caption{Comparison of inaction, action and empty loop rules.}
\vspace{-1mm}
\label{rulescompim.tab}
\begin{center}
\small\begin{tabular}{|c||c|c|c|}
\hline
Rules & State change & Time progress & Activities execution\\
\hline
Inaction rules & $-$ & $-$ & $-$\\
Action rules with stochastic multiactions & $\pm$ & $+$ & $+$\\
Action rules with immediate multiactions & $\pm$ & $-$ & $+$\\
Empty loop rule & $-$ & $+$ & $-$\\
\hline
\end{tabular}
\end{center}
\end{table}

\subsection{Transition systems}

We now construct labeled probabilistic transition systems associated with dynamic expressions to define their
operational semantics.

\begin{definition}
The {\em derivation set} of a dynamic expression $G$, denoted by $DR(G)$, is the minimal set with
\begin{itemize}

\item $[G]_\approx\in DR(G)$;

\item if $[H]_\approx\in DR(G)$ and $\exists\Upsilon ,\ H\stackrel{\Upsilon}{\rightarrow}\widetilde{H}$ then
$[\widetilde{H}]_\approx\in DR(G)$.

\end{itemize}
\end{definition}

Let $G$ be a dynamic expression and $s,\tilde{s}\in DR(G)$.

The set of {\em all sets of activities executable in} $s$ is defined as $Exec(s)=\{\Upsilon\mid\exists H\in s,\
\exists\widetilde{H},\ H\stackrel{\Upsilon}{\rightarrow}\widetilde{H}\}$. It can be proved by induction on the
structure of expressions that $\Upsilon\in Exec(s)\setminus\{\emptyset\}$ implies $\exists H\in s,\ \Upsilon\in
Now(H)$. The reverse statement does not hold in general, as the next example shows.

\begin{example}
Let $H,H'$ be from Example \ref{cannow.exm} and $s\!=\![H]_\approx\!=\![H']_\approx$. We have
$Now(H)=\{\{(\{a\},\natural_1)\}\}$ and $Now(H')=\{\{(\{b\},\frac{1}{2})\}\}$. Since only rules {\bf Ci} and {\bf B}
can be applied to $H$,
and no action rule can be applied to $H'$, we get $Exec(s)=\{\{(\{a\},\natural_1)\}\}$. Then, for $H'\in s$ and
$\Upsilon =\{(\{b\},\frac{1}{2})\}\in Now(H')$, we get $\Upsilon\not\in Exec(s)$.
\end{example}

The state $s$ is {\em tangible} if $Exec(s)\subseteq\naturals_{\rm fin}^\mathcal{SL}$. For tangible states we may have
$Exec(s)=\{\emptyset\}$. Otherwise, the state $s$ is {\em vanishing}, and in this case $Exec(s)\subseteq\naturals_{\rm
fin}^\mathcal{IL}\setminus\{\emptyset\}$. The set of {\em all tangible states from $DR(G)$} is denoted by $DR_{\rm
T}(G)$, and the set of {\em all vanishing states from $DR(G)$} is denoted by $DR_{\rm V}(G)$. Clearly, $DR(G)=DR_{\rm
T}(G)\uplus DR_{\rm V}(G)$ ($\uplus$ denotes disjoint union).

Note that if $\Upsilon\in Exec(s)$ then by rules {\bf P2}, {\bf P2i}, {\bf Sy2}, {\bf Sy2i} and definition of
$Exec(s)$, for all $\Xi\subseteq\Upsilon ,\ \Xi\neq\emptyset$, we have $\Xi\in Exec(s)$.

Since the inaction rules only distribute and move upper and lower bars along the syntax of dynamic expressions, all
$H\in s$ have the same underlying static expression $F$. The action rules {\bf Sy2} and {\bf Sy2i} are the only ones
that generate new activities. Since we have a finite number of operators $\sy$ in $F$ and all the multiaction parts of
the activities are finite multisets, the number of the new synchronized activities is also finite. The action rules
contribute to $Exec(s)$ (in addition to the empty set, if rule {\bf El} is applicable) only the sets consisting both of
activities from $F$ and the new activities, produced by {\bf Sy2} and {\bf Sy2i}. Since we have a finite number of such
activities, the set $Exec(s)$ is finite, hence, summation and multiplication by its elements are well-defined. Similar
reasoning can be used to demonstrate that for all dynamic expressions $H$ (not just for those from $s$), $Now(H)$ is a
finite set.

Let $\Upsilon\in Exec(s)\setminus\{\emptyset\}$. The {\em probability that the set of stochastic multiactions
$\Upsilon$ is ready for execution in $s$} or the {\em weight of the set of immediate multiactions $\Upsilon$ which is
ready for execution in $s$} is
$$PF(\Upsilon ,s)=
\left\{
\begin{array}{ll}
\displaystyle\prod_{(\alpha ,\rho )\in\Upsilon}\rho\cdot\prod_{\{\{(\beta ,\chi )\}\in Exec(s)\mid (\beta ,\chi
)\not\in\Upsilon\}}(1-\chi ), & \mbox{if }s\in DR_{\rm T}(G);\\
\displaystyle\sum_{(\alpha ,\natural_l)\in\Upsilon}l, & \mbox{if }s\in DR_{\rm V}(G).
\end{array}
\right.$$

In the case $\Upsilon =\emptyset$ and $s\in DR_{\rm T}(G)$ we define
$$PF(\emptyset ,s)=
\left\{
\begin{array}{ll}
\displaystyle\prod_{\{(\beta ,\chi )\}\in Exec(s)}(1-\chi ), & \mbox{if }Exec(s)\neq\{\emptyset\};\\
1, & \mbox{if }Exec(s)=\{\emptyset\}.
\end{array}
\right.$$

If $s\in DR_{\rm T}(G)$ and $Exec(s)\neq\{\emptyset\}$ then $PF(\Upsilon ,s)$ can be interpreted as a {\em joint}
probability of independent events (in a probability sense, i.e. the probability of intersection of these events is
equal to the product of their probabilities). Each such an event consists in the positive or negative decision to be
executed of a particular stochastic multiaction. Every executable stochastic multiaction decides probabilistically
(using its probabilistic part) and independently (from others), if it wants to be executed in $s$. If $\Upsilon$ is a
set of all executable stochastic multiactions which have decided to be executed in $s$ and $\Upsilon\in Exec(s)$ then
$\Upsilon$ is ready for execution in $s$. The multiplication in the definition is used because it reflects the
probability of the independent event intersection. Alternatively, when $\Upsilon\neq\emptyset ,\ PF(\Upsilon ,s)$ can
be interpreted as the probability to execute {\em exclusively} the set of stochastic multiactions $\Upsilon$ in $s$,
i.e. the probability of {\em intersection} of two events calculated using the conditional probability formula in the
form
${\sf P}(X\cap Y)={\sf P}(X|Y){\sf P}(Y)=\prod_{(\alpha ,\rho )\in\Upsilon}\rho\cdot\prod_{\{\{(\beta ,\chi )\}\in
Exec(s)\mid (\beta ,\chi )\not\in\Upsilon\}}(1-\chi )$, as shown in \cite{TMV14}. When $\Upsilon =\emptyset ,\
PF(\Upsilon ,s)$ can be interpreted as the probability not to execute in $s$ any executable stochastic multiactions,
thus, $PF(\emptyset ,s)=\prod_{\{(\beta ,\chi )\}\in Exec(s)}(1-\chi )$. When only the empty set of activities can be
executed in $s$, i.e. $Exec(s)=\{\emptyset\}$, we take $PF(\emptyset ,s)=1$, since then we stay in $s$. For $s\in
DR_{\rm T}(G)$ we have $PF(\emptyset ,s)\in (0;1]$, hence, we can stay in $s$ at the next time moment with a certain
positive probability.

If $s\in DR_{\rm V}(G)$ then $PF(\Upsilon ,s)$ can be interpreted as the {\em overall (cumulative)} weight of the
immediate multiactions from $\Upsilon$, i.e. the sum of all their weights. The summation here is used since the weights
can be seen as the rewards which are collected \cite{Ros96}. In addition, this means that concurrent execution of the
immediate multiactions has more importance than that of every single one.
Thus, this reasoning is the same as that used to define the weight of synchronized immediate multiactions in the rule
{\bf Sy2i}.

Note that the definition of $PF(\Upsilon ,s)$ (as well as our definitions of other probability functions) is based on
the enumeration of activities which is considered implicit.

Let $\Upsilon\in Exec(s)$. Besides $\Upsilon$, some other sets of activities may be ready for execution in $s$, hence,
a kind of conditioning or normalization is needed to calculate the execution probability. The {\em probability to
execute the set of activities $\Upsilon$ in $s$} is
$$PT(\Upsilon ,s)=\frac{PF(\Upsilon ,s)}{\displaystyle\sum_{\Xi\in Exec(s)}PF(\Xi ,s)}.$$

If $s\in DR_{\rm T}(G)$ then $PT(\Upsilon ,s)$ can be interpreted as the {\em conditional} probability to execute
$\Upsilon$ in $s$ calculated using the conditional probability formula in the form ${\sf P}(Z|W)=\frac{{\sf P}(Z\cap
W)}{{\sf P}(W)}
=\frac{PF(\Upsilon ,s)}{\sum_{\Xi\in Exec(s)}PF(\Xi ,s)}$, as shown in \cite{TMV14}. Note that $PF(\Upsilon ,s)$ can be
seen as the {\em potential} probability to execute $\Upsilon$ in $s$, since we have $PF(\Upsilon ,s)=PT(\Upsilon ,s)$
only when {\em all} sets (including the empty one) consisting of the executable stochastic multiactions can be executed
in $s$. In this case, all the mentioned stochastic multiactions can be executed in parallel in $s$ and we have
$\sum_{\Xi\in Exec(s)}PF(\Xi ,s)=1$, since this sum collects the products of {\em all} combinations of the probability
parts of the stochastic multiactions and the negations of these parts. But in general, for example, for two stochastic
multiactions $(\alpha ,\rho )$ and $(\beta ,\chi )$ executable in $s$, it may happen that they cannot be executed in
$s$ in parallel, i.e. $\emptyset ,\{(\alpha ,\rho )\},\{(\beta ,\chi )\}\in Exec(s)$, but $\{(\alpha ,\rho ),(\beta
,\chi )\}\not\in Exec(s)$. For $s\in DR_{\rm T}(G)$ we have $PT(\emptyset ,s)\in (0;1]$, hence, there is a non-zero
probability to stay in the state $s$ at the next moment, and the residence time in $s$ is at least
one
time unit.

If $s\in DR_{\rm V}(G)$ then $PT(\Upsilon ,s)$ can be interpreted as the weight of the set of immediate multiactions
$\Upsilon$ which is ready for execution in $s$ {\em normalized} by the weights of {\em all} the sets executable in $s$.
This approach is analogous to that used in the EMPA definition of the probabilities of immediate actions executable from
the same process state \cite{BGo98} (inspired by way in which the probabilities of conflicting immediate transitions in
GSPNs are calculated \cite{Bal07}). The only difference is that we have a step semantics and, for every set of
immediate multiactions executed in parallel, we use its cumulative weight.

Note that the sum of outgoing probabilities for the expressions belonging to the derivations of $G$ is equal to $1$.
More formally, $\forall s\in DR(G),\ \sum_{\Upsilon\in Exec(s)}PT(\Upsilon ,s)=1$. This, obviously, follows from the
definition of $PT(\Upsilon ,s)$, and guarantees that it always defines a probability distribution.

The {\em probability to move from $s$ to $\tilde{s}$ by executing any set of activities} is
$$PM(s,\tilde{s})=\sum_{\{\Upsilon\mid\exists H\in s,\ \exists\widetilde{H}\in\tilde{s},\
H\stackrel{\Upsilon}{\rightarrow}\widetilde{H}\}}PT(\Upsilon ,s).$$
The summation above reflects the probability of the mutually exclusive event union, since\\
$\sum_{\{\Upsilon\mid\exists H\in s,\ \exists\widetilde{H}\in\tilde{s},\
H\stackrel{\Upsilon}{\rightarrow}\widetilde{H}\}}PT(\Upsilon ,s)=\frac{1}{\sum_{\Xi\in Exec(s)}PF(\Xi ,s)}\cdot
\sum_{\{\Upsilon\mid\exists H\in s,\ \exists\widetilde{H}\in\tilde{s},\
H\stackrel{\Upsilon}{\rightarrow}\widetilde{H}\}}PF(\Upsilon ,s)$, where\\
for each $\Upsilon ,\ PF(\Upsilon ,s)$ is the probability of the exclusive execution of $\Upsilon$ in $s$.

Note that $\forall s\in DR(G),\ \sum_{\{\tilde{s}\mid\exists H\in s,\ \exists\widetilde{H}\in\tilde{s},\
\exists\Upsilon ,\ H\stackrel{\Upsilon}{\rightarrow}\widetilde{H}\}}PM(s,\tilde{s})=\\
\sum_{\{\tilde{s}\mid\exists H\in s,\ \exists\widetilde{H}\in\tilde{s},\ \exists\Upsilon ,\
H\stackrel{\Upsilon}{\rightarrow}\widetilde{H}\}} \sum_{\{\Upsilon\mid\exists H\in s,\
\exists\widetilde{H}\in\tilde{s},\ H\stackrel{\Upsilon}{\rightarrow}\widetilde{H}\}}PT(\Upsilon ,s)=
\sum_{\Upsilon\in Exec(s)}PT(\Upsilon ,s)=1$.

\begin{example}
Let $E=(\{a\},\rho )\cho (\{a\},\chi )$, where $\rho ,\chi\in (0;1)$. $DR(\overline{E})$ consists of the equivalence
classes $s_1=[\overline{E}]_\approx$ and $s_2=[\underline{E}]_\approx$. We have $DR_{\rm T}(\overline{E})=\{s_1,s_2\}$.
The execution probabilities are calcu\-lated as follows. Since $Exec(s_1)=\{\emptyset ,\{(\{a\},\rho
)\},\{(\{a\},\chi )\}\}$, we get $PF(\{(\{a\},\rho )\},s_1)=\rho (1-\chi ),\\
PF(\{(\{a\},\chi )\},s_1)=\chi (1-\rho )$ and $PF(\emptyset ,s_1)=(1-\rho )(1-\chi )$. Then $\sum_{\Xi\in
Exec(s_1)}PF(\Xi ,s_1)=\rho (1-\chi )+\chi (1-\rho )+(1-\rho )(1-\chi )=1-\rho\chi$. Thus, $PT(\{(\{a\},\rho
)\},s_1)=\frac{\rho (1-\chi )}{1-\rho\chi},\ PT(\{(\{a\},\chi )\},s_1)=\frac{\chi (1-\rho )}{1-\rho\chi}$ and
$PT(\emptyset ,s_1)=PM(s_1,s_1)=\frac{(1-\rho )(1-\chi )}{1-\rho\chi}$. Next, $Exec(s_2)=\{\emptyset\}$, hence,\\
$\sum_{\Xi\in Exec(s_2)}PF(\Xi ,s_2)=PF(\emptyset ,s_2)=1$ and $PT(\emptyset ,s_2)=PM(s_2,s_2)=\frac{1}{1}=1$.\\
Finally, $PM(s_1,s_2)=PT(\{(\{a\},\rho )\},s_1)+PT(\{(\{a\},\chi )\},s_1)=\frac{\rho (1-\chi )}{1-\rho\chi}+\frac{\chi
(1-\rho )}{1-\rho\chi}=\frac{\rho +\chi -2\rho\chi}{1-\rho\chi}$.

Let $E'=(\{a\},\natural_l)\cho (\{a\},\natural_m)$, where $l,m\in\reals_{>0}$. $DR(\overline{E'})$ consists of
the equivalence classes\\
$s_1'=[\overline{E'}]_\approx$ and $s_2'=[\underline{E'}]_\approx$. We have $DR_{\rm T}(\overline{E'})=\{s_2'\}$ and
$DR_{\rm V}(\overline{E'})=\{s_1'\}$. The execution probabilities are calculated as follows. Since
$Exec(s_1')=\{\{(\{a\},\natural_l)\},\{(\{a\},\natural_m)\}\}$, we get $PF(\{(\{a\},\natural_l)\},s_1')=l$ and
$PF(\{(\{a\},\natural_m)\},s_1')=m$. Then $\sum_{\Xi\in Exec(s_1')}PF(\Xi ,s_1')=l+m$. Thus,
$PT(\{(\{a\},\natural_l)\},s_1')=\frac{l}{l+m}$ and $PT(\{(\{a\},\natural_m)\},s_1')= \frac{m}{l+m}$. Next,
$Exec(s_2')=\{\emptyset\}$, hence, $\sum_{\Xi\in Exec(s_2')}PF(\Xi ,s_2')=PF(\emptyset ,s_2')=1$ and $PT(\emptyset
,s_2')=PM(s_2',s_2')= \frac{1}{1}=1$.\\
Finally, $PM(s_1',s_2')=PT(\{(\{a\},\natural_l)\},s_1')+PT(\{(\{a\},\natural_m)\},s_1')=\frac{l}{l+m}+\frac{m}{l+m}=1$.
\label{trprob.exm}
\end{example}

\begin{definition}
Let $G$ be a dynamic expression. The {\em (labeled probabilistic) transition system} of $G$ is a quadruple
$TS(G)=(S_G,L_G,\mathcal{T}_G,s_G)$, where
\begin{itemize}

\item the set of {\em states} is $S_G=DR(G)$;

\item the set of {\em labels} is $L_G=2^\mathcal{SIL}\times (0;1]$;

\item the set of {\em transitions} is $\mathcal{T}_G\!=\!\{(s,(\Upsilon ,PT(\Upsilon ,s)),\tilde{s})\!\mid\!
s,\tilde{s}\in DR(G),\ \exists H\in s,\
\exists\widetilde{H}\in\tilde{s},\ H\stackrel{\Upsilon}{\rightarrow}\widetilde{H}\}$;

\item the {\em initial state} is $s_G=[G]_\approx$.

\end{itemize}
\end{definition}

The definition of $TS(G)$ is correct, i.e. for every state, the sum of the probabilities of all the transitions
starting from it is $1$. This is guaranteed by the note after the definition of $PT(\Upsilon ,s)$. Thus, we have
defined a {\em generative} model of probabilistic processes \cite{GSS95}. The reason is that the sum of the
probabilities of the transitions with all possible labels should be equal to $1$, not only of those with the same
labels (up to enumeration of activities they include) as in the {\em reactive} models, and we do not have a nested
probabilistic choice as in the {\em stratified} models.

The transition system $TS(G)$ associated with a dynamic expression $G$ describes all the steps (concurrent executions)
that occur at discrete time moments with some (one-step) probability and consist of sets of activities. Every step
consisting of stochastic multiactions or the empty step (i.e. that consisting of the empty set of activities) occurs
instantly after one discrete time unit delay. Each step consisting of immediate multiactions occurs instantly without
any delay. The step can change the current state. The states are the structural equivalence classes of dynamic
expressions obtained by application of action rules starting from the expressions belonging to $[G]_\approx$. A
transition $(s,(\Upsilon ,\mathcal{P}),\tilde{s})\in\mathcal{T}_G$ will be written as
$s\stackrel{\Upsilon}{\rightarrow}_\mathcal{P}\tilde{s}$, interpreted as: the probability to change $s$ to $\tilde{s}$
as a result of executing $\Upsilon$ is $\mathcal{P}$.

For tangible states, $\Upsilon$ can be the empty set, and its execution does not change the current state (i.e. the
equivalence class), since we get a loop transition $s\stackrel{\emptyset}{\rightarrow}_\mathcal{P}s$ from a tangible
state $s$ to itself. This corresponds to the application of the empty loop rule to expressions from the equivalence
class. We keep track of such executions, called {\em empty loops}, since they have non-zero probabilities. This follows
from the definition of $PF(\emptyset ,s)$ and the fact that multiaction probabilities cannot be equal to $1$ as they
belong to $(0;1)$. For vanishing states, $\Upsilon$ cannot be the empty set, since we must execute some immediate
multiactions from them at the current moment.

The step probabilities belong to the interval $(0;1]$, being $1$ in the case when we cannot leave a tangible state $s$
and the only transition leaving it is the empty loop one $s\stackrel{\emptyset}{\rightarrow}_1 s$, or if there is just
a single transition from a vanishing state to any other one. We write $s\stackrel{\Upsilon}{\rightarrow}\tilde{s}$ if
$\exists\mathcal{P},\ s\stackrel{\Upsilon}{\rightarrow}_\mathcal{P}\tilde{s}$ and $s\rightarrow\tilde{s}$ if
$\exists\Upsilon ,\ s\stackrel{\Upsilon}{\rightarrow}\tilde{s}$.

The first equivalence we are going to introduce is isomorphism, which is a coincidence of systems up to renaming of
their components or states.

\begin{definition}
Let $TS(G)=(S_G,L_G,\mathcal{T}_G,s_G)$ and $TS(G')=(S_{G'},L_{G'},\mathcal{T}_{G'},s_{G'})$ be the transition systems
of dynamic expressions $G$ and $G'$, respectively. A mapping $\beta :S_G\rightarrow S_{G'}$ is an {\em isomorphism}
between $TS(G)$ and $TS(G')$, denoted by $\beta :TS(G)\simeq TS(G')$, if
\begin{enumerate}

\item $\beta$ is a bijection such that $\beta (s_G)=s_{G'}$;

\item $\forall s,\tilde{s}\in S_G,\ \forall\Upsilon ,\ s\stackrel{\Upsilon}{\rightarrow}_\mathcal{P}\tilde{s}\
\Leftrightarrow\ \beta (s)\stackrel{\Upsilon}{\rightarrow}_\mathcal{P}\beta (\tilde{s})$.

\end{enumerate}
Two transition systems $TS(G)$ and $TS(G')$ are {\em isomorphic}, denoted by $TS(G)\simeq TS(G')$, if $\exists\beta
:TS(G)\simeq TS(G')$.
\end{definition}

\begin{definition}
Two dynamic expressions $G$ and $G'$ are {\em equivalent w.r.t. transition systems}, denoted by $G=_{\rm ts}G'$,
if $TS(G)\simeq TS(G')$.
\end{definition}

\begin{example}
Consider the expression ${\sf Stop}=(\{g\},\frac{1}{2})\rs g$ specifying the non-terminating process that performs only
empty loops with probability $1$. Then, for $\rho ,\chi ,\theta ,\phi\in (0;1)$ and $l,m\in\reals_{>0}$, let\\
$E=[(\{a\},\rho )*((\{b\},\chi );(((\{c\},\natural_l);(\{d\},\theta ))\cho ((\{e\},\natural_m);(\{f\},\phi ))))*{\sf
Stop}]$.

$DR(\overline{E})$ consists of the equivalence classes\\
$\begin{array}{c}
s_1=[[\overline{(\{a\},\rho )}*((\{b\},\chi );(((\{c\},\natural_l);(\{d\},\theta ))\cho
((\{e\},\natural_m);(\{f\},\phi ))))*{\sf Stop}]]_\approx ,\\[1mm]
s_2=[[(\{a\},\rho )*(\overline{(\{b\},\chi )};(((\{c\},\natural_l);(\{d\},\theta ))\cho
((\{e\},\natural_m);(\{f\},\phi ))))*{\sf Stop}]]_\approx ,\\[1mm]
s_3=[[(\{a\},\rho )*((\{b\},\chi );\overline{(((\{c\},\natural_l);(\{d\},\theta ))\cho
((\{e\},\natural_m);(\{f\},\phi )))})*{\sf Stop}]]_\approx ,\\[1mm]
s_4=[[(\{a\},\rho )*((\{b\},\chi );(((\{c\},\natural_l);\overline{(\{d\},\theta ))}\cho
((\{e\},\natural_m);(\{f\},\phi ))))*{\sf Stop}]]_\approx ,\\[1mm]
s_5=[[(\{a\},\rho )*((\{b\},\chi );(((\{c\},\natural_l);(\{d\},\theta ))\cho
((\{e\},\natural_m);\overline{(\{f\},\phi ))}))*{\sf Stop}]]_\approx .
\end{array}$

We have $DR_{\rm T}(\overline{E})=\{s_1,s_2,s_4,s_5\}$ and $DR_{\rm V}(\overline{E})=\{s_3\}$. In the first part of
Figure \ref{tsboxrgnewrw.fig}, the transition system $TS(\overline{E})$ is presented. The tangible states are depicted
in ovals and the vanishing ones
in boxes. For simplicity of the graphical representation, the singleton sets of activities are written without outer
braces.
\label{ts.exm}
\end{example}

\section{Denotational semantics}
\label{denosem.sec}

In this section, we construct the denotational semantics via a subclass of labeled discrete time stochastic and
immediate PNs (LDTSIPNs), called discrete time stochastic and immediate Petri boxes (dtsi-boxes).

\subsection{Labeled DTSIPNs}

Let us introduce a class of labeled discrete time stochastic and immediate Petri nets (LDTSIPNs), a subclass of DTSPNs
\cite{Mol81,Mol85} (we do not allow the transition probabilities to be equal to $1$) extended with transition labeling
and immediate transitions. LDTSIPNs resemble in part discrete time deterministic and stochastic PNs (DTDSPNs)
\cite{ZFH01}, as well as discrete deterministic and stochastic PNs (DDSPNs) \cite{ZCH97}. DTDSPNs and DDSPNs are the
extensions of DTSPNs with deterministic transitions (having fixed delay that can be zero), inhibitor arcs, priorities
and guards. Next, while stochastic transitions of DTDSPNs, like those of DTSPNs, have geometrically distributed delays,
stochastic transitions of DDSPNs have discrete time phase distributed delays. Nevertheless, LDTSIPNs are not subsumed
by DTDSPNs or DDSPNs, since LDTSIPNs have a step semantics while DTDSPNs and DDSPNs have interleaving one. LDTSIPNs are
somewhat similar to labeled weighted DTSPNs from \cite{BT01}, but in the latter there are no immediate transitions, all
(stochastic) transitions have weights, the transition probabilities may be equal to $1$ and only maximal fireable
subsets of the enabled transitions are fired.

Stochastic preemptive time Petri nets (spTPNs) \cite{BSV05} is a discrete time model with a maximal step semantics,
where both time ticks and instantaneous parallel firings of maximal transition sets are possible, but the transition
steps in LDTSIPNs are not obliged to be maximal. The transition delays in spTPNs are governed by static general
discrete distributions, associated with the transitions, while the transitions of LDTSIPNs are only associated with
probabilities, used later to calculate the step probabilities after one unit (from tangible markings) or zero (from
vanishing markings) delay. Further, LDTSIPNs have just geometrically distributed or deterministic zero delays in the
markings. Moreover, the discrete time tick and concurrent transition firing are treated in spTPNs as different events
while firing every (possibly empty) set of stochastic transitions in LDTSIPNs requires one unit time delay. spTPNs are
essentially a modification and extension of unlabeled LWDTSPNs with additional facilities, such as inhibitor arcs,
priorities, resources, preemptions, schedulers etc. However, the price of such an expressiveness of spTPNs is that the
model is rather intricate and difficult to analyze.

Note that guards in DTDSPNs and DDSPNs, inhibitor arcs and priorities in DTDSPNs, DDSPNs and spTPNs, the maximal step
semantics of LWDTSPNs and spTPNs make these models Turing powerful, resulting in undecidability of important
behavioural properties.

\begin{definition}
A {\em labeled discrete time stochastic and immediate Petri net (LDTSIPN)} is a tuple\\
$N=(P_N,T_N,W_N,\Omega_N,\mathcal{L}_N,M_N)$, where
\begin{itemize}

\item $P_N$ and $T_N={\it Ts}_N\uplus{\it Ti}_N$ are finite sets of {\em places} and {\em stochastic and immediate
transitions}, respectively, such that $P_N\cup T_N\neq\emptyset$ and $P_N\cap T_N=\emptyset$;

\item $W_N:(P_N\times T_N)\cup (T_N\times P_N)\rightarrow\naturals$ is a function for the {\em weights of arcs}
between places and transitions;

\item $\Omega_N$ is the {\em transition probability and weight} function such that
\begin{itemize}

\item $\Omega_N|_{{\it Ts}_N}:{\it Ts}\rightarrow (0;1)$ (it associates stochastic transitions with probabilities);

\item $\Omega_N|_{{\it Ti}_N}:{\it Ti}\rightarrow\reals_{>0}$ (it associates immediate transitions with
weights);

\end{itemize}

\item $\mathcal{L}_N:T_N\rightarrow\mathcal{L}$ is the {\em transition labeling} function assigning multiactions to
transitions;

\item $M_N\in\naturals_{\rm fin}^{P_N}$ is the {\em initial marking}.

\end{itemize}
\end{definition}

The graphical representation of LDTSIPNs is like that for standard labeled PNs, but with probabilities or weights
written near the corresponding transitions. Square boxes of normal thickness depict stochastic transitions, and those
with thick borders represent immediate transitions. If the probabilities or the weights are not given in the picture,
they are considered to be of no importance in the corresponding examples, such as those describing the stationary
behaviour. The weights of arcs are depicted with them. The names of places and transitions are depicted near them when
needed.

Let $N$ be an LDTSIPN and $t\in T_N,\ U\in\naturals_{\rm fin}^{T_N}$. The {\em precondition} $^\bullet t$ and the
{\em postcondition} $t^\bullet$ of $t$ are the multisets of places $({^\bullet}t)(p)=W_N(p,t)$ and $(t^\bullet
)(p)=W_N(t,p)$. The {\em precondition} $^\bullet U$ and the {\em postcondition} $U^\bullet$ of $U$ are the multisets of
places ${^\bullet}U=\sum_{t\in U}{^\bullet}t$ and $U^\bullet =\sum_{t\in U}t^\bullet$. Note that for $U=\emptyset$ we
have ${^\bullet}\emptyset =\emptyset =\emptyset^\bullet$.

Let $N$ be an LDTSIPN and $M,\widetilde{M}\in\naturals_{\rm fin}^{P_N}$. Immediate transitions have a priority over
stochastic ones, thus, immediate transitions always fire first if they can. A transition $t\in T_N$ is {\em enabled}
at $M$ if ${^\bullet}t\subseteq M$ and one of the following holds: (1) $t\in{\it Ti}_N$ or (2) $\forall u\in T_N,\
{^\bullet}u\subseteq M\ \Rightarrow\ u\in{\it Ts}_N$.

Thus, a transition is enabled
at a marking if it has enough tokens in
its input places (i.e. in
the places from its precondition) and it is immediate one; otherwise, when it is stochastic,
there exists no immediate transition with enough tokens in
its input places. Let $Ena(M)$ be the set of {\em all transitions enabled
at $M$}. By definition, it follows that $Ena(M)\subseteq{\it Ti}_N$ or $Ena(M)\subseteq {\it Ts}_N$. A set of
transitions $U\subseteq Ena(M)$ is {\em enabled}
at a marking $M$ if ${^\bullet}U\subseteq M$. Firings of transitions are atomic operations, and transitions may fire
concurrently in steps. We assume that all transitions participating in a step should differ, hence, only the sets (not
multisets) of transitions may fire. Thus, we do not allow self-concurrency, i.e. firing of transitions in parallel to
themselves. This restriction is introduced to avoid some technical difficulties while calculating probabilities for
multisets of transitions as we shall see after the following formal definitions. Moreover, we do not need to consider
self-concurrency, since denotational semantics of expressions will be defined via dtsi-boxes which are safe LDTSIPNs
(hence, no self-concurrency is~possible).

The marking $M$ is {\em tangible}, denoted by $tang(M)$, if $Ena(M)\subseteq{\it Ts}_N$, in particular, if
$Ena(M)=\emptyset$. Otherwise, the marking $M$ is {\em vanishing}, denoted by $vanish(M)$, and in this case
$Ena(M)\subseteq{\it Ti}_N$ and $Ena(M)\neq\emptyset$. If $tang(M)$ then a stochastic transition $t\in Ena(M)$ fires
with probability $\Omega_N(t)$ when no other stochastic transitions conflicting with it are enabled.

Let $U\subseteq Ena(M),\ U\neq\emptyset$ and ${^\bullet}U\subseteq M$. The {\em probability that the set of stochastic
transitions $U$ is ready for firing in $M$} or the {\em weight of the set of immediate transitions $U$ which is ready
for firing in $M$} is
$$PF(U,M)=
\left\{
\begin{array}{ll}
\displaystyle\prod_{t\in U}\Omega_N(t)\cdot\prod_{u\in Ena(M)\setminus U}(1-\Omega_N(u)), & \mbox{if }tang(M);\\
\displaystyle\sum_{t\in U}\Omega_N(t), & \mbox{if }vanish(M).
\end{array}
\right.$$

In the case $U=\emptyset$ and $tang(M)$ we define
$$PF(\emptyset ,M)=
\left\{
\begin{array}{ll}
\displaystyle\prod_{u\in Ena(M)}(1-\Omega_N(u)), & \mbox{if }Ena(M)\neq\emptyset ;\\
1, & \mbox{if }Ena(M)=\emptyset .
\end{array}
\right.$$

Let $U\subseteq Ena(M),\ U\neq\emptyset$ and ${^\bullet}U\subseteq M$ or $U=\emptyset$ and $tang(M)$. Besides $U$, some
other sets of transitions may be ready for firing in $M$, hence, conditioning or normalization is needed to calculate
the firing probability. The concurrent firing of the transitions from $U$ changes the marking $M$ to
$\widetilde{M}=M-{^\bullet}U+U^\bullet$, denoted by $M\stackrel{U}{\rightarrow}_\mathcal{P}\widetilde{M}$, where
$\mathcal{P}=PT(U,M)$ is the {\em probability that the set of transitions $U$ fires in $M$} defined as
$$PT(U,M)=\frac{PF(U,M)}{\displaystyle\sum_{\{V\mid{^\bullet}V\subseteq M\}}PF(V,M)}.$$

Observe that in the case $U=\emptyset$ and $tang(M)$ we have $M=\widetilde{M}$. Note that for all markings of an
LDTSIPN $N$, the sum of outgoing probabilities is equal to $1$, i.e.
$\forall M\in\naturals_{\rm fin}^{P_N},\ \sum_{\{U\mid{^\bullet}U\subseteq M\}}PT(U,M)=1$. This follows from the
definition of $PT(U,M)$ and guarantees that it defines a probability distribution.

We write $M\stackrel{U}{\rightarrow}\widetilde{M}$ if $\exists\mathcal{P},\
M\stackrel{U}{\rightarrow}_\mathcal{P}\widetilde{M}$ and $M\rightarrow\widetilde{M}$ if $\exists U,\
M\stackrel{U}{\rightarrow}\widetilde{M}$.

The {\em probability to move from $M$ to $\widetilde{M}$ by firing any set of transitions} is
$$PM(M,\widetilde{M})=\sum_{\{U\mid M\stackrel{U}{\rightarrow}\widetilde{M}\}}PT(U,M).$$

Since $PM(M,\widetilde{M})$ is the probability for {\em any} (including the empty one) transition set to change marking
$M$ to $\widetilde{M}$, we use summation in the definition. Note that $\forall M\in\naturals_{\rm fin}^{P_N},\
\sum_{\{\widetilde{M}\mid M\rightarrow\widetilde{M}\}}PM(M,\widetilde{M})=\\
\sum_{\{\widetilde{M}\mid M\rightarrow\widetilde{M}\}}\sum_{\{U\mid M\stackrel{U}{\rightarrow}\widetilde{M}\}}PT(U,M)=
\sum_{\{U\mid{^\bullet}U\subseteq M\}}PT(U,M)=1$.

\begin{definition}
Let $N$ be an LDTSIPN. The {\em reachability set} of $N$, denoted by $RS(N)$, is the minimal set of markings such that
\begin{itemize}

\item $M_N\in RS(N)$;

\item if $M\in RS(N)$ and $M\rightarrow\widetilde{M}$ then $\widetilde{M}\in RS(N)$.

\end{itemize}
\end{definition}

\begin{definition}
Let $N$ be an LDTSIPN. The {\em reachability graph} of $N$ is a (labeled probabilistic) transition system
$RG(N)=(S_N,L_N,\mathcal{T}_N,s_N)$, where
\begin{itemize}

\item the set of {\em states} is $S_N=RS(N)$;

\item the set of {\em labels} is $L_N=2^{T_N}\times (0;1]$;

\item the set of {\em transitions} is $\mathcal{T}_N=\{(M,(U,\mathcal{P}),\widetilde{M})\mid M,\widetilde{M}\in
RS(N),\ M\stackrel{U}{\rightarrow}_\mathcal{P}\widetilde{M}\}$;

\item the {\em initial state} is $s_N=M_N$.

\end{itemize}
\end{definition}

Let $RS_{\rm T}(N)$ be the set of {\em all tangible markings from $RS(N)$} and $RS_{\rm V}(N)$ be the set of {\em all
vanishing markings from $RS(N)$}. Obviously, $RS(N)=RS_{\rm T}(N)\uplus RS_{\rm V}(N)$.

\subsection{Algebra of dtsi-boxes}

We now introduce discrete time stochastic and immediate Petri boxes and algebraic operations to define the net
representation of dtsiPBC expressions.

\begin{definition}
A {\em discrete time stochastic and immediate Petri box (dtsi-box)} is a tuple\\
$N=(P_N,T_N,W_N,\Lambda_N)$, where
\begin{itemize}

\item $P_N$ and $T_N$ are finite sets of {\em places} and {\em transitions}, respectively, such that $P_N\cup
T_N\neq\emptyset$ and $P_N\cap T_N=\emptyset$;

\item $W_N:(P_N\times T_N)\cup (T_N\times P_N)\rightarrow\naturals$ is a function providing the {\em weights of
arcs};

\item $\Lambda_N$ is the {\em place and transition labeling} function such that
\begin{itemize}

\item $\Lambda_N|_{P_N}:P_N\rightarrow\{{\sf e},{\sf i},{\sf x}\}$ (it specifies {\em entry, internal} and {\em exit}
places);

\item $\Lambda_N|_{T_N}:T_N\rightarrow\{\varrho\mid\varrho\subseteq 2^\mathcal{SIL}\times\mathcal{SIL}\}$ (it
associates transitions with {\em relabeling relations}
on activities).

\end{itemize}

\end{itemize}

Moreover, $\forall t\in T_N,\ {^\bullet}t\neq\emptyset\neq t^\bullet$. Next, for the set of {\em entry} places of $N$,
defined as ${^\circ}N=\{p\in P_N\mid\Lambda_N(p)={\sf e}\}$, and for the set of {\em exit} places of $N$, defined as
$N^\circ=\{p\in P_N\mid\Lambda_N(p)={\sf x}\}$, the following holds: ${^\circ}N\neq\emptyset\neq N^\circ ,\
{^\bullet}({^\circ}N)=\emptyset =(N^\circ )^\bullet$.
\end{definition}

A dtsi-box is {\em plain} if $\forall t\in T_N,\ \exists (\alpha ,\kappa )\in\mathcal{SIL},\
\Lambda_N(t)=\varrho_{(\alpha ,\kappa )}$, where $\varrho_{(\alpha ,\kappa )}=\{(\emptyset ,(\alpha ,\kappa ))\}$ is a
{\em constant relabeling},
identified with the activity $(\alpha ,\kappa )$. A {\em marked plain dtsi-box} is a pair $(N,M_N)$, where $N$ is a
plain dtsi-box and $M_N\in\naturals_{\rm fin}^{P_N}$ is its marking. We denote $\overline{N}=(N,{^\circ}N)$ and
$\underline{N}=(N,N^\circ )$. Note that a marked plain dtsi-box $(P_N,T_N,W_N,\Lambda_N,M_N)$ could be interpreted as
the LDTSIPN $(P_N,T_N,W_N,\Omega_N,\mathcal{L}_N,M_N)$, where functions $\Omega_N$ and $\mathcal{L}_N$ are defined as
follows:
$\forall t\in T_N,\\
\Omega_N(t)=\kappa$ if $\kappa\in (0;1)$; or $\Omega_N(t)=l$ if $\kappa =\natural_l,\ l\in \reals_{>0}$; and
$\mathcal{L}_N(t)=\alpha$, where $\Lambda_N(t)=\varrho_{(\alpha ,\kappa )}$. Behaviour of the marked dtsi-boxes follows
from the firing rule of LDTSIPNs. A plain dtsi-box $N$ is {\em $n$-bounded} ($n\in\naturals$) if $\overline{N}$ is so,
i.e. $\forall M\in RS(\overline{N}),\ \forall p\in P_N,\ M(p)\leq n$, and it is {\em safe} if it is $1$-bounded. A
plain dtsi-box $N$ is {\em clean} if $\forall M\in RS(\overline{N}),\ {^\circ}N\subseteq M\ \Rightarrow\ M={^\circ}N$
and $N^\circ\subseteq M\ \Rightarrow\ M=N^\circ$, i.e. if there are tokens in all its entry (exit) places then no other
places have tokens.

The structure of the plain dtsi-box corresponding to a static expression is constructed like in PBC \cite{BKo95,BDK01},
i.e. we use simultaneous refinement and relabeling meta-operator (net refinement) in addition to the {\em operator
dtsi-boxes} corresponding to the algebraic operations of dtsiPBC and featuring transformational transition relabelings.
Operator dtsi-boxes specify $n$-ary functions from plain dtsi-boxes to plain dtsi-boxes (we have $1\leq n\leq 3$ in
dtsiPBC). Thus, as we shall see in Theorem \ref{safeclean.the}, the resulting plain dtsi-boxes are safe and clean. In
the definition of the denotational semantics, we shall apply standard constructions used for PBC. Let $\Theta$ denote
{\em operator box} and $u$ denote {\em transition name} from PBC setting.

The relabeling relations $\varrho\subseteq 2^\mathcal{SIL}\times\mathcal{SIL}$ are defined as follows:
\begin{itemize}

\item $\varrho_{\rm id}=\{(\{(\alpha ,\kappa )\},(\alpha ,\kappa ))\mid (\alpha ,\kappa )\in\mathcal{SIL}\}$ is the
{\em identity relabeling};

\item $\varrho_{(\alpha ,\kappa )}=\{(\emptyset ,(\alpha ,\kappa ))\}$ is the {\em constant relabeling}, identified
with $(\alpha ,\kappa )\in\mathcal{SIL}$;

\item $\varrho_{[f]}=\{(\{(\alpha ,\kappa )\},(f(\alpha ),\kappa ))\mid (\alpha ,\kappa )\in\mathcal{SIL}\}$;

\item $\varrho_{\!\!\rs a}=\{(\{(\alpha ,\kappa )\},(\alpha ,\kappa ))\mid (\alpha ,\kappa )\in\mathcal{SIL},\
a,\hat{a}\not\in\alpha\}$;

\item $\varrho_{\!\!\sy a}$ is the least relabeling containing $\varrho_{\rm id}$ such that if $(\Upsilon ,(\alpha
,\kappa )),(\Xi ,(\beta ,\lambda ))\in\varrho_{\!\!\sy a},\ a\in\alpha ,\ \hat{a}\in\beta$ then
\begin{itemize}

\item $(\Upsilon +\Xi ,(\alpha\oplus_a\beta ,\kappa\cdot\lambda ))\in\varrho_{\!\!\sy a}$ if $\kappa ,\lambda\in
(0;1)$;

\item $(\Upsilon +\Xi ,(\alpha\oplus_a\beta ,\natural_{l+m}))\in\varrho_{\!\!\sy a}$ if $\kappa =\natural_l,\ \lambda
=\natural_m,\ l,m\in\reals_{>0}$.

\end{itemize}

\end{itemize}

The plain dtsi-boxes $N_{(\alpha ,\rho )_{\iota}},\ N_{(\alpha ,\natural_l)_{\iota}}$, where $\rho\in (0;1),\
l\in\reals_{>0}$, and operator dtsi-boxes are presented in Figure \ref{dtsiboxopsmrw.fig}. The label {\sf i} of
internal places is usually omitted.

\begin{figure}
\begin{center}
\includegraphics[width=\textwidth]{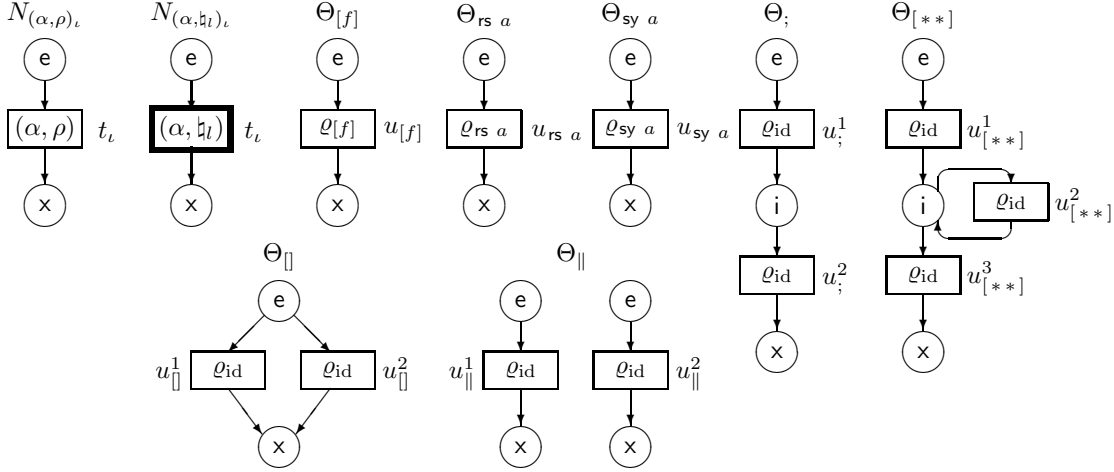}
\end{center}
\vspace{-6mm}
\caption{The plain and operator dtsi-boxes.}
\label{dtsiboxopsmrw.fig}
\end{figure}

In the case of the iteration, a decision that we must take is the selection of the operator box that we shall use for
it, since we have two proposals in plain PBC for that purpose \cite{BDK01}. One of them provides us with a safe version
with six transitions in the operator box, but there is also a simpler version, which has only three transitions. In
general, in PBC, with the latter version we may generate $2$-bounded nets, which only occurs when a parallel behaviour
appears at the highest level of the body of the iteration. Nevertheless, in our case, and due to the syntactical
restriction introduced for regular terms, this particular situation cannot occur, so that the net obtained will be
always safe.

To construct a semantic function
assigning a plain dtsi-box
to every static expression of dtsiPBC, we
define the {\em enumeration} function $Enu:T_N\rightarrow Num$, which associates the numberings with transitions of a
plain dtsi-box $N$ according to those of activities. For synchronization, the function associates with the re\-sulting
new transition the concatenation of the parenthesized numberings of the transitions it comes~from.

We now define the enumeration function $Enu$ for every operator of dtsiPBC. Let $Box_{\rm dtsi}(E)=\\
(P_E,T_E,W_E,\Lambda_E)$ be the plain dtsi-box corresponding to a static expression $E$, and $Enu_E:T_E\rightarrow Num$
be the enumeration function for $Box_{\rm dtsi}(E)$. We
use the similar notation for static expressions $F$ and
$K$.
\begin{itemize}

\item $Box_{\rm dtsi}((\alpha ,\kappa )_{\iota})=N_{(\alpha ,\kappa )_{\iota}}$. Since a single transition $t_{\iota}$
corresponds to the activity $(\alpha ,\kappa )_{\iota}\in\mathcal{SIL}$, their numberings coincide:
$Enu(t_{\iota})=\iota$.

\item $Box_{\rm dtsi}(E\circ F)=\Theta_{\circ}(Box_{\rm dtsi}(E),Box_{\rm dtsi}(F)),\ \circ\in\{;,\cho ,\|\}$. Since we
do not introduce new transitions, we preserve the initial numbering:
$Enu(t)=\left\{
\begin{array}{ll}
Enu_E(t), & \mbox{if }t\in T_E;\\
Enu_F(t), & \mbox{if }t\in T_F.
\end{array}
\right.$

\item $Box_{\rm dtsi}(E[f])=\Theta_{[f]}(Box_{\rm dtsi}(E))$. Since we only replace the labels of some multiactions by
a bijection, we preserve the initial numbering: $Enu(t)=Enu_E(t),\ t\in T_E$.

\item $Box_{\rm dtsi}(E\rs a)=\Theta_{\!\!\rs a}(Box_{\rm dtsi}(E))$. Since we remove all transitions labeled with
multiactions containing $a$ or $\hat{a}$, the remaining transitions numbering is not changed:
$Enu(t)=Enu_E(t),\\
t\in T_E,\ a,\hat{a}\not\in\alpha ,\ \Lambda_E(t)=\varrho_{(\alpha ,\kappa )}$.

\item $Box_{\rm dtsi}(E\sy a)=\Theta_{\!\!\sy a}(Box_{\rm dtsi}(E))$. Note that $\forall v,w\in T_E$, such that
$\Lambda_E(v)=\varrho_{(\alpha ,\kappa )},\\
\Lambda_E(w)=\varrho_{(\beta ,\lambda )}$ and $a\in\alpha ,\ \hat{a}\in\beta$, the new transition $t$ resulting
from synchronization of $v$ and $w$ has the label
$\Lambda (t)=\varrho_{(\alpha\oplus_a\beta ,\kappa\cdot\lambda )}$ if $t$ is a stochastic transition; or
$\Lambda (t)=\varrho_{(\alpha\oplus_a\beta ,\natural_{l+m})}$ if $t$ is an immediate one ($\kappa =\natural_l,\
\lambda =\natural_m,\ l,m\in\reals_{>0}$); and the numbering $Enu(t)\!=\!(Enu_E(v))(Enu_E(w))$. Thus, the
enumeration function is defined as $Enu(t)=\left\{
\begin{array}{ll}
Enu_E(t), & \mbox{if }t\in T_E;\\
(Enu_E(v))(Enu_E(w)), & \mbox{if }t\mbox{ results}\\
 & \hspace{-32mm}\mbox{from synchronization of }v\mbox{ and }w.
\end{array}
\right.$

According to the definition of $\varrho_{\!\!\sy a}$, the synchronization is only possible when all the transitions
in the set are stochastic or
all of them are immediate. If we synchronize the same set of transitions in different orders, we obtain several
resulting transitions with the same label and probability or weight, but with the different numberings having the
same content. Then we only consider a single transition from the resulting ones in the plain dtsi-box to avoid
introducing redundant transitions.

For example, if the transitions $t$ and $u$ are generated by synchronizing $v$ and $w$ in different orders, we have
$\Lambda (t)=\varrho_{(\alpha\oplus_a\beta ,\kappa\cdot\lambda )}=\Lambda (u)$ for stochastic transitions or
$\Lambda (t)=\varrho_{(\alpha\oplus_a\beta ,\natural_{l+m})}=\Lambda (u)$ for immediate ones ($\kappa =\natural_l,\
\lambda =\natural_m,\ l,m\in\reals_{>0}$), but $Enu(t)=(Enu_E(v))(Enu_E(w))\neq (Enu_E(w))(Enu_E(v))=Enu(u)$
while $Cont(Enu(t))=Cont(Enu(v))\cup Cont(Enu(w))=Cont(Enu(u))$. Then only one transition $t$ (or, symmetrically,
$u$) will appear in $Box_{\rm dtsi}(E\sy a)$.

\item $Box_{\rm dtsi}([E*F*K])=\Theta_{[\,*\,*\,]}(Box_{\rm dtsi}(E),Box_{\rm dtsi}(F),Box_{\rm dtsi}(K))$. Since we do
not introduce new transitions, we preserve the initial numbering:
$Enu(t)=\left\{
\begin{array}{ll}
Enu_E(t), & \mbox{if }t\in T_E;\\
Enu_F(t), & \mbox{if }t\in T_F;\\
Enu_K(t), & \mbox{if }t\in T_K.
\end{array}
\right.$

\end{itemize}

We now can formally define the denotational semantics as a homomorphism.

\begin{definition}
Let $(\alpha ,\kappa )\in\mathcal{SIL},\ a\in Act$ and $E,F,K\in RegStatExpr$. The {\em denotational semantics} of
dtsiPBC is a mapping $Box_{\rm dtsi}$ from $RegStatExpr$ into the domain of plain dtsi-boxes defined as~follows:
\begin{enumerate}

\item $Box_{\rm dtsi}((\alpha ,\kappa )_{\iota})=N_{(\alpha ,\kappa )_{\iota}}$;

\item $Box_{\rm dtsi}(E\circ F)=\Theta_{\circ}(Box_{\rm dtsi}(E),Box_{\rm dtsi}(F)),\ \circ\in\{;,\cho ,\|\}$;

\item $Box_{\rm dtsi}(E[f])=\Theta_{[f]}(Box_{\rm dtsi}(E))$;

\item $Box_{\rm dtsi}(E\circ a)=\Theta_{\circ a}(Box_{\rm dtsi}(E)),\ \circ\in\{\!\!\rs\!\!,\!\!\sy\!\!\}$;

\item $Box_{\rm dtsi}([E*F*K])=\Theta_{[\,*\,*\,]}(Box_{\rm dtsi}(E),Box_{\rm dtsi}(F),Box_{\rm dtsi}(K))$.

\end{enumerate}
\end{definition}

The dtsi-boxes of dynamic expressions can be defined. For $E\in RegStatExpr$, let $Box_{\rm dtsi}(\overline{E})\!=\!
\overline{Box_{\rm dtsi}(E)}$ and $Box_{\rm dtsi}(\underline{E})=\underline{Box_{\rm dtsi}(E)}$. This definition is
compositional in the sense that, for any arbitrary dynamic expression, we may decompose it in some inner dynamic and
static expressions, for which we may apply the definition, thus obtaining the corresponding plain dtsi-boxes, which can
be joined according to the term structure (by definition of $Box_{\rm dtsi}$), the resulting plain box being marked in
the places that were marked in the argument nets.

\begin{theorem}
For any static expression $E,\ Box_{\rm dtsi}(\overline{E})$ is safe and clean.
\label{safeclean.the}
\end{theorem}
\begin{proof}
The structure of the net is obtained as in PBC \cite{BKo95,BDK01}, combining both refinement and relabeling. Hence, the
dtsi-boxes thus obtained will be safe and clean.
\end{proof}

Let $\simeq$ denote isomorphism between transition systems and reachability graphs that binds their initial states. The
names of transitions of the dtsi-box corresponding to a static expression could be identified with the enumerated
activities of the latter.

\begin{theorem}
For any static expression $E,\ TS(\overline{E})\simeq RG(Box_{\rm dtsi}(\overline{E}))$.
\label{opdensem.the}
\end{theorem}
\begin{proof}
As for the qualitative behaviour, we have the same isomorphism as in PBC \cite{BKo95,BDK01}.
The quantitative behaviour is the same by the following reasons. First, the activities of an expression have the
probability or weight parts coinciding with the probabilities or weights of the transitions belonging to the
corresponding dtsi-box. Second, we use analogous probability or weight functions to construct the corresponding
transition systems and reachability graphs.
\end{proof}

\begin{example}
Let $E$ be from Example \ref{ts.exm}. In Figure \ref{tsboxrgnewrw.fig}, the marked dtsi-box $N=Box_{\rm
dtsi}(\overline{E})$ and its reachability graph $RG(N)$ are presented. It is easy to see that $TS(\overline{E})$ and
$RG(N)$ are isomorphic.
\label{boxrg.exm}
\end{example}

\begin{figure}
\begin{center}
\includegraphics[scale=0.9]{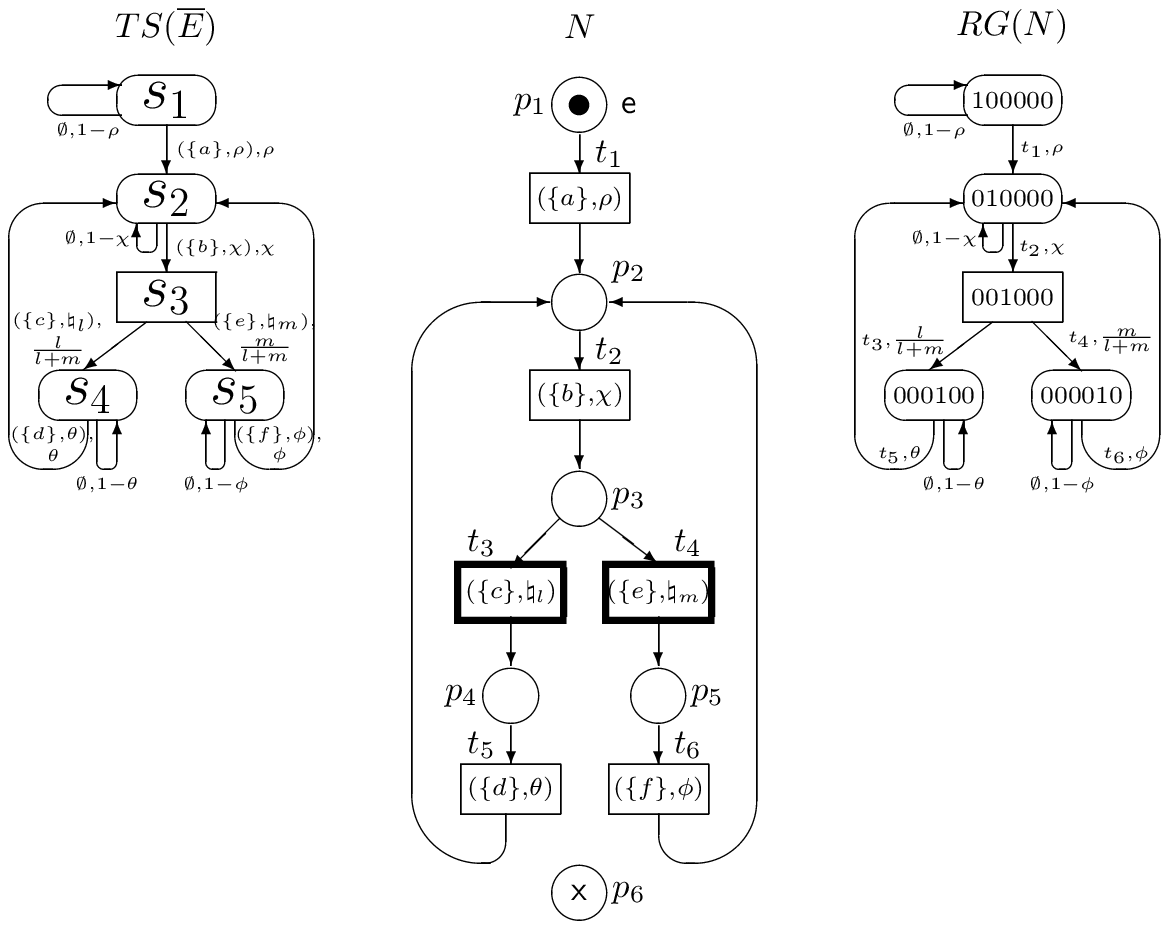}
\end{center}
\vspace{-6mm}
\caption{The transition system of $\overline{E}$, marked dtsi-box $N={\it Box}_{{\rm dtsi}}(\overline{E})$ and its
reachability graph for $E=[(\{a\},\rho )*\protect\newline
((\{b\},\chi );(((\{c\},\natural_l);(\{d\},\theta ))\cho ((\{e\},\natural_m);(\{f\},\phi ))))*{\sf Stop}]$.}
\label{tsboxrgnewrw.fig}
\end{figure}

\section{Performance evaluation}
\label{perfeval.sec}

In this section we demonstrate how Markov chains corresponding to the expressions and dtsi-boxes can be constructed and
then used for performance evaluation.

\subsection{Analysis of the underlying SMC}

For a dynamic expression $G$, a discrete random variable is associated with every tangible state $s\in DR_{\rm T}(G)$.
The variable captures a residence time in the state. One can interpret staying in a state at the next discrete time
moment as a failure and leaving it as a success of some trial series. It is easy to see that the random variables are
geometrically distributed with the parameter $1-PM(s,s)$, since the probability to stay in $s$ for $k-1$ time moments
and leave it at the moment $k\geq 1$ is $PM(s,s)^{k-1}(1-PM(s,s))$ (the residence time is $k$ in this case, and this
formula defines the probability mass function (PMF) of residence time in $s$). Hence, the probability distribution
function (PDF) of residence time in $s$ is $1-PM(s,s)^k\ (k\geq 0)$ (the probability that the residence time in $s$ is
less than or equal to $k$). The mean value formula for the geometrical distribution allows us to calculate the average
sojourn time in $s$ as $\frac{1}{1-PM(s,s)}$. Clearly, the average sojourn time in a vanishing state is zero. Let $s\in
DR(G)$.

The {\em average sojourn time in the state $s$} is
$${\it SJ}(s)=
\left\{
\begin{array}{ll}
\frac{1}{1-PM(s,s)}, & \mbox{if }s\in DR_{\rm T}(G);\\
0, & \mbox{if }s\in DR_{\rm V}(G).
\end{array}
\right.$$
The {\em average sojourn time vector} ${\it SJ}$ of $G$ has the elements ${\it SJ}(s),\ s\in DR(G)$.

The {\em sojourn time variance in the state $s$} is
$${\it VAR}(s)=
\left\{
\begin{array}{ll}
\frac{PM(s,s)}{(1-PM(s,s))^2}, & \mbox{if }s\in DR_{\rm T}(G);\\
0, & \mbox{if }s\in DR_{\rm V}(G).
\end{array}
\right.$$
The {\em sojourn time variance vector} ${\it VAR}$ of $G$ has the elements ${\it VAR}(s),\ s\in DR(G)$.

To evaluate performance of the system specified by a dynamic expression $G$, we should investigate the stochastic
process associated with it. The process is the underlying SMC \cite{Ros96,Kul09}, denoted by ${\it SMC}(G)$, which can
be analyzed by extracting from it the embedded (absorbing) discrete time Markov chain (EDTMC) corresponding to $G$,
denoted by ${\it EDTMC}(G)$. The construction of the latter is analogous to that applied in GSPNs
\cite{Mar90,Bal01,Bal07}, DTDSPNs \cite{ZFH01} and DDSPNs \cite{ZCH97}. ${\it EDTMC}(G)$ only describes the state
changes of ${\it SMC}(G)$ while ignoring its time characteristics. Thus, to construct the EDTMC, we should abstract
from all time aspects of behaviour of the SMC, i.e. from the sojourn time in its states. The (local) sojourn time in
every state of the EDTMC is equal to one discrete time unit. Each SMC is fully described by the EDTMC and the state
sojourn time distributions (the latter can be specified by the vector of PDFs of residence time in the states)
\cite{Hav01}.

Let $G$ be a dynamic expression and $s,\tilde{s}\in DR(G)$. The transition system $TS(G)$ can have self-loops from a
state to itself with a non-zero probability.
Clearly, the current state remains unchanged in this case.

Let $s\rightarrow s$. The {\em probability to stay in $s$ due to $k\ (k\geq 1)$ self-loops} is $PM(s,s)^k$.

Let $s\rightarrow\tilde{s}$ and $s\neq\tilde{s}$. The {\em probability to move from $s$ to $\tilde{s}$ by
executing any set of activities after possible self-loops} is
$$\begin{array}{c}
PM^*(s,\tilde{s})=\left\{
\begin{array}{ll}
PM(s,\tilde{s})\sum_{k=0}^{\infty}PM(s,s)^k=\frac{PM(s,\tilde{s})}{1-PM(s,s)}, & \mbox{if }s\rightarrow s;\\
PM(s,\tilde{s}), & \mbox{otherwise};
\end{array}
\right\}={\it SL}(s)PM(s,\tilde{s}),\\
\mbox{ where }{\it SL}(s)=\left\{
\begin{array}{ll}
\frac{1}{1-PM(s,s)}, & \mbox{if }s\rightarrow s;\\
1, & \mbox{otherwise}.
\end{array}
\right.
\end{array}$$
Here ${\it SL}(s)$ is the {\em self-loops abstraction factor in the state $s$}. The {\em self-loops
abstraction vector} of $G$, denoted by ${\it SL}$, has the elements ${\it SL}(s),\ s\in DR(G)$. The value $k=0$ in the
summation above corresponds to the case when no self-loops occur. Note that $\forall s\in DR_{\rm T}(G),\ {\it
SL}(s)=\frac{1}{1-PM(s,s)}={\it SJ}(s)$, hence, $\forall s\in DR_{\rm T}(G),\ PM^*(s,\tilde{s})={\it
SJ}(s)PM(s,\tilde{s})$, since we always have the empty loop (the self-loop) $s\stackrel{\emptyset}{\rightarrow}s$ from
every tangible state $s$. Empty loops are not possible from vanishing states, hence, $\forall s\in DR_{\rm V}(G),\
PM^*(s,\tilde{s})=\frac{PM(s,\tilde{s})}{1-PM(s,s)}$, when there are non-empty self-loops (produced by iteration) from
$s$, or $PM^*(s,\tilde{s})=PM(s,\tilde{s})$, when there are no self-loops from $s$.

Notice that $PM^*(s,\tilde{s})$ defines a probability distribution, since $\forall s\in DR(G)$, such that $s$ is not a
terminal state, i.e. there are transitions to different states after possible self-loops from it, we have\\
$\sum_{\{\tilde{s}\mid s\rightarrow\tilde{s},\ s\neq\tilde{s}\}}PM^*(s,\tilde{s})=
\frac{1}{1-PM(s,s)}\sum_{\{\tilde{s}\mid s\rightarrow\tilde{s},\ s\neq\tilde{s}\}}PM(s,\tilde{s})=
\frac{1}{1-PM(s,s)}(1-PM(s,s))=1$.

\begin{definition}
Let $G$ be a dynamic expression. The {\em embedded (absorbing) discrete time Markov chain (EDTMC)} of $G$, denoted by
${\it EDTMC}(G)$, has the state space $DR(G)$, the initial state $[G]_\approx$ and the transitions
$s\doublera_\mathcal{P}\tilde{s}$ if $s\rightarrow\tilde{s}$ and $s\neq\tilde{s}$, where
$\mathcal{P}=PM^*(s,\tilde{s})$.

The {\em underlying SMC} of $G$, denoted by ${\it SMC}(G)$, has the EDTMC ${\it EDTMC}(G)$ and the sojourn time in
every $s\in DR_{\rm T}(G)$ is geometrically distributed with the parameter $1-PM(s,s)$ while the sojourn time in every
$s\in DR_{\rm V}(G)$ is zero.
\end{definition}

Let $G$ be a dynamic expression. The elements $\mathcal{P}_{ij}^*\ (1\leq i,j\leq n=|DR(G)|)$ of the (one-step)
transition probability matrix (TPM) ${\bf P}^*$ for ${\it EDTMC}(G)$ are
$\mathcal{P}_{ij}^*=\left\{
\begin{array}{ll}
PM^*(s_i,s_j), & \mbox{if }s_i\rightarrow s_j,\ s_i\neq s_j;\\
0, & \mbox{otherwise}.
\end{array}
\right.$

\noindent The transient ($k$-step, $k\in\naturals$) PMF $\psi^*[k]=(\psi^*[k](s_1),\ldots ,\psi^*[k](s_n))$ for ${\it
EDTMC}(G)$ is calculated as
$$\psi^*[k]=\psi^*[0]({\bf P}^*)^k,$$
where $\psi^*[0]\!=\!(\psi^*[0](s_1),\ldots ,\psi^*[0](s_n))$ is the initial PMF defined as
$\psi^*[0](s_i)\!=\!\left\{
\begin{array}{ll}
1, & \mbox{if }s_i=[G]_\approx ;\\
0, & \mbox{otherwise}.
\end{array}
\right.$\\
Note also that $\psi^*[k+1]=\psi^*[k]{\bf P}^*\ (k\in\naturals)$.

\noindent The steady-state PMF $\psi^*=(\psi^*(s_1),\ldots ,\psi^*(s_n))$ for ${\it EDTMC}(G)$ is a solution of the
equation system
$$\left\{
\begin{array}{l}
\psi^*({\bf P}^*-{\bf I})={\bf 0}\\
\psi^*{\bf 1}^{\rm T}=1
\end{array}
\right.,$$
where ${\bf I}$ is the identity matrix of order $n$ and ${\bf 0}$ (${\bf 1}$) is a row vector of $n$ values $0$ ($1$).

Note that the vector $\psi^*$ exists and is unique if ${\it EDTMC}(G)$ is ergodic. Then ${\it EDTMC}(G)$ has a single
steady state, and we have $\psi^*=\lim_{k\to\infty}\psi^*[k]$.

The steady-state PMF for the underlying semi-Markov chain ${\it SMC}(G)$ is calculated via multiplication of every
$\psi^*(s_i)\ (1\leq i\leq n)$ by the average sojourn time ${\it SJ}(s_i)$ in the state $s_i$, after which we normalize
the resulting values. Remember that for a vanishing state $s\in DR_{\rm V}(G)$ we have ${\it SJ}(s)=0$.

Thus, the steady-state PMF $\varphi =(\varphi (s_1),\ldots ,\varphi (s_n))$ for ${\it SMC}(G)$ is
$$\varphi (s_i)=\left\{
\begin{array}{ll}
\frac{\displaystyle\psi^*(s_i){\it SJ}(s_i)}{\displaystyle\sum_{j=1}^n\psi^*(s_j){\it SJ}(s_j)}, & \mbox{if }s_i\in
DR_{\rm T}(G);\\
0, & \mbox{if }s_i\in DR_{\rm V}(G).
\end{array}
\right.$$
Thus, to calculate $\varphi$, we apply abstraction from self-loops to get ${\bf P}^*$ and then $\psi^*$, followed by
weighting by ${\it SJ}$ and normalization. ${\it EDTMC}(G)$ has no self-loops, unlike ${\it SMC}(G)$, hence, the
behaviour of ${\it EDTMC}(G)$ stabilizes quicker than that of ${\it SMC}(G)$ (if each of them has a single steady
state), since ${\bf P}^*$ has only zero elements at the main diagonal.

\begin{example}
Let $E$ be from Example \ref{ts.exm}. In Figure \ref{exprboxsdtmc.fig}, the underlying SMC ${\it SMC}(\overline{E})$ is
presented. The average sojourn times in the states of the underlying SMC are written next to them in bold font. The
average sojourn time vector of $\overline{E}$ is
${\it SJ}=\left(\frac{1}{\rho},\frac{1}{\chi},0,\frac{1}{\theta},\frac{1}{\phi}\right)$.\\
The sojourn time variance vector of $\overline{E}$ is
${\it VAR}=\left(\frac{1-\rho}{\rho^2},\frac{1-\chi}{\chi^2},0,\frac{1-\theta}{\theta^2},
\frac{1-\phi}{\phi^2}\right)$.\\
The TPM for ${\it EDTMC}(\overline{E})$ is
${\bf P}^*=\left(\begin{array}{ccccc}
0 & 1 & 0 & 0 & 0\\
0 & 0 & 1 & 0 & 0\\
0 & 0 & 0 & \frac{l}{l+m} & \frac{m}{l+m}\\
0 & 1 & 0 & 0 & 0\\
0 & 1 & 0 & 0 & 0
\end{array}\right)$.\\
The steady-state PMF for ${\it EDTMC}(\overline{E})$ is
$\psi^*=\left(0,\frac{1}{3},\frac{1}{3},\frac{l}{3(l+m)},\frac{m}{3(l+m)}\right)$.\\
The steady-state PMF $\psi^*$ weighted by ${\it SJ}$ is
$\left(0,\frac{1}{3\chi},0,\frac{l}{3\theta (l+m)},\frac{m}{3\phi (l+m)}\right)$.\\
It remains to normalize the steady-state weighted PMF, dividing it by the sum of its components\\
$\psi^*{\it SJ}^{\rm T}=\frac{\theta\phi (l+m)+\chi (\phi l+\theta m)}{3\chi\theta\phi (l+m)}$.\\
The steady-state PMF for ${\it SMC}(\overline{E})$ is $\varphi=\frac{1}{\theta\phi (l+m)+\chi (\phi l+\theta
m)}(0,\theta\phi (l+m),0,\chi\phi l,\chi\theta m)$.
\label{exprsmc.exm}
\end{example}

Let $G$ be a dynamic expression and $s,\tilde{s}\in DR(G),\ S,\widetilde{S}\subseteq DR(G)$. The following standard
{\em performance indices (measures)} can be calculated based on the steady-state PMF $\varphi$ for ${\it SMC}(G)$
and the average sojourn time vector ${\it SJ}$ of $G$ \cite{MAS85,Kat96}.
\begin{itemize}

\item The {\em average recurrence (return) time in the state $s$} is $\frac{1}{\varphi (s)}$.

\item The {\em fraction of residence time in the state $s$} is $\varphi (s)$.

\item The {\em fraction of residence time in the set of states $S$} or the {\em probability of the event determined by
a condition that is true for all states from $S$} is $\sum_{s\in S}\varphi (s)$.

\item The {\em relative fraction of residence time in the set of states $S$ w.r.t. that in
$\widetilde{S}$} is $\frac{\sum_{s\in S}\varphi (s)}{\sum_{\tilde{s}\in\widetilde{S}}\varphi (\tilde{s})}$.

\item The {\em rate of leaving the state $s$} is $\frac{\varphi (s)}{{\it SJ}(s)}$.

\item The {\em steady-state probability to perform a step with a multiset of activities $\Xi$} is\\
$\sum_{s\in DR(G)}\varphi (s)\sum_{\{\Upsilon\mid\Xi\subseteq\Upsilon\}}PT(\Upsilon ,s)$.

\item The {\em probability of the event determined by a reward function $r$ on the states} is $\sum_{s\in
DR(G)}\varphi (s)r(s)$, where $\forall s\in DR(G),\ 0\leq r(s)\leq 1$.

\end{itemize}

Let $N=(P_N,T_N,W_N,\Omega_N,\mathcal{L}_N,M_N)$ be a LDTSIPN and $M,\widetilde{M}\in\naturals_{\rm fin}^{P_N}$. Then
the average sojourn time ${\it SJ}(M)$, the sojourn time variance ${\it VAR}(M)$, the probabilities
$PM^*(M,\widetilde{M})$, the transition relation $M\doublera_\mathcal{P}\widetilde{M}$, the EDTMC ${\it EDTMC}(N)$, the
underlying SMC ${\it SMC}(N)$ and the steady-state PMF for it are defined like the corresponding notions for dynamic
expressions. Since every marked plain dtsi-box could be interpreted as the LDTSIPN, we can evaluate performance with
the LDTSIPNs corresponding to dtsi-boxes and then transfer the results to the latter.

Let $\simeq$ denote isomorphism between SMCs that binds their initial states, where two SMCs are isomorphic if their
EDTMCs are so and the sojourn times in the isomorphic states are identically distributed.

\begin{proposition}
For any static expression $E,\ {\it SMC}(\overline{E})\simeq{\it SMC}(Box_{\rm dtsi}(\overline{E}))$.
\label{smcs.pro}
\end{proposition}
\begin{proof}
By Theorem \ref{opdensem.the}, definitions of underlying SMCs for dynamic expressions and LDTSIPNs, and by the
following. First, for the associated SMCs, the average sojourn time in the states is the same, since it is defined via
the analogous probability functions. Second, the transition probabilities of the associated SMCs are the sums of those
belonging to transition systems or reachability graphs.
\end{proof}

\begin{example}
Let $E$ be from Example \ref{ts.exm}. In Figure \ref{exprboxsdtmc.fig}, the underlying SMC ${\it SMC}(N)$ is presented.
Clearly, ${\it SMC}(\overline{E})$ and ${\it SMC}(N)$ are isomorphic.
\label{boxsmc.exm}
\end{example}

\subsection{Analysis of the DTMC}

Let us consider an alternative solution method, studying the DTMCs of expressions ba\-sed on the state change
probabilities $PM(s,\tilde{s})$.

\begin{definition}
Let $G$ be a dynamic expression. The {\em discrete time Markov chain (DTMC)} of $G$, denoted by ${\it DTMC}(G)$, has
the state space $DR(G)$, the initial state $[G]_\approx$ and the transitions $s\rightarrow_\mathcal{P}\tilde{s}$, where
$\mathcal{P}=PM(s,\tilde{s})$.
\end{definition}

\noindent One can see that ${\it EDTMC}(G)$ is constructed from ${\it DTMC}(G)$ as follows. For each state of ${\it
DTMC}(G)$, we remove a possible self-loop associated with it and then normalize the probabilities of the remaining
transitions from the state. Thus, ${\it EDTMC}(G)$ and ${\it DTMC}(G)$ differ only by existence of self-loops and
magnitudes of the probabilities of the remaining transitions. Hence, ${\it EDTMC}(G)$ and ${\it DTMC}(G)$ have the same
communication classes of states and ${\it EDTMC}(G)$ is irreducible iff ${\it DTMC}(G)$ is so. Since both ${\it
EDTMC}(G)$ and ${\it DTMC}(G)$ are finite, they are positive recurrent. Thus, in case of irreducibility, each of them
has a single stationary PMF. Note that ${\it EDTMC}(G)$ and/or ${\it DTMC}(G)$ may be periodic, thus having a unique
stationary distribution, but no steady-state (limiting) one. For example, it may happen that ${\it EDTMC}(G)$ is
periodic while ${\it DTMC}(G)$ is aperiodic due to self-loops associated with some states of the latter. The states of
${\it SMC}(G)$ are classified using ${\it EDTMC}(G)$, hence, ${\it SMC}(G)$ is irreducible (positive recurrent,
aperiodic) iff ${\it EDTMC}(G)$ is so.

Let $G$ be a dynamic expression. The elements $\mathcal{P}_{ij}\ (1\leq i,j\leq n=|DR(G)|)$ of (one-step) transition
probability matrix (TPM) ${\bf P}$ for ${\it DTMC}(G)$ are defined as
$\mathcal{P}_{ij}=\left\{
\begin{array}{ll}
PM(s_i,s_j), & \mbox{if }s_i\rightarrow s_j;\\
0, & \mbox{otherwise}.
\end{array}
\right.$

The steady-state PMF $\psi$ for ${\it DTMC}(G)$ is defined like the corresponding notion for ${\it EDTMC}(G)$. Let us
determine a relationship between steady-state PMFs for ${\it DTMC}(G)$ and ${\it EDTMC}(G)$. The theorem below proposes
the required equation.

Let us introduce a helpful notation. For a vector $v=(v_1,\ldots ,v_n)$, let ${\bf Diag}(v)$ be a diagonal matrix of
order $n$ with the elements $Diag_{ij}(v)\ (1\leq i,j\leq n)$ defined as $Diag_{ij}(v)=\left\{
\begin{array}{ll}
v_i, & \mbox{if }i=j;\\
0, & \mbox{otherwise}.
\end{array}
\right.$

\begin{theorem}
Let $G$ be a dynamic expression and ${\it SL}$ be its self-loops abstraction vector. Then the steady-state PMFs
$\psi$ for ${\it DTMC}(G)$ and $\psi^*$ for ${\it EDTMC}(G)$ are related as follows: $\forall s\in DR(G)$,
$$\psi (s)=\frac{\psi^*(s){\it SL}(s)}{\displaystyle\sum_{\tilde{s}\in DR(G)}\psi^*(\tilde{s}){\it SL}(\tilde{s})}.$$
\label{pmfsim.the}
\end{theorem}
\begin{proof}
Let ${\it PSL}$ be a vector with the elements ${\it PSL}(s)\!=\!\left\{\!
\begin{array}{ll}
PM(s,s), & \mbox{if }s\rightarrow s;\\
0, & \mbox{otherwise}.
\end{array}\right.$
By definition of $PM^*(s,\tilde{s})$, we have ${\bf P}^*={\bf Diag}({\it SL})({\bf P}-{\bf Diag}({\it PSL}))$. Further,
$\psi^*({\bf P}^*-{\bf I})={\bf 0}$ and $\psi^*{\bf P}^*=\psi^*$. After replacement of ${\bf P}^*$ by ${\bf Diag}({\it
SL})({\bf P}-{\bf Diag}({\it PSL}))$ we obtain $\psi^*{\bf Diag}({\it SL})({\bf P}-{\bf Diag}({\it PSL}))=\psi^*$ and\
$\psi^*{\bf Diag}({\it SL}){\bf P}=\psi^*({\bf Diag}({\it SL}){\bf Diag}({\it PSL})+{\bf I})$. Note that $\forall s\in
DR(G)$, we have\\
${\it SL}(s){\it PSL}(s)+1=
\left\{\begin{array}{ll}
{\it SL}(s)PM(s,s)+1=\frac{PM(s,s)}{1-PM(s,s)}+1=\frac{1}{1-PM(s,s)}, & \mbox{if }s\rightarrow s;\\
{\it SL}(s)\cdot 0+1=1, & \mbox{otherwise};
\end{array}\right\}={\it SL}(s)$. Hence, ${\bf Diag}({\it SL}){\bf Diag}({\it PSL})+{\bf I}={\bf Diag}({\it SL})$.
Thus, $\psi^*{\bf Diag}({\it SL}){\bf P}=\psi^*{\bf Diag}({\it SL})$. Then, for $v=\psi^*{\bf Diag}({\it SL})$, we have
$v{\bf P}=v$ and $v({\bf P}-{\bf I})={\bf 0}$. In order to calculate $\psi$ on the basis of $v$, we must normalize it,
dividing its elements by their sum, since we should have $\psi{\bf 1}^{\rm T}=1$ as a result: $\psi=\frac{1}{v{\bf
1}^{\rm T}}v=\frac{1}{\psi^*{\bf Diag}({\it SL}){\bf 1}^{\rm T}}\psi^*{\bf Diag}({\it SL})$. Thus, the elements of
$\psi$ are calculated as follows: $\forall s\in DR(G),\ \psi (s)=\frac{\psi^*(s){\it SL}(s)}{\sum_{\tilde{s}\in
DR(G)}\psi^*(\tilde{s}){\it SL}(\tilde{s})}$. Then $\psi$ is a solution of the equation system $\left\{\begin{array}{l}
\psi({\bf P}-{\bf I})={\bf 0}\\
\psi{\bf 1}^{\rm T}=1
\end{array}
\right.$,
hence, it is the steady-state PMF for ${\it DTMC}(G)$.
\end{proof}

The next proposition relates the steady-state PMFs for ${\it SMC}(G)$ and ${\it DTMC}(G)$.

\begin{proposition}
Let $G$ be a dynamic expression, $\varphi$ be the steady-state PMF for ${\it SMC}(G)$ and $\psi$ be the
steady-state PMF for ${\it DTMC}(G)$. Then $\forall s\in DR(G)$,
$$\varphi (s)=\left\{
\begin{array}{ll}
\frac{\displaystyle\psi (s)}{\displaystyle\sum_{\tilde{s}\in DR_{\rm T}(G)}\psi (\tilde{s})}, & \mbox{if }s\in
DR_{\rm T}(G);\\
0, & \mbox{if }s\in DR_{\rm V}(G).
\end{array}
\right.$$
\label{pmfsmc.pro}
\end{proposition}
\begin{proof}
Let $s\in DR_{\rm T}(G)$. Remember that $\forall s\in DR_{\rm T}(G),\ {\it SL}(s)={\it SJ}(s)$ and $\forall s\in
DR_{\rm V}(G),\ {\it SJ}(s)=0$. Then, by Theorem \ref{pmfsim.the},
$\frac{\psi (s)}{\sum_{\tilde{s}\in DR_{\rm T}(G)}\psi (\tilde{s})}=\frac{\frac{\psi^*(s){\it
SL}(s)}{\sum_{\tilde{s}\in DR(G)}\psi^*(\tilde{s}){\it SL}(\tilde{s})}}{\sum_{\tilde{s}\in DR_{\rm
T}(G)}\left(\frac{\psi^*(\tilde{s}){\it SL}(\tilde{s})}{\sum_{\breve{s}\in DR(G)}\psi^*(\breve{s}){\it
SL}(\breve{s})}\right)}=\frac{\psi^*(s){\it SL}(s)}{\sum_{\tilde{s}\in
DR(G)}\psi^*(\tilde{s}){\it SL}(\tilde{s})}\cdot\\
\frac{\sum_{\breve{s}\in DR(G)}\psi^*(\breve{s}){\it SL}(\breve{s})}{\sum_{\tilde{s}\in DR_{\rm T}(G)}\psi^*(\tilde{s})
{\it SL}(\tilde{s})}\!=\!\frac{\psi^*(s){\it SL}(s)}{\sum_{\tilde{s}\in DR_{\rm T}(G)}\psi^*(\tilde{s}){\it
SL}(\tilde{s})}\!=\!\frac{\psi^*(s){\it SJ}(s)}{\sum_{\tilde{s}\in DR_{\rm T}(G)}\psi^*(\tilde{s}){\it
SJ}(\tilde{s})}\!=\!\frac{\psi^*(s){\it SJ}(s)}{\sum_{\tilde{s}\in DR(G)}\psi^*(\tilde{s}){\it
SJ}(\tilde{s})}\!=\!\varphi (s)$.
\end{proof}

Thus, to calculate $\varphi$, one can only apply normalization to some elements of $\psi$ (corresponding to the
tangible states), instead of abstracting from self-loops to get ${\bf P}^*$ and then $\psi^*$, followed by weighting
by ${\it SJ}$ and normalization. Hence, using ${\it DTMC}(G)$ instead of ${\it EDTMC}(G)$ allows one to avoid
multistage analysis, but the payment for it is more time-consuming numerical and more complex analytical calculation of
$\psi$ w.r.t. $\psi^*$. The reason is that ${\it DTMC}(G)$ has self-loops, unlike ${\it EDTMC}(G)$, hence, the
behaviour of ${\it DTMC}(G)$ stabilizes slower than that of ${\it EDTMC}(G)$ (if each of them has a single steady
state) and ${\bf P}$ is
denser matrix than ${\bf P}^*$, since ${\bf P}$ may additionally have non-zero elements at the main diagonal.
Nevertheless, Proposition \ref{pmfsmc.pro} is very important, since the relationship between $\varphi$ and $\psi$ it
discovers will be used in Section \ref{stationary.sec} to prove preservation of the stationary behaviour by a
stochastic equivalence.

\begin{example}
Let $E$ be from Example \ref{ts.exm}. In Figure \ref{exprboxsdtmc.fig}, the DTMC ${\it DTMC}(\overline{E})$ is
presented.\\
The TPM for ${\it DTMC}(\overline{E})$ is
${\bf P}=\left(\begin{array}{ccccc}
1-\rho & \rho & 0 & 0 & 0\\
0 & 1-\chi & \chi & 0 & 0\\
0 & 0 & 0 & \frac{l}{l+m} & \frac{m}{l+m}\\
0 & \theta & 0 & 1-\theta & 0\\
0 & \phi & 0 & 0 & 1-\phi
\end{array}\right)$.\\
The steady-state PMF for ${\it DTMC}(\overline{E})$ is $\psi\!=\!\frac{1}{\theta\phi(1+\chi )(l+m)+\chi (\phi
l+\theta m)}(0,\!\theta\phi (l+m),\!\chi\theta\phi (l+m),\!\chi\phi l,\!\chi\theta m)$.\\
Since $DR_{\rm T}(\overline{E})=\{s_1,s_2,s_4,s_5\}$ and $DR_{\rm V}(\overline{E})=\{s_3\}$, we have\\
$\sum_{\tilde{s}\in DR_{\rm T}(\overline{E})}\psi (\tilde{s})=\psi (s_1)+\psi (s_2)+\psi (s_4)+\psi (s_5)=
\frac{\theta\phi(l+m)+\chi (\phi l+\theta m)}{\theta\phi(1+\chi )(l+m)+\chi (\phi l+\theta m)}$. By Proposition
\ref{pmfsmc.pro},\\
$\begin{array}{l}

\varphi (s_1)=0\cdot\frac{\theta\phi(1+\chi )(l+m)+\chi (\phi l+\theta m)}{\theta\phi(l+m)+\chi (\phi l+\theta
m)}=0,\\[1mm]

\varphi (s_2)=\frac{\theta\phi (l+m)}{\theta\phi(1+\chi )(l+m)+\chi (\phi l+\theta m)}\cdot
\frac{\theta\phi(1+\chi )(l+m)+\chi (\phi l+\theta m)}{\theta\phi(l+m)+\chi (\phi l+\theta m)}=
\frac{\theta\phi (l+m)}{\theta\phi(l+m)+\chi (\phi l+\theta m)},\\[1mm]

\varphi (s_3)=0,\\[1mm]

\varphi (s_4)=\frac{\chi\phi l}{\theta\phi(1+\chi )(l+m)+\chi (\phi l+\theta m)}\cdot
\frac{\theta\phi(1+\chi )(l+m)+\chi (\phi l+\theta m)}{\theta\phi(l+m)+\chi (\phi l+\theta m)}=
\frac{\chi\phi l}{\theta\phi(l+m)+\chi (\phi l+\theta m)},\\[1mm]

\varphi (s_5)=\frac{\chi\theta m}{\theta\phi(1+\chi )(l+m)+\chi (\phi l+\theta m)}\cdot
\frac{\theta\phi(1+\chi )(l+m)+\chi (\phi l+\theta m)}{\theta\phi(l+m)+\chi (\phi l+\theta m)}=
\frac{\chi\theta m}{\theta\phi(l+m)+\chi (\phi l+\theta m)}.

\end{array}$\\
Thus, the steady-state PMF for ${\it SMC}(\overline{E})$ is $\varphi=\frac{1}{\theta\phi (l+m)+\chi (\phi l+\theta
m)}(0,\theta\phi (l+m),0,\chi\phi l,\chi\theta m)$. This coincides with the result obtained in Example
\ref{exprsmc.exm} with the use of $\psi^*$ and ${\it SJ}$.
\label{exprdtmc.exm}
\end{example}

\begin{figure}
\begin{center}
\includegraphics[scale=0.9]{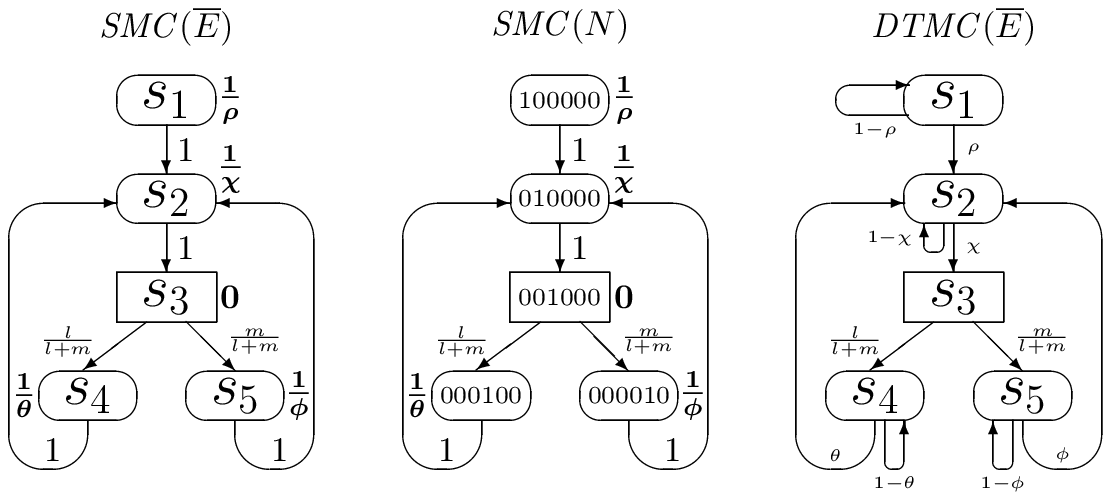}
\end{center}
\vspace{-7mm}
\caption{The underlying SMCs of $\overline{E}$ and $N=Box_{\rm dtsi}(\overline{E})$ and DTMC of $\overline{E}$ for
$E=[(\{a\},\rho )*((\{b\},\chi );\protect\newline
(((\{c\},\natural_l);(\{d\},\theta ))\cho ((\{e\},\natural_m);(\{f\},\phi ))))*{\sf Stop}]$.}
\label{exprboxsdtmc.fig}
\end{figure}

\section{Stochastic equivalences}
\label{stocheqs.sec}

Consider the expressions $E=(\{a\},\frac{1}{2})$ and $E'=(\{a\},\frac{1}{3})_1\cho (\{a\},\frac{1}{3})_2$, for which
$\overline{E}\neq_{\rm ts}\overline{E'}$, since $TS(\overline{E})$ has only one transition from the initial to the
final state (with probability $\frac{1}{2}$) while $TS(\overline{E'})$ has two such ones (with probabilities
$\frac{1}{4}$). On the other hand, all the mentioned transitions are labeled by activities with the same multiaction
part $\{a\}$. Next, the overall probabilities of the mentioned transitions of $TS(\overline{E})$ and
$TS(\overline{E'})$ coincide: $\frac{1}{2}=\frac{1}{4}+\frac{1}{4}$. Further, $TS(\overline{E})$ (as well as
$TS(\overline{E'})$) has one empty loop transition from the initial state to itself with probability $\frac{1}{2}$ and
one empty loop transition from the final state to itself with probability $1$. The empty loop transitions are labeled
by the empty set of activities. For calculating the transition probabilities of $TS(\overline{E'})$, take $\rho =\chi
=\frac{1}{3}$ in Example \ref{trprob.exm}. Then you will see that the probability parts $\frac{1}{3}$ and $\frac{1}{3}$
of the activities $(\{a\},\frac{1}{3})_1$ and $(\{a\},\frac{1}{3})_2$ are ``splitted'' among probabilities
$\frac{1}{4}$ and $\frac{1}{4}$ of the corresponding transitions and the probability $\frac{1}{2}$ of the empty loop
transition. Unlike $=_{\rm ts}$, most of the probabilistic and stochastic equivalences proposed in the literature do
not differentiate between the processes such as those specified by $E$ and $E'$. In Figure \ref{exmsteqimm.fig}(a), the
marked dtsi-boxes corresponding to the dynamic expressions $\overline{E}$ and $\overline{E'}$ are presented, i.e.
$N=Box_{dtsi}(\overline{E})$ and $N'=Box_{dtsi}(\overline{E'})$.

Since the semantic equivalence $=_{\rm ts}$ is too discriminating in many cases, we need weaker equivalence notions.
These equivalences should possess the following necessary properties. First, any two equivalent processes must have the
same sequences of multisets of multiactions, which are the multiaction parts of the activities executed in steps
starting from the initial states of the processes. Second, for every such sequence, its execution probabilities within
both processes must coincide. Third, the desired equivalence should preserve the branching structure of computations,
i.e. the points of choice of an external observer between several extensions of a particular computation should be
taken into account. In this section, we define one such notion: step stochastic bisimulation equivalence.

\subsection{Step stochastic bisimulation equivalence}

Bisimulation equivalences respect the particular points of choice in the behaviour of a system. To define stochastic
bisimulation equivalences, we
consider a bisimulation as an {\em equivalence} relation that partitions the states of the {\em union} of the
transition systems $TS(G)$ and $TS(G')$ of two dynamic expressions $G$ and $G'$ to be compared. For $G$ and $G'$ to be
bisimulation equivalent, the initial states $[G]_\approx$ and $[G']_\approx$ of their transition systems should be
related by a bisimulation having the following transfer property: if two states are related then in each of them the
same multisets of multiactions can occur, leading with the iden\-ti\-cal overall probability from each of the two
states to {\em the same equivalence class} for every such multiset.

Thus, we follow the approaches of \cite{JS90,LS91,HR94,Hil96,BGo98,Bern07,Bern15}, but we implement step semantics
instead of interleaving one considered in these papers. We use the generative probabilistic transition systems, like in
\cite{JS90}, in contrast to the reactive model, treated in \cite{LS91}, and we take transition probabilities instead of
transition rates from \cite{HR94,Hil96,BGo98,Bern07,Bern15}. Hence,
step stochastic bisimulation equivalence that define further is (in
a probability sense) comparable only with interleaving probabilistic bisimulation one from \cite{JS90}, and our
equivalence is obviously stronger.

In the definition below, we consider $\mathcal{L}(\Upsilon )\in\naturals_{\rm fin}^\mathcal{L}$ for
$\Upsilon\in\naturals_{\rm fin}^\mathcal{SIL}$, i.e. (possibly empty) multisets of multiactions. The multiactions can
be empty as well. In this case, $\mathcal{L}(\Upsilon )$ contains the elements $\emptyset$, but it is not empty itself.

Let $G$ be a dynamic expression and $\mathcal{H}\subseteq DR(G)$. For any $s\in DR(G)$ and $A\in\naturals_{\rm
fin}^\mathcal{L}$, we write $s\stackrel{A}{\rightarrow}_\mathcal{P}\mathcal{H}$, where
$\mathcal{P}=PM_A(s,\mathcal{H})$ is the {\em overall probability to move from $s$ into the set of states $\mathcal{H}$
via steps with the multiaction part $A$} defined as
$$PM_A(s,\mathcal{H})=\sum_{\{\Upsilon\mid\exists\tilde{s}\in\mathcal{H},\ s\stackrel{\Upsilon}{\rightarrow}\tilde{s},\
\mathcal{L}(\Upsilon )=A\}}PT(\Upsilon ,s).$$

We write $s\stackrel{A}{\rightarrow}\mathcal{H}$ if $\exists\mathcal{P},\
s\stackrel{A}{\rightarrow}_\mathcal{P}\mathcal{H}$. Further, we write $s\rightarrow_\mathcal{P}\mathcal{H}$ if $\exists
A,\ s\stackrel{A}{\rightarrow}\mathcal{H}$, where $\mathcal{P}=PM(s,\mathcal{H})$ is the {\em overall probability to
move from $s$ into the set of states $\mathcal{H}$ via any steps} defined as
$$PM(s,\mathcal{H})=\sum_{\{\Upsilon\mid\exists\tilde{s}\in\mathcal{H},\ s\stackrel{\Upsilon}{\rightarrow}\tilde{s}\}}
PT(\Upsilon ,s).$$

To introduce a stochastic bisimulation between dynamic expressions $G$ and $G'$, we should consider the ``composite''
set of states $DR(G)\cup DR(G')$, since we have to identify the probabilities to come from any two equivalent states
into the same ``composite'' equivalence class (w.r.t. the stochastic bisimulation). For $G\neq G'$, transitions
starting from the states of $DR(G)$ (or $DR(G')$) always lead to those from the same set, since $DR(G)\cap
DR(G')=\emptyset$, allowing us to ``mix'' the sets of states in the definition of stochastic bisimulation.

\begin{definition}
Let $G$ and $G'$ be dynamic expressions. An {\em equivalence} relation $\mathcal{R}\subseteq (DR(G)\cup DR(G'))^2$ is a
{\em step stochastic bisimulation} between $G$ and $G'$, denoted by $\mathcal{R}:G\bis_{\rm ss}G'$, if:
\begin{enumerate}

\item $([G]_\approx ,[G']_\approx )\in\mathcal{R}$.

\item $(s_1,s_2)\in\mathcal{R}\ \Rightarrow\ \forall\mathcal{H}\in (DR(G)\cup DR(G'))/_\mathcal{R},\ \forall
A\in\naturals_{\rm fin}^\mathcal{L},\ s_1\stackrel{A}{\rightarrow}_\mathcal{P}\mathcal{H}\ \Leftrightarrow\
s_2\stackrel{A}{\rightarrow}_\mathcal{P}\mathcal{H}$.

\end{enumerate}
Two dynamic expressions $G$ and $G'$ are {\em step stochastic bisimulation equivalent}, denoted by $G\bis_{\rm ss}G'$,
if $\exists\mathcal{R}:G\bis_{\rm ss}G'$.
\end{definition}

The following proposition states that every step stochastic bisimulation binds tangible states only with tangible ones
and the same is valid for vanishing states.

\begin{proposition}
Let $G$ and $G'$ be dynamic expressions and $\mathcal{R}:G\bis_{\rm ss}G'$. Then
$$\mathcal{R}\subseteq (DR_{\rm T}(G)\cup DR_{\rm T}(G'))^2\uplus (DR_{\rm V}(G)\cup DR_{\rm V}(G'))^2.$$
\label{bissplit.pro}
\end{proposition}
\begin{proof}
By definition of transition systems of expressions, for every tangible state, there is an empty loop from it, and no
empty loop transitions are possible from vanishing states. Further, $\mathcal{R}$ preserves empty loops. To verify
this, first take $A=\emptyset$ in its definition to get $\forall (s_1,s_2)\in\mathcal{R},\ \forall\mathcal{H}\in
(DR(G)\cup DR(G'))/_\mathcal{R},\ s_1\stackrel{\emptyset}{\rightarrow}_\mathcal{P}\mathcal{H}\ \Leftrightarrow\\
s_2\stackrel{\emptyset}{\rightarrow}_\mathcal{P}\mathcal{H}$, and then observe that the empty loop transition from a
state leads only to the same state.
\end{proof}

Let $\mathcal{R}_{\rm ss}(G,G')=\bigcup\{\mathcal{R}\mid\mathcal{R}:G\bis_{\rm ss}G'\}$ be the {\em union of all step
stochastic bisimulations} between $G$ and $G'$. The following proposition proves that $\mathcal{R}_{\rm ss}(G,G')$ is
also an {\em equivalence} and $\mathcal{R}_{\rm ss}(G,G'):G\bis_{\rm ss}G'$.

\begin{proposition}
Let $G$ and $G'$ be dynamic expressions and $G\bis_{\rm ss}G'$. Then $\mathcal{R}_{\rm ss}(G,G')$ is the largest step
stochastic bisimulation between $G$ and $G'$.
\label{largestbisim.pro}
\end{proposition}
\begin{proof}
See Appendix \ref{largestbisim.ssc}.
\end{proof}

In \cite{Bai96}, an algorithm for strong probabilistic bisimulation on labeled probabilistic transition systems (a
reformulation of probabilistic automata) was proposed with time complexity $O(n^2 m)$, where $n$ is the number of
states and $m$ is the number of transitions. In \cite{BEMC00}, a decision algorithm for strong probabilistic
bisimulation on generative labeled probabilistic transition systems was constructed with time complexity $O(m\log n)$
and space complexity $O(m+n)$. In \cite{CSe02}, a polynomial algorithm for strong probabilistic bisimulation on
probabilistic automata was presented. The mentioned algorithms for interleaving probabilistic bisimulation equivalence
can be adapted for $\bis_{ss}$ using the method from \cite{JM96}, applied to get the decidability results for step
bisimulation equivalence. The method respects that transition systems in interleaving and step semantics differ only by
availability of the additional transitions corresponding to parallel execution of activities in the latter (which is
our case).

\subsection{Interrelations of the stochastic equivalences}

We now compare the discrimination power of the stochastic equivalences.

\begin{theorem}
For dynamic expressions $G,\ G'$ the next {\em strict} implications hold:
$$G\approx G'\ \Rightarrow\ G=_{\rm ts}G'\ \Rightarrow\ G\bis_{\rm ss}G'.$$
\label{intsteqim.the}
\end{theorem}
\begin{proof}
Let us check the validity of the implications.
\begin{itemize}

\item The implication $=_{\rm ts}\Rightarrow\bis_{\rm ss}$ is proved as follows. Let $\beta :G=_{\rm ts}G'$. Then it is
easy to see that $\mathcal{R}:G\bis_{\rm ss}G'$, where $\mathcal{R}=\{(s,\beta (s))\mid s\in DR(G)\}$.

\item The implication $\approx\Rightarrow =_{\rm ts}$ is valid, since the transition system of a dynamic formula is
defined based on its structural equivalence class.

\end{itemize}

Let us see that that the implications are strict, by the following counterexamples.
\begin{itemize}

\item[(a)] Let $E=(\{a\},\frac{1}{2})$ and $E'=(\{a\},\frac{1}{3})_1\cho (\{a\},\frac{1}{3})_2$. Then
$\overline{E}\bis_{\rm ss}\overline{E'}$, but $\overline{E}\neq_{\rm ts}\overline{E'}$, since $TS(\overline{E})$
has only one transition from the initial to the final state while $TS(\overline{E'})$ has two such ones.

\item[(b)] Let $E=(\{a\},\frac{1}{2});(\{\hat{a}\},\frac{1}{2})$ and $E'=((\{a\},\frac{1}{2});
(\{\hat{a}\},\frac{1}{2}))\sy a$. Then $\overline{E}=_{\rm ts}\overline{E'}$, but
$\overline{E}\not\approx\overline{E'}$, since $\overline{E}$ and $\overline{E'}$ cannot be reached from each other
by inaction rules.

\end{itemize}
\end{proof}

\begin{example}
In Figure \ref{exmsteqimm.fig}, the marked dtsi-boxes corresponding to the dynamic expressions from
examples of Theorem \ref{intsteqim.the} are presented, i.e. $N=Box_{\rm dtsi}(\overline{E})$ and $N'=Box_{\rm
dtsi}(\overline{E'})$ for each picture (a)--(b).
\end{example}

\begin{figure}
\begin{center}
\includegraphics[scale=0.9]{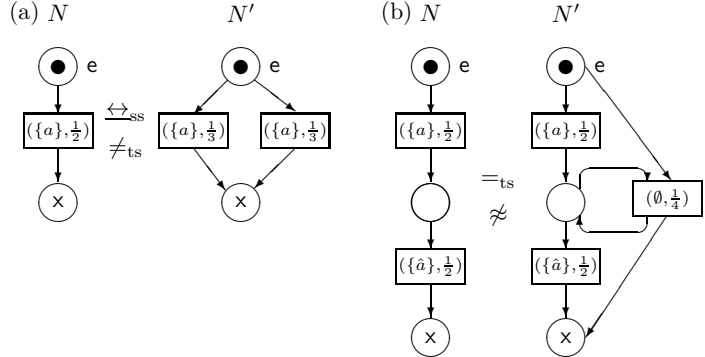}
\end{center}
\vspace{-7mm}
\caption{Dtsi-boxes of the dynamic expressions from equivalence examples of Theorem \ref{intsteqim.the}.}
\label{exmsteqimm.fig}
\end{figure}

\section{Reduction modulo equivalences}
\label{reduction.sec}

The proposed equivalences
can be used to reduce transition systems and SMCs of expressions (reachability graphs and SMCs of dtsi-boxes).
Reductions of graph-based models, like transition systems, reachability graphs and SMCs, result in those with less
states (the graph nodes). The goal of the reduction is to decrease the number of states in the semantic representation
of the modeled system while preserving its important qualitative and quantitative properties. Thus, the reduction
allows one to simplify the behavioural and performance analysis of systems.

An {\em autobisimulation} is a bisimulation between an expression and itself. For a dynamic expression $G$ and a step
stochastic autobisimulation on it $\mathcal{R}:G\bis_{\rm ss}G$, let $\mathcal{K}\in DR(G)/_\mathcal{R}$ and
$s_1,s_2\in\mathcal{K}$. We have $\forall\widetilde{\mathcal{K}}\in DR(G)/_\mathcal{R},\ \forall A\in\naturals_{\rm
fin}^\mathcal{L},\ s_1\stackrel{A}{\rightarrow}_\mathcal{P}\widetilde{\mathcal{K}}\ \Leftrightarrow\
s_2\stackrel{A}{\rightarrow}_\mathcal{P}\widetilde{\mathcal{K}}$. The previous equality is valid for all
$s_1,s_2\in\mathcal{K}$, hence, we can rewrite it as
$\mathcal{K}\stackrel{A}{\rightarrow}_\mathcal{P}\widetilde{\mathcal{K}}$, where
$\mathcal{P}=PM_A(\mathcal{K},\widetilde{\mathcal{K}})=PM_A(s_1,\widetilde{\mathcal{K}})=
PM_A(s_2,\widetilde{\mathcal{K}})$. We write $\mathcal{K}\stackrel{A}{\rightarrow}\widetilde{\mathcal{K}}$ if
$\exists\mathcal{P},\ \mathcal{K}\stackrel{A}{\rightarrow}_\mathcal{P}\widetilde{\mathcal{K}}$ and
$\mathcal{K}\rightarrow\widetilde{\mathcal{K}}$ if $\exists A,\
\mathcal{K}\stackrel{A}{\rightarrow}\widetilde{\mathcal{K}}$. The similar arguments allow us to write
$\mathcal{K}\rightarrow_\mathcal{P}\widetilde{\mathcal{K}}$, where
$\mathcal{P}=PM(\mathcal{K},\widetilde{\mathcal{K}})=PM(s_1,\widetilde{\mathcal{K}})=PM(s_2,\widetilde{\mathcal{K}})$.

By Proposition \ref{bissplit.pro}, $\mathcal{R}\subseteq (DR_{\rm T}(G))^2\uplus (DR_{\rm V}(G))^2$. Hence,
$\forall\mathcal{K}\in DR(G)/_\mathcal{R}$, all states from $\mathcal{K}$ are tangible if $\mathcal{K}\in DR_{\rm
T}(G)/_\mathcal{R}$, or vanishing if $\mathcal{K}\in DR_{\rm V}(G)/_\mathcal{R}$.

The {\em average sojourn time in the equivalence class (w.r.t. $\mathcal{R}$) of states $\mathcal{K}$} is
$${\it SJ}_\mathcal{R}(\mathcal{K})=
\left\{
\begin{array}{ll}
\frac{1}{1-PM(\mathcal{K},\mathcal{K})}, & \mbox{if }\mathcal{K}\in DR_{\rm T}(G)/_\mathcal{R};\\
0, & \mbox{if }\mathcal{K}\in DR_{\rm V}(G)/_\mathcal{R}.
\end{array}
\right.$$
The {\em average sojourn time vector for the equivalence classes (w.r.t. $\mathcal{R}$) of states} ${\it
SJ}_\mathcal{R}$ of $G$ has the elements ${\it SJ}_\mathcal{R}(\mathcal{K}),\ \mathcal{K}\in DR(G)/_\mathcal{R}$.

The {\em sojourn time variance in the equivalence class (w.r.t. $\mathcal{R}$) of states $\mathcal{K}$} is
$${\it VAR}_\mathcal{R}(\mathcal{K})=
\left\{
\begin{array}{ll}
\frac{PM(\mathcal{K},\mathcal{K})}{(1-PM(\mathcal{K},\mathcal{K}))^2}, & \mbox{if }\mathcal{K}\in
DR_{\rm T}(G)/_\mathcal{R};\\
0, & \mbox{if }\mathcal{K}\in DR_{\rm V}(G)/_\mathcal{R}.
\end{array}
\right.$$
The {\em sojourn time variance vector for the equivalence classes (w.r.t. $\mathcal{R}$) of states} ${\it
VAR}_\mathcal{R}$ of $G$ has the elements ${\it VAR}_\mathcal{R}(\mathcal{K}),\ \mathcal{K}\in DR(G)/_\mathcal{R}$.

Let $\mathcal{R}_{\rm ss}(G)=\bigcup\{\mathcal{R}\mid\mathcal{R}:G\bis_{\rm ss}G\}$ be the {\em union of all step
stochastic autobisimulations} on $G$. By Proposition \ref{largestbisim.pro}, $\mathcal{R}_{\rm ss}(G)$ is the largest
step stochastic autobisimulation on $G$. Based on the equivalence classes w.r.t. $\mathcal{R}_{\rm ss}(G)$, the
quotient (by $\bis_{\rm ss}$) transition systems and the quotient (by $\bis_{\rm ss}$) underlying SMCs of expressions
can be defined. The mentioned equivalence classes become the quotient states. The average sojourn time in a quotient
state is that in the corresponding equivalence class. Every quotient transition between two such composite states
represents all steps (having the same multiaction part in case of the transition system quotient) from the first state
to the second one.

\begin{definition}
Let $G$ be a dynamic expression. The {\em quotient (by $\bis_{\rm ss}$) (labeled probabilistic) transition system} of
$G$ is a quadruple $TS_{\bis_{\rm ss}}(G)=(S_{\bis_{\rm ss}},L_{\bis_{\rm ss}},\mathcal{T}_{\bis_{\rm ss}},s_{\bis_{\rm
ss}})$, where
\begin{itemize}

\item $S_{\bis_{\rm ss}}=DR(G)/_{\mathcal{R}_{\rm ss}(G)}$;

\item $L_{\bis_{\rm ss}}=\naturals_{\rm fin}^\mathcal{L}\times (0;1]$;

\item $\mathcal{T}_{\bis_{\rm ss}}=\{(\mathcal{K},(A,PM_A(\mathcal{K},\widetilde{\mathcal{K}})),
\widetilde{\mathcal{K}})\mid\mathcal{K}, \widetilde{\mathcal{K}}\in DR(G)/_{\mathcal{R}_{\rm ss}(G)},\
\mathcal{K}\stackrel{A}{\rightarrow}\widetilde{\mathcal{K}}\}$;

\item $s_{\bis_{\rm ss}}=[[G]_\approx]_{\mathcal{R}_{\rm ss}(G)}$.

\end{itemize}
\end{definition}

The transition $(\mathcal{K},(A,\mathcal{P}),\widetilde{\mathcal{K}})\in\mathcal{T}_{\bis_{\rm ss}}$ will be written as
$\mathcal{K}\stackrel{A}{\rightarrow}_\mathcal{P}\widetilde{\mathcal{K}}$.

\begin{example}
Let $F$ be an abstraction of the static expression $E$ from Example \ref{ts.exm}, with $c=e,\ d=f,\\
\theta =\phi$, i.e. $F=[(\{a\},\rho )*((\{b\},\chi );(((\{c\},\natural_l);(\{d\},\theta ))\cho
((\{c\},\natural_m);(\{d\},\theta ))))*{\sf Stop}]$. Then $DR(\overline{F})=\{s_1,s_2,s_3,s_4,s_5\}$ is obtained from
$DR(\overline{E})$ via substitution of the symbols $e,\ f,\ \phi$ by $c,\ d,\ \theta$, respectively, in the
specifications of the corresponding states from the latter set. We have $DR_T(\overline{F})=\{s_1,s_2,s_4,s_5\}$ and
$DR_V(\overline{F})=\{s_3\}$. Further, $DR(\overline{F})/_{\mathcal{R}_{\rm
ss}(\overline{F})}=\{\mathcal{K}_1,\mathcal{K}_2,\mathcal{K}_3, \mathcal{K}_4\}$, where $\mathcal{K}_1=\{s_1\},\
\mathcal{K}_2=\{s_2\},\ \mathcal{K}_3=\{s_3\},\ \mathcal{K}_4=\{s_4,s_5\}$. We also have
$DR_T(\overline{F})/_{\mathcal{R}_{\rm ss}(\overline{F})}=\{\mathcal{K}_1,\mathcal{K}_2,\mathcal{K}_4\}$ and
$DR_V(\overline{F})/_{\mathcal{R}_{\rm ss}(\overline{F})}=\{\mathcal{K}_3\}$. In Figure \ref{exprqtssdtmc.fig}, the
quotient transition system $TS_{\bis_{\rm ss}}(\overline{F})$ is presented.
\label{qts.exm}
\end{example}

The {\em quotient (by $\bis_{\rm ss}$) average sojourn time vector} of $G$ is defined as ${\it SJ}_{\bis_{\rm ss}}={\it
SJ}_{\mathcal{R}_{\rm ss}(G)}$. The {\em quotient (by $\bis_{\rm ss}$) sojourn time variance vector} of $G$ is defined
as ${\it VAR}_{\bis_{\rm ss}}={\it VAR}_{\mathcal{R}_{\rm ss}(G)}$.

Let $\mathcal{K}\rightarrow\widetilde{\mathcal{K}}$ and $\mathcal{K}\neq\widetilde{\mathcal{K}}$. The {\em probability
to move from $\mathcal{K}$ to $\widetilde{\mathcal{K}}$ by executing any set of activities after possible self-loops}
is
$$PM^*(\mathcal{K},\widetilde{\mathcal{K}})=\left\{
\begin{array}{ll}
PM(\mathcal{K},\widetilde{\mathcal{K}})\sum_{k=0}^{\infty}PM(\mathcal{K},\mathcal{K})^k=
\frac{PM(\mathcal{K},\widetilde{\mathcal{K}})}{1-PM(\mathcal{K},\mathcal{K})}, &
\mbox{if }\mathcal{K}\rightarrow\mathcal{K};\\
PM(\mathcal{K},\widetilde{\mathcal{K}}), & \mbox{otherwise}.
\end{array}
\right.$$
The value $k=0$ in the summation above corresponds to the case with no self-loops. Note that $\forall \mathcal{K}\in
DR_{\rm T}(G)/_{\mathcal{R}_{\rm ss}(G)},\ PM^*(\mathcal{K},\widetilde{\mathcal{K}})= {\it SJ}_{\bis_{\rm
ss}}(\mathcal{K})PM(\mathcal{K},\widetilde{\mathcal{K}})$, since we always have the empty loop (self-loop)
$\mathcal{K}\stackrel{\emptyset}{\rightarrow}\mathcal{K}$ from every equivalence class of tangible states
$\mathcal{K}$. Empty loops are not possible from equivalence classes of vanishing states, hence, $\forall
\mathcal{K}\in DR_{\rm V}(G)/_{\mathcal{R}_{\rm ss}(G)},\ PM^*(\mathcal{K},\widetilde{\mathcal{K}})=
\frac{PM(\mathcal{K},\widetilde{\mathcal{K}})}{1-PM(\mathcal{K},\mathcal{K})}$, when there are non-empty self-loops
(produced by iteration) from $\mathcal{K}$, or $PM^*(\mathcal{K},\widetilde{\mathcal{K}})=
PM(\mathcal{K},\widetilde{\mathcal{K}})$, when there are no self-loops from $\mathcal{K}$.

\begin{definition}
Let $G$ be a dynamic expression. The {\em quotient (by $\bis_{\rm ss}$) EDTMC} of $G$, denoted by\\
${\it EDTMC}_{\bis_{\rm ss}}(G)$, has the state space $DR(G)/_{\mathcal{R}_{\rm ss}(G)}$, the initial state
$[[G]_\approx]_{\mathcal{R}_{\rm ss}(G)}$ and the transitions\\
$\mathcal{K}\doublera_\mathcal{P}\widetilde{\mathcal{K}}$ if $\mathcal{K}\rightarrow\widetilde{\mathcal{K}}$ and
$\mathcal{K}\neq\widetilde{\mathcal{K}}$, where $\mathcal{P}=PM^*(\mathcal{K},\widetilde{\mathcal{K}})$. The {\em
quotient (by $\bis_{\rm ss}$) underlying SMC} of $G$, denoted by ${\it SMC}_{\bis_{\rm ss}}(G)$, has the EDTMC ${\it
EDTMC}_{\bis_{\rm ss}}(G)$ and the sojourn time in every $\mathcal{K}\in DR_{\rm T}(G)/_{\mathcal{R}_{\rm ss}(G)}$ is
geometrically distributed with the parameter $1-PM(\mathcal{K},\mathcal{K})$ while that in every $\mathcal{K}\in
DR_{\rm V}(G)/_{\mathcal{R}_{\rm ss}(G)}$ is zero.
\end{definition}

The steady-state PMFs $\psi_{\bis_{\rm ss}}^*$ for ${\it EDTMC}_{\bis_{\rm ss}}(G)$ and $\varphi_{\bis_{\rm
ss}}$ for ${\it SMC}_{\bis_{\rm ss}}(G)$ are defined like the corresponding notions $\psi^*$ for ${\it EDTMC}(G)$
and $\varphi$ for ${\it SMC}(G)$.

\begin{example}
Let $F$ be from Example \ref{qts.exm}. In Figure \ref{exprqtssdtmc.fig}, the quotient underlying SMC ${\it
SMC}_{\bis_{\rm ss}}(\overline{F})$ is presented.
\label{exprqsmc.exm}
\end{example}

The quotients of both transition systems and underlying SMCs are their minimal reductions modulo step stochastic
bisimulations. The quotients simplify analysis of system properties, preserved by $\bis_{\rm ss}$, since less states
should be examined for it. Such reduction method resembles that from \cite{AS92}, based on place bisimulation
equivalence for PNs, but the former method merges states, while the latter one merges places.

Moreover, there exist
algorithms
to construct the quotients of transition systems by an equivalence (like bisimulation one) \cite{PT87} and those of
(discrete or continuous time) Markov chains by ordinary lumping \cite{DHS03}.
These algorithms have time complexity $O(m\log n)$ and space complexity $O(m+n)$, where $n$ is the number of states and
$m$ is the number of transitions. As mentioned in \cite{WDH10}, the algorithm from \cite{DHS03} can be easily adjusted
to produce quotients of labeled probabilistic transition systems by the probabilistic bisimulation equivalence. In
\cite{WDH10}, the symbolic partition refinement algorithm on the state space of CTMCs was proposed. The algorithm can
be applied to DTMCs and labeled probabilistic transition systems. Such a symbolic lumping is memory efficient due to
compact representation of the state space partition. It is time efficient, since fast algorithm of the partition
representation and refinement is applied. In \cite{EHSTZ13}, a polynomial-time algorithm for minimizing behaviour of
probabilistic automata by probabilistic bisimulation equivalence was outlined that results in the canonical quotient
structures. One could adapt the above algorithms for our framework.

Let us define quotient (by $\bis_{\rm ss}$) DTMCs of expressions based on probabilities
$PM(\mathcal{K},\widetilde{\mathcal{K}})$.

\begin{definition}
Let $G$ be a dynamic expression. The {\em quotient (by $\bis_{\rm ss}$) DTMC} of $G$, denoted by\\
${\it DTMC}_{\bis_{\rm ss}}(G)$, has the state space $DR(G)/_{\mathcal{R}_{\rm ss}(G)}$, the initial state
$[[G]_\approx]_{\mathcal{R}_{\rm ss}(G)}$ and the transitions
$\mathcal{K}\rightarrow_\mathcal{P}\widetilde{\mathcal{K}}$, where
$\mathcal{P}=PM(\mathcal{K},\widetilde{\mathcal{K}})$.
\end{definition}

\noindent The steady-state PMF $\psi_{\bis_{\rm ss}}$ for ${\it DTMC}_{\bis_{\rm ss}}(G)$ is defined like the
corresponding notion $\psi$ for ${\it DTMC}(G)$.

\begin{example}
Let $F$ be from Example \ref{qts.exm}. In Figure \ref{exprqtssdtmc.fig}, the quotient DTMC
${\it DTMC}_{\bis_{\rm ss}}(\overline{F})$ is presented.
\label{exprqdtmc.exm}
\end{example}

\begin{figure}
\begin{center}
\includegraphics[scale=0.9]{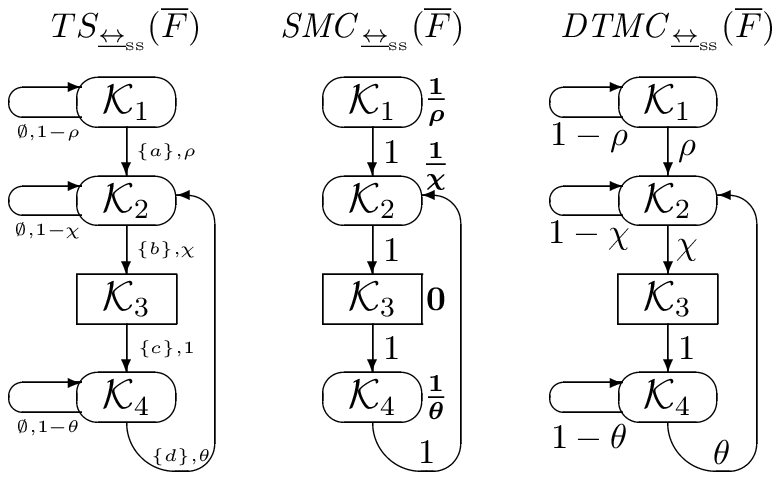}
\end{center}
\vspace{-6mm}
\caption{The quotient transition system, quotient underlying SMC and quotient DTMC of $\overline{F}$ for
$F=[(\{a\},\rho )*\protect\newline
((\{b\},\chi );(((\{c\},\natural_l);(\{d\},\theta ))\cho ((\{c\},\natural_m);(\{d\},\theta ))))*{\sf Stop}]$.}
\label{exprqtssdtmc.fig}
\end{figure}

Clearly, the relationships between the steady-state PMFs $\psi_{\bis_{\rm ss}}$ and $\psi_{\bis_{\rm ss}}^*$, as well
as $\varphi_{\bis_{\rm ss}}$ and $\psi_{\bis_{\rm ss}}$, are the same as those
between their
``non-quotient'' versions in Theorem \ref{pmfsim.the} and Proposition \ref{pmfsmc.pro}.

The detailed illustrative quotient example will be presented in Section \ref{gshmsysim.sec}.

In \cite{Buc94b}, it is proven that irreducibility is preserved by aggregation w.r.t. any partition (or equivalence
relation) on the states of finite DTMCs (so they are also positive recurrent). Aggregation decreases the number of
states, hence, the aggregated DTMCs are also finite and positive recurrence is preserved by every aggregation. It is
known \cite{Ros96,Kul09} that irreducible and positive recurrent DTMCs have a single stationary PMF. Note that the
original and/or aggregated DTMCs may be periodic, thus having a unique stationary distribution, but no steady-state
(limiting) one. For example, it may happen that the original DTMC is aperiodic while the aggregated DTMC is periodic
due to merging some states of the former. Thus, both finite irreducible DTMCs and their arbitrary aggregates have a
single stationary PMF. It is also shown in \cite{Buc94b} that for every DTMC aggregated by ordinary lumpability, the
stationary probability of each aggregate state is a sum of the stationary probabilities of all its constituent states
from the original DTMC. The information about individual stationary probabilities of the original DTMC is lost after
such a summation, but in many cases, the stationary probabilities of the aggregated DTMC are enough to calculate
performance measures of the high-level model, from which the original DTMC is extracted. As mentioned in \cite{Buc94b},
in some applications, the aggregated DTMC can be extracted directly from the high-level model. Thus, the aggregation
techniques based on lumping are of practical importance, since they allow one to reduce the state space of the modeled
systems, hence, the computational costs for evaluating their performance.

Let $G$ be a dynamic expression. By definition of $\bis_{\rm ss}$, the relation $\mathcal{R}_{\rm ss}(G)$ on $TS(G)$
induces ordinary lumping on ${\it SMC}(G)$, i.e. if the states of $TS(G)$ are related by $\mathcal{R}_{\rm ss}(G)$ then
the same states in ${\it SMC}(G)$ are related by ordinary lumping. The quotient (maximal aggregate) of ${\it SMC}(G)$
by such an induced ordinary lumping is ${\it SMC}_{\bis_{\rm ss}}(G)$. Since we consider only finite SMCs,
irreducibility of ${\it SMC}(G)$ will imply irreducibility of ${\it SMC}_{\bis_{\rm ss}}(G)$ and they are positive
recurrent. Then a unique quotient stationary PMF of ${\it SMC}_{\bis_{\rm ss}}(G)$ can be calculated from a unique
original stationary PMF of ${\it SMC}(G)$ by summing some elements of the latter, as described in \cite{Buc94b}.
Similar arguments demonstrate that the same holds for ${\it DTMC}(G)$ and ${\it DTMC}_{\bis_{\rm ss}}(G)$.

\section{Stationary behaviour}
\label{stationary.sec}

Let us examine how the proposed equivalences can be used to compare the behaviour of stochastic processes in their
steady states. We shall consider only formulas specifying stochastic processes with infinite behaviour, i.e.
expressions with the iteration operator. Note that the iteration operator does not guarantee infiniteness of behaviour,
since there can exist a deadlock (blocking) within the body (the second argument) of iteration when the corresponding
subprocess does not reach its final state by some reasons. In particular, if the body of iteration contains the ${\sf
Stop}$ expression then the iteration will be ``broken''. On the other hand, the iteration body can be left after a
finite number of
repeated executions and
perform the iteration termination.
To avoid executing activities after the iteration body, we take ${\sf Stop}$ as the termination argument of iteration.

Like in the framework of SMCs, in LDTSIPNs the most common systems for performance analysis are {\em ergodic}
(irreducible, positive recurrent and aperiodic) ones. For ergodic LDTSIPNs, the steady-state marking probabilities
exist and can be determined. In \cite{Mol81,Mol85}, the following sufficient (but not necessary) conditions for
ergodicity of DTSPNs are stated: {\em liveness} (for each transition and any reachable marking there exist a sequence
of markings from it leading to the marking enabling that transition), {\em boundedness} (for any reachable marking the
number of tokens in every place is not greater than some fixed number) and {\em nondeterminism} (the transition
probabilities are strictly less than $1$).

Consider dtsi-box of a dynamic expression $G=\overline{[E*F*{\sf Stop}]}$ specifying a process, which we assume has no
deadlocks while
performing $F$.
If, starting in $[[E*\overline{F}*{\sf Stop}]]_\approx$ and ending in $[[E*\underline{F}*{\sf Stop}]]_\approx$, only
tangible states are passed through, then the three ergodicity conditions are satisfied: the subnet corresponding to the
looping of the iteration body $F$ is live, safe ($1$-bounded) and nondeterministic (since all markings of the subnet
are tangible and non-terminal, the probabilities of transitions from them are strictly less than $1$). Hence, according
to \cite{Mol81,Mol85}, for the dtsi-box, its underlying SMC, restricted to the markings of the mentioned subnet, is
ergodic. The isomorphism between SMCs of expressions and those of the corresponding dtsi-boxes, which is stated by
Proposition \ref{smcs.pro}, guarantees that $SMC(G)$ is ergodic, if restricted to the states between
$[[E*\overline{F}*{\sf Stop}]]_\approx$ and $[[E*\underline{F}*{\sf Stop}]]_\approx$.

The ergodicity conditions above are not necessary, i.e. there exist dynamic expressions with vanishing states traversed
while executing their iteration bodies, such that the properly restricted underlying SMCs are nevertheless ergodic, as
Example \ref{exprsmc.exm} demonstrated. However, it has been shown in \cite{BKr02} that even live, safe and
nondeterministic DTSPNs (as well as live and safe CTSPNs and GSPNs) may be non-ergodic.

In this section, we consider only the process expressions such that their underlying SMCs contain exactly one closed
communication class of states, and this class should also be ergodic to ensure uniqueness of the stationary
distribution, which is also the limiting one. The states not belonging to that class do not disturb the uniqueness,
since the closed communication class is single, hence, they all are transient. Then, for each transient state, the
steady-state probability to be in it is zero while the steady-state probability to enter into the ergodic class
starting from that state is equal to one.
A communication class of states is their equivalence class w.r.t. communication relation, i.e. a maximal subset of
communicating states. A communication class of states is closed if only the states belonging to it are accessible from
every its state.

\subsection{Steady state, residence time and equivalences}

The following proposition demonstrates that, for two dynamic expressions related by $\bis_{\rm ss}$, the steady-state
probabilities to enter into an equivalence class coincide, or the mean recurrence time for an equivalence class is the
same for both expressions.

\begin{proposition}
Let $G,G'$ be dynamic expressions with $\mathcal{R}:G\bis_{\rm ss}G',\ \varphi$ be the steady-state PMF for ${\it
SMC}(G)$ and $\varphi '$ be the steady-state PMF for ${\it SMC}(G')$. Then $\forall\mathcal{H}\in (DR(G)\cup
DR(G'))/_\mathcal{R}$,
$$\sum_{s\in\mathcal{H}\cap DR(G)}\varphi (s)=\sum_{s'\in\mathcal{H}\cap DR(G')}\varphi '(s').$$
\label{statprob.pro}
\end{proposition}
\begin{proof}
See Appendix \ref{statprob.ssc}.
\end{proof}

Let $G$ be a dynamic expression, $\varphi$ be the steady-state PMF for ${\it SMC}(G)$ and $\varphi_{\bis_{\rm ss}}$ be
the steady-state PMF for ${\it SMC}_{\bis_{\rm ss}}(G)$. By Proposition \ref{statprob.pro}, we have
$\forall\mathcal{K}\in DR(G)/_{\mathcal{R}_{\rm ss}(G)},\ \varphi_{\bis_{\rm ss}}(\mathcal{K})=
\sum_{s\in\mathcal{K}}\varphi (s)$. Hence, using ${\it SMC}_{\bis_{\rm ss}}(G)$ instead of ${\it SMC}(G)$ simplifies
the analytical solution, since we have less states, but constructing the TPM for ${\it EDTMC}_{\bis_{\rm ss}}(G)$,
denoted by ${\bf P}_{\bis_{\rm ss}}^*$, also requires some efforts, including determining $\mathcal{R}_{\rm ss}(G)$ and
calculating the probabilities to move from one equivalence class to other. The behaviour of ${\it EDTMC}_{\bis_{\rm
ss}}(G)$ stabilizes quicker than that of ${\it EDTMC}(G)$ (if each of them has a single steady state), since ${\bf
P}_{\bis_{\rm ss}}^*$ is denser matrix than ${\bf P}^*$ (the TPM for ${\it EDTMC}(G)$) due to the fact that the former
matrix is smaller and the transitions between the equivalence classes ``include'' all the transitions between the
states belonging to these equivalence classes.

By Proposition \ref{statprob.pro}, $\bis_{\rm ss}$ preserves the quantitative properties of the stationary behaviour
(the level of SMCs). We now demonstrate that the qualitative properties based on the multiaction labels are preserved
as well (the transition systems level).

\begin{definition}
A {\em derived step trace} of a dynamic expression $G$ is a chain $\Sigma =A_1\cdots A_n\in (\naturals_{\rm
fin}^\mathcal{L})^*$, where $\exists s\in DR(G),\
s\stackrel{\Upsilon_1}{\rightarrow}s_1\stackrel{\Upsilon_2}{\rightarrow}\cdots \stackrel{\Upsilon_n}{\rightarrow}s_n,\
\mathcal{L}(\Upsilon_i)=A_i\ (1\leq i\leq n)$. Then the {\em probability to execute the derived step trace $\Sigma$ in
$s$} is
$$PT(\Sigma ,s)=\sum_{\{\Upsilon_1,\ldots ,\Upsilon_n\mid s=s_0\stackrel{\Upsilon_1}{\rightarrow}s_1
\stackrel{\Upsilon_2}{\rightarrow}\cdots\stackrel{\Upsilon_n}{\rightarrow}s_n,\ \mathcal{L}(\Upsilon_i)=A_i\ (1\leq
i\leq n)\}}\prod_{i=1}^{n}PT(\Upsilon_i,s_{i-1}).$$
\end{definition}

The following theorem demonstrates that, for two dynamic expressions related by $\bis_{\rm ss}$, the steady-state
probabilities to enter into an equivalence class and start a derived step trace from it coincide.

\begin{theorem}
Let $G,G'$ be dynamic expressions with $\mathcal{R}:G\bis_{\rm ss}G',\ \varphi$ be the steady-state PMF for ${\it
SMC}(G),\ \varphi '$ be the steady-state PMF for ${\it SMC}(G')$ and $\Sigma$ be a derived step trace of $G$ and $G'$.
Then $\forall\mathcal{H}\in (DR(G)\cup DR(G'))/_\mathcal{R}$,
$$\sum_{s\in\mathcal{H}\cap DR(G)}\varphi (s)PT(\Sigma ,s)=\sum_{s'\in\mathcal{H}\cap DR(G')}\varphi '(s')PT(\Sigma
,s').$$
\label{stattrace.the}
\end{theorem}
\begin{proof}
See Appendix \ref{stattrace.ssc}.
\end{proof}

Let $G$ be a dynamic expression, $\varphi$ be the steady-state PMF for ${\it SMC}(G),\ \varphi_{\bis_{\rm ss}}$ be the
steady-state PMF for ${\it SMC}_{\bis_{\rm ss}}(G)$ and $\Sigma$ be a derived step trace of $G$. By Theorem
\ref{stattrace.the}, we have $\forall\mathcal{K}\in DR(G)/_{\mathcal{R}_{\rm ss}(G)},\\
\varphi_{\bis_{\rm ss}}(\mathcal{K})PT(\Sigma ,\mathcal{K})=\sum_{s\in\mathcal{K}}\varphi (s)PT(\Sigma ,s)$, where
$\forall s\in\mathcal{K},\ PT(\Sigma ,\mathcal{K})=PT(\Sigma ,s)$.

We now present a result not concerning the steady-state probabilities, but revealing important properties of residence
time in the equivalence classes. The
next proposition demonstrates that, for two dynamic expressions related by $\bis_{\rm ss}$, the sojourn time averages
(and variances) in an equivalence class coincide.

\begin{proposition}
Let $G,G'$ be dynamic expressions with $\mathcal{R}\!:\!G\bis_{\rm ss}G'$. Then $\forall\mathcal{H}\!\in\!(DR(G)\cup
DR(G'))/_\mathcal{R}$,
$$\begin{array}{c}
{\it SJ}_{\mathcal{R}\cap (DR(G))^2}(\mathcal{H}\cap DR(G))={\it SJ}_{\mathcal{R}\cap (DR(G'))^2}(\mathcal{H}\cap
DR(G')),\\[1mm]
{\it VAR}_{\mathcal{R}\cap (DR(G))^2}(\mathcal{H}\cap DR(G))={\it VAR}_{\mathcal{R}\cap (DR(G'))^2}(\mathcal{H}\cap
DR(G')).
\end{array}$$
\label{sjavevar.pro}
\end{proposition}
\begin{proof}
See Appendix \ref{sjavevar.ssc}.
\end{proof}

\begin{example}
Let $E=[(\{a\},\frac{1}{2})*((\{b\},\frac{1}{2});((\{c\},\frac{1}{3})_1\cho (\{c\},\frac{1}{3})_2))*{\sf Stop}],\\
E'=[(\{a\},\frac{1}{2})*(((\{b\},\frac{1}{3})_1;(\{c\},\frac{1}{2})_1)\cho
((\{b\},\frac{1}{3})_2;(\{c\},\frac{1}{2})_2))*{\sf Stop}]$. It holds that $\overline{E}\bis_{\rm ss}\overline{E'}$.\\
$DR(\overline{E})$ consists of the equivalence classes\\
$\begin{array}{c}
s_1=[[\overline{(\{a\},\frac{1}{2})}*((\{b\},\frac{1}{2});((\{c\},\frac{1}{3})_1\cho (\{c\},\frac{1}{3})_2))*{\sf
Stop}]]_\approx ,\\[1mm]
s_2=[[(\{a\},\frac{1}{2})*(\overline{(\{b\},\frac{1}{2})};((\{c\},\frac{1}{3})_1\cho (\{c\},\frac{1}{3})_2))*{\sf
Stop}]]_\approx ,\\[1mm]
s_3=[[(\{a\},\frac{1}{2})*((\{b\},\frac{1}{2});\overline{((\{c\},\frac{1}{3})_1\cho (\{c\},\frac{1}{3})_2)})*{\sf
Stop}]]_\approx .
\end{array}$\\
$DR(\overline{E'})$ consists of the equivalence classes\\
$\begin{array}{c}
s_1'=[[\overline{(\{a\},\frac{1}{2})}*(((\{b\},\frac{1}{3})_1;(\{c\},\frac{1}{2})_1)\cho
((\{b\},\frac{1}{3})_2;(\{c\},\frac{1}{2})_2))*{\sf Stop}]]_\approx ,\\[1mm]
s_2'=[[(\{a\},\frac{1}{2})*\overline{(((\{b\},\frac{1}{3})_1;(\{c\},\frac{1}{2})_1)\cho
((\{b\},\frac{1}{3})_2;(\{c\},\frac{1}{2})_2))}*{\sf Stop}]]_\approx ,\\[1mm]
s_3'=[[(\{a\},\frac{1}{2})*(((\{b\},\frac{1}{3})_1;\overline{(\{c\},\frac{1}{2})_1})\cho
((\{b\},\frac{1}{3})_2;(\{c\},\frac{1}{2})_2))*{\sf Stop}]]_\approx ,\\[1mm]
s_4'=[[(\{a\},\frac{1}{2})*(((\{b\},\frac{1}{3})_1;(\{c\},\frac{1}{2})_1)\cho
((\{b\},\frac{1}{3})_2;\overline{(\{c\},\frac{1}{2})_2)})*{\sf Stop}]]_\approx .
\end{array}$\\
The steady-state PMFs $\varphi$ for ${\it SMC}(\overline{E})$ and $\varphi '$ for ${\it SMC}(\overline{E'})$ are
$\varphi=\left(0,\frac{1}{2},\frac{1}{2}\right),\ \varphi '=\left(0,\frac{1}{2},\frac{1}{4},\frac{1}{4}\right)$.

Let us consider the equivalence class (w.r.t. $\mathcal{R}_{\rm ss}(\overline{E},\overline{E'})$)
$\mathcal{H}=\{s_3,s_3',s_4'\}$. One can see that the steady-state probabilities for $\mathcal{H}$ coincide:
$\sum_{s\in\mathcal{H}\cap DR(\overline{E})}\varphi (s)=\varphi (s_3)=\frac{1}{2}=\frac{1}{4}+\frac{1}{4}=\varphi
'(s_3')+\varphi '(s_4')=\sum_{s'\in\mathcal{H}\cap DR(\overline{E'})}\varphi '(s')$. Let $\Sigma =\{\{c\}\}$. The
steady-state probabilities to enter into the equivalence class $\mathcal{H}$ and start the derived step trace $\Sigma$
from it coincide as well: $\varphi (s_3)(PT(\{(\{c\},\frac{1}{3})_1\},s_3)+PT(\{(\{c\},\frac{1}{3})_2\},s_3))=
\frac{1}{2}\left(\frac{1}{4}+\frac{1}{4}\right)=\frac{1}{4}=\frac{1}{4}\cdot\frac{1}{2}+\frac{1}{4}\cdot\frac{1}{2}=
\varphi '(s_3')PT(\{(\{c\},\frac{1}{2})_1\},s_3')+\\
\varphi '(s_4')PT(\{(\{c\},\frac{1}{2})_2\},s_4')$.

Further, the sojourn time averages in the equivalence class $\mathcal{H}$ coincide:\\
${\it SJ}_{\mathcal{R}_{\rm ss}(\overline{E},\overline{E'})\cap (DR(\overline{E}))^2}(\mathcal{H}\cap DR(G))=
{\it SJ}_{\mathcal{R}_{\rm ss}(\overline{E},\overline{E'})\cap (DR(\overline{E}))^2}(\{s_3\})=
\frac{1}{1-PM(\{s_3\},\{s_3\})}=\\
\frac{1}{1-PM(s_3,s_3)}=\frac{1}{1-\frac{1}{2}}=2=\frac{1}{1-\frac{1}{2}}=\frac{1}{1-PM(s_3',s_3')}=
\frac{1}{1-PM(s_4',s_4')}=\frac{1}{1-PM(\{s_3',s_4'\},\{s_3',s_4'\})}=\\
{\it SJ}_{\mathcal{R}_{\rm ss}(\overline{E},\overline{E'})\cap (DR(\overline{E'}))^2}(\{s_3',s_4'\})=
{\it SJ}_{\mathcal{R}_{\rm ss}(\overline{E},\overline{E'})\cap (DR(\overline{E'}))^2}(\mathcal{H}\cap DR(G'))$.

Finally, the sojourn time variances in the equivalence class $\mathcal{H}$ coincide:\\
${\it VAR}_{\mathcal{R}_{\rm ss}(\overline{E},\overline{E'})\cap (DR(\overline{E}))^2}(\mathcal{H}\cap DR(G))=
{\it VAR}_{\mathcal{R}_{\rm ss}(\overline{E},\overline{E'})\cap (DR(\overline{E}))^2}(\{s_3\})=
\frac{PM(\{s_3\},\{s_3\})}{(1-PM(\{s_3\},\{s_3\}))^2}=\\
\frac{PM(s_3,s_3)}{(1-PM(s_3,s_3))^2}=\frac{\frac{1}{2}}{\left(1-\frac{1}{2}\right)^2}=2=
\frac{\frac{1}{2}}{\left(1-\frac{1}{2}\right)^2}=\frac{PM(s_3',s_3')}{(1-PM(s_3',s_3'))^2}=
\frac{PM(s_4',s_4')}{(1-PM(s_4',s_4'))^2}=\\
\frac{PM(\{s_3',s_4'\},\{s_3',s_4'\})}{(1-PM(\{s_3',s_4'\},\{s_3',s_4'\}))^2}\!\!=\!\!
{\it VAR}_{\mathcal{R}_{\rm ss}(\overline{E},\overline{E'})\cap (DR(\overline{E'}))^2}(\{s_3',s_4'\})\!\!=\!\!
{\it VAR}_{\mathcal{R}_{\rm ss}(\overline{E},\overline{E'})\cap (DR(\overline{E'}))^2}(\mathcal{H}\cap DR(G'))$.

In Figure \ref{ssbssptimm.fig}, the marked dtsi-boxes corresponding to the dynamic expressions above are presented,
i.e. $N=Box_{\rm dtsi}(\overline{E})$ and $N'=Box_{\rm dtsi}(\overline{E'})$.
\label{ssbsspt.exm}
\end{example}

\begin{figure}
\begin{center}
\includegraphics[scale=0.9]{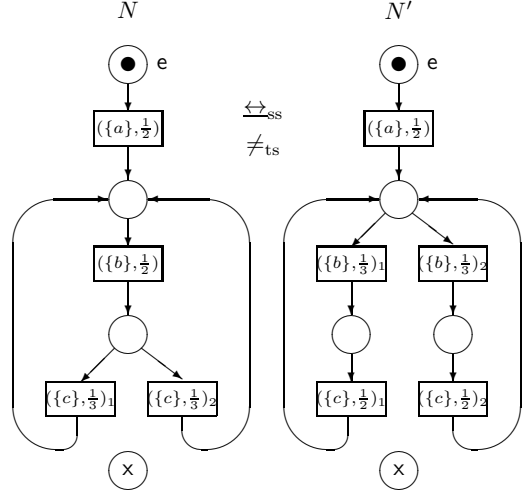}
\end{center}
\vspace{-6mm}
\caption{$\bis_{\rm ss}$ preserves steady-state behaviour and sojourn time properties in the equivalence classes.}
\label{ssbssptimm.fig}
\end{figure}

\subsection{Preservation of performance and simplification of its analysis}

Many performance indices are based on the steady-state probabilities to enter into a set of similar states or, after
coming in it, to start a derived step trace from this set. The similarity of states is usually captured by an
equivalence relation, hence, the sets are often the equivalence classes. Proposition \ref{statprob.pro}, Theorem
\ref{stattrace.the} and Proposition \ref{sjavevar.pro} guarantee coincidence of the mentioned indices for the
expressions related by $\bis_{\rm ss}$. Thus, $\bis_{\rm ss}$ (hence, all the stronger equivalences considered)
preserves performance of stochastic systems modeled by expressions of dtsiPBC.

It is also easier to evaluate performance using an SMC with less states, since in this case the size of the transition
probability matrix
is smaller, and we
solve systems of less equations to calculate the steady-state probabilities. The reasoning above validates the
following method of performance analysis simplification.
\begin{enumerate}

\item The investigated system is specified by a static expression of dtsiPBC.

\item The transition system of the expression is constructed.

\item After treating the transition system for self-similarity, a step stochastic autobisimulation equivalence for the
expression is determined.

\item The quotient underlying SMC is derived from the quotient transition system.

\item Stationary probabilities and performance indices are obtained using the SMC.

\end{enumerate}

The limitation of the method above is its applicability only to the expressions such that their underlying SMCs contain
exactly one closed communication class of states, and this class should also be ergodic to ensure uniqueness of the
stationary distribution. If an SMC contains several closed communication classes of states that are all ergodic then
several stationary distributions may exist, which depend on the initial PMF. There is an analytical method to determine
stationary probabilities for SMCs of this kind as well \cite{Kul09}. Note that the underlying SMC of every process
expression has only one initial PMF (that at the time moment $0$), hence, the stationary distribution will be unique in
this case too. The general steady-state probabilities are then calculated as the sum of the stationary probabilities of
all the ergodic subsets of states, weighted by the probabilities to enter into these subsets, starting from the initial
state and passing through some transient states. It is worth applying the method only to the systems with similar
subprocesses.

Before calculating stationary probabilities, we can further reduce the quotient underlying SMC, using the algorithm
from \cite{MBCDF95,Bal01,Bal07} that eliminates vanishing states from the corresponding EDTMC and thereby decreases the
size of its TPM. For SMCs reduction we can also apply an analogue of the deterministic barrier partitioning method from
\cite{GDBK11} for semi-Markov processes (SMPs), which allows one to perform quicker the first passage-time analysis.
Another option is the method of stochastic state classes \cite{HPRV12} for generalized SMPs (GSMPs) reduction that
simplifies the transient performance analysis.

\section{Generalized shared memory system}
\label{gshmsysim.sec}

Let us consider a model of two processors accessing a common shared memory described in \cite{MBCDF95,Bal01,Bal07} in
the continuous time setting on GSPNs. We shall analyze this shared memory system in the discrete time stochastic
setting of dtsiPBC, where concurrent execution of activities is possible, while no two transitions of a GSPN may fire
simultaneously (in parallel). Our model parameterizes that from \cite{TMV13}. The model behaves as follows. After
activation of the system (turning the computer on), two processors are active, and the common memory is available. Each
processor can request an access to the memory after which the instantaneous decision is made. When the decision is made
in favour of a processor, it starts acquisition of the memory and the other processor should wait until the former one
ends its memory operations, and the system returns to the state with both active processors and available common
memory. The diagram of the system is depicted in Figure \ref{shmdiagr.fig}.

\begin{figure}
\begin{center}
\includegraphics[scale=0.8]{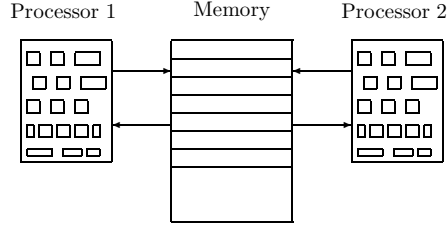}
\end{center}
\vspace{-6mm}
\caption{The diagram of the shared memory system.}
\label{shmdiagr.fig}
\end{figure}

\subsection{The concrete system}

The meaning of actions from the dtsiPBC expressions which will specify the system modules is as follows. The action $a$
corresponds to the system activation. The actions $r_i\ (1\leq i\leq 2)$ represent the common memory request of
processor $i$. The instantaneous actions $d_i$ correspond to the decision on the memory allocation in favour of the
processor $i$. The actions $m_i$ represent the common memory access of processor $i$. The other actions are used for
communication purposes only via synchronization, and we abstract from them later using restriction. For $a_1,\ldots
,a_n\in Act\ (n\in\naturals)$, we shall abbreviate $\sy a_1\cdots\sy a_n\rs a_1\cdots\rs a_n$ to $\sr (a_1,\ldots
,a_n)$.

We take general values for all multiaction probabilities and weights in the specification. Let all stochastic
multiactions have the same generalized probability $\rho\!\!\in\!\!(0;1)$ and all immediate ones have the same generalized
weight $l\!\!\in\!\!\reals_{>0}$. The resulting specification $K$ of the generalized shared memory system~is~below.

The static expression of the first processor is\\
$K_1=[(\{x_1\},\rho )*((\{r_1\},\rho );(\{d_1,y_1\},\natural_l);(\{m_1,z_1\},\rho ))*{\sf Stop}]$.

The static expression of the second processor is\\
$K_2=[(\{x_2\},\rho )*((\{r_2\},\rho );(\{d_2,y_2\},\natural_l);(\{m_2,z_2\},\rho ))*{\sf Stop}]$.

The static expression of the shared memory is\\
$K_3=[(\{a,\widehat{x_1},\widehat{x_2}\},\rho )*(((\{\widehat{y_1}\},\natural_l);(\{\widehat{z_1}\},\rho ))\cho
((\{\widehat{y_2}\},\natural_l);(\{\widehat{z_2}\},\rho )))*{\sf Stop}]$.

The static expression of the generalized shared memory system is\\
$K=(K_1\| K_2\| K_3)\sr (x_1,x_2,y_1,y_2,z_1,z_2)$.

As a result of the synchronization of immediate multiactions $(\{d_i,y_i\},\natural_l)$ and
$(\{\widehat{y_i}\},\natural_l)$ we get\\
$(\{d_i\},\natural_{2l})\ (1\leq i\leq 2)$. The synchronization of stochastic multiactions $(\{m_i,z_i\},\rho )$ and
$(\{\widehat{z_i}\},\rho )$ produces $(\{m_i\},\rho^2)\ (1\leq i\leq 2)$. The result of synchronization of
$(\{a,\widehat{x_1},\widehat{x_2}\},\rho )$ with $(\{x_1\},\rho )$ is $(\{a,\widehat{x_2}\},\rho^2)$, and that of
synchronization of $(\{a,\widehat{x_1},\widehat{x_2}\},\rho )$ with $(\{x_2\},\rho )$ is
$(\{a,\widehat{x_1}\},\rho^2)$. After applying synchronization to $(\{a,\widehat{x_2}\},\rho^2)$ and $(\{x_2\},\rho )$,
as well as to $(\{a,\widehat{x_1}\},\rho^2)$ and $(\{x_1\},\rho )$, we get the same activity $(\{a\},\rho^3)$.

We have $DR_{\rm T}(\overline{K})=\{\tilde{s}_1,\tilde{s}_2,\tilde{s}_5,\tilde{s}_5,\tilde{s}_8,\tilde{s}_9\}$ and
$DR_{\rm V}(\overline{K})=\{\tilde{s}_3,\tilde{s}_4,\tilde{s}_6\}$.

The interpretation of the states is: $\tilde{s}_1$ is the initial state, $\tilde{s}_2$: the system is activated and the
memory is not requested, $\tilde{s}_3$: the memory is requested by the first processor, $\tilde{s}_4$: the memory is
requested by the second processor, $\tilde{s}_5$: the memory is allocated to the first processor, $\tilde{s}_6$: the
memory is requested by two processors, $\tilde{s}_7$: the memory is allocated to the second processor, $\tilde{s}_8$:
the memory is allocated to the first processor and the memory is requested by the second processor, $\tilde{s}_9$: the
memory is allocated to the second processor and the memory is requested by the first processor.

In Figure \ref{shmimgtsrw.fig}, the transition system $TS(\overline{K})$ is presented. In Figure \ref{shmimgsmc.fig},
the underlying SMC ${\it SMC}(\overline{K})$ is depicted. Note that, in step semantics, we may execute the following
activities in parallel: $(\{r_1\},\rho ),\\
(\{r_2\},\rho )$, as well as $(\{r_1\},\rho ),(\{m_2\},\rho^2)$, and $(\{r_2\},\rho ),(\{m_1\},\rho^2)$. Therefore, the
state $\tilde{s}_6$ only exists in step semantics, since it is reachable exclusively by executing $(\{r_1\},\rho )$ and
$(\{r_2\},\rho )$ in parallel.

\begin{figure}
\begin{center}
\includegraphics[scale=0.9]{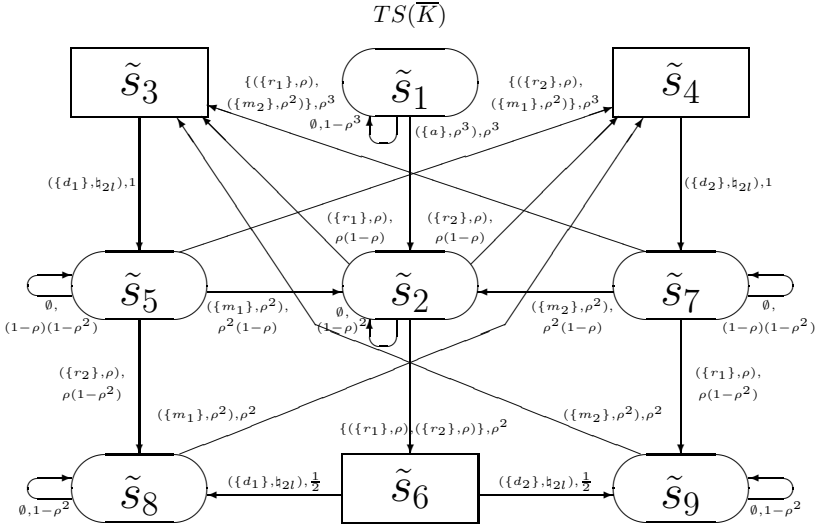}
\end{center}
\vspace{-6mm}
\caption{The transition system of the generalized shared memory system.}
\label{shmimgtsrw.fig}
\end{figure}

\begin{figure}
\begin{center}
\includegraphics[scale=0.9]{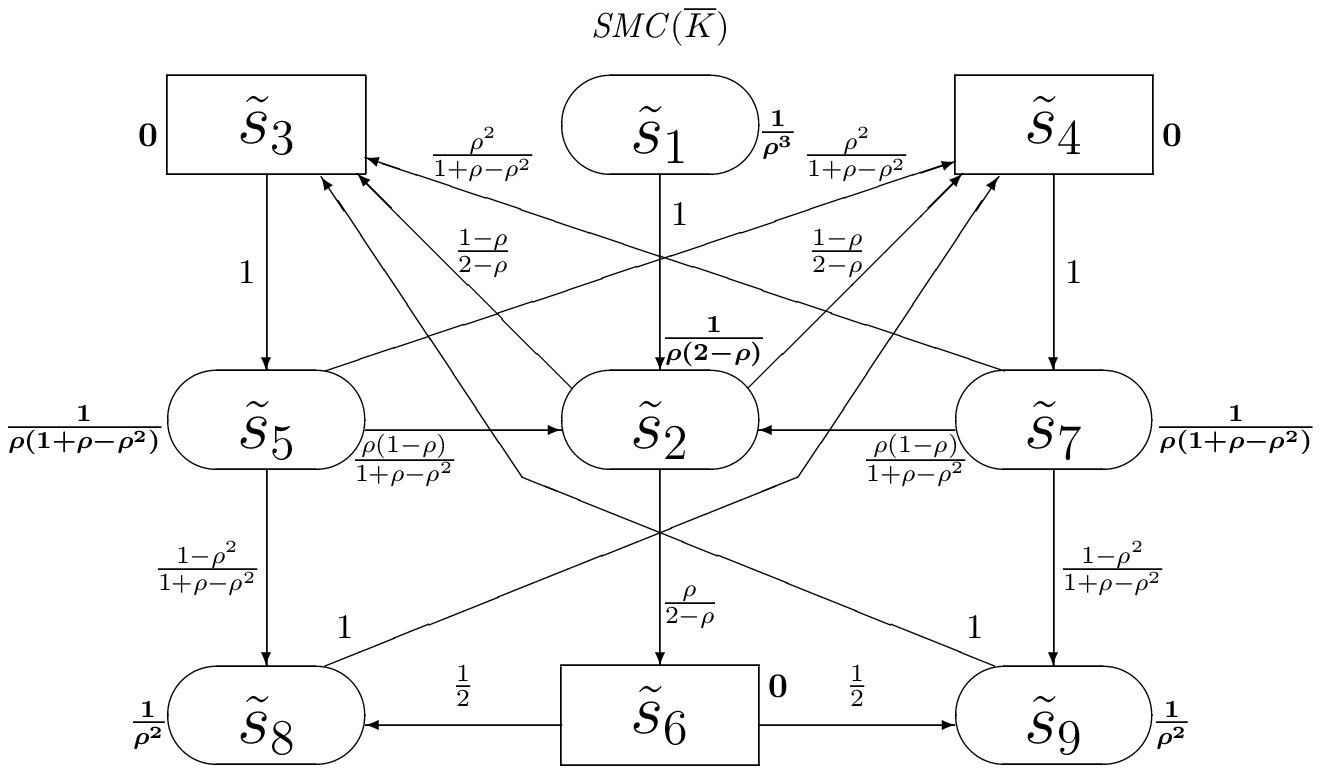}
\end{center}
\vspace{-6mm}
\caption{The underlying SMC of the generalized shared memory system.}
\label{shmimgsmc.fig}
\end{figure}

The average sojourn time vector of $\overline{K}$ is $\widetilde{\it SJ}=\left(\frac{1}{\rho^3},\frac{1}{\rho (2-\rho
)},0,0,\frac{1}{\rho (1+\rho -\rho^2)},0, \frac{1}{\rho (1+\rho -\rho^2)},\frac{1}{\rho^2},\frac{1}{\rho^2}\right)$.

The sojourn time variance vector of $\overline{K}$ is\\
$\begin{array}{c}
\widetilde{\it VAR}\!\!=\!\!\left(\!\frac{1-\rho^3}{\rho^6},\frac{(1-\rho )^2}{\rho^2(2-\rho )^2},0,0,\frac{(1-\rho
)^2(1+\rho )}{\rho^2(1+\rho -\rho^2)^2},0,\frac{(1-\rho )^2(1+\rho )}{\rho^2(1+\rho
-\rho^2)^2},\frac{1-\rho^2}{\rho^4},\frac{1-\rho^2}{\rho^4}\!\right).
\end{array}$
The TPM for ${\it EDTMC}(\overline{K})$~is\\
$\widetilde{\bf P}^*=\left(\begin{array}{ccccccccc}
0 & 1 & 0 & 0 & 0 & 0 & 0 & 0 & 0\\
0 & 0 & \frac{1-\rho}{2-\rho} & \frac{1-\rho}{2-\rho} & 0 & \frac{\rho}{2-\rho} & 0 & 0 & 0\\
0 & 0 & 0 & 0 & 1 & 0 & 0 & 0 & 0\\
0 & 0 & 0 & 0 & 0 & 0 & 1 & 0 & 0\\
0 & \frac{\rho (1-\rho )}{1+\rho -\rho^2} & 0 & \frac{\rho^2}{1+\rho -\rho^2} & 0 & 0 & 0 &
\frac{1-\rho^2}{1+\rho -\rho^2} & 0\\
0 & 0 & 0 & 0 & 0 & 0 & 0 & \frac{1}{2} & \frac{1}{2}\\
0 & \frac{\rho (1-\rho )}{1+\rho -\rho^2} & \frac{\rho^2}{1+\rho -\rho^2} & 0 & 0 & 0 & 0 & 0 &
\frac{1-\rho^2}{1+\rho -\rho^2}\\
0 & 0 & 0 & 1 & 0 & 0 & 0 & 0 & 0\\
0 & 0 & 1 & 0 & 0 & 0 & 0 & 0 & 0
\end{array}\right)$.

The steady-state PMF for ${\it EDTMC}(\overline{K})$ is $\tilde{\psi}^*=\frac{1}{2(6+3\rho -9\rho^2+2\rho^3)}
(0,2\rho (2-3\rho -\rho^2),2+\rho -3\rho^2+\rho^3,\\
2+\rho -3\rho^2+\rho^3,2+\rho -3\rho^2+\rho^3,2\rho^2(1-\rho ),2+\rho -3\rho^2+\rho^3,2-\rho -\rho^2,2-\rho -\rho^2)$.

The steady-state PMF $\tilde{\psi}^*$ weighted by $\widetilde{\it SJ}$ is\\
$\frac{1}{2\rho^2(6+3\rho -9\rho^2+2\rho^3)}(0,2\rho^2(1-\rho ),0,0,\rho (2-\rho ),0,\rho (2-\rho ),2-\rho -\rho^2,
2-\rho -\rho^2)$.

We normalize the steady-state weighted PMF, dividing it by the sum of its components\\
$\tilde{\psi}^*\widetilde{\it SJ}^{\rm T}=\frac{2+\rho -\rho^2-\rho^3}{\rho^2(6+3\rho -9\rho^2+2\rho^3)}$.
The steady-state PMF for ${\it SMC}(\overline{K})$ is\\
$\tilde{\varphi}=\frac{1}{2(2+\rho -\rho^2-\rho^3)}(0,2\rho^2(1-\rho ),0,0,\rho (2-\rho ),0,\rho (2-\rho ),
2-\rho -\rho^2,2-\rho -\rho^2)$.

We can now calculate the main performance indices.
\begin{itemize}

\item The average recurrence time in the state $\tilde{s}_2$, where no processor requests the memory, called the {\em
average system run-through}, is $\frac{1}{\tilde{\varphi}_2}=\frac{2+\rho -\rho^2-\rho^3}{\rho^2(1-\rho )}$.

\item The common memory is available only in the states $\tilde{s}_2,\tilde{s}_3,\tilde{s}_4,\tilde{s}_6$. The
steady-state probability that the memory is available is $\tilde{\varphi}_2+\tilde{\varphi}_3+\tilde{\varphi}_4+
\tilde{\varphi}_6=\frac{\rho^2(1-\rho )}{2+\rho -\rho^2-\rho^3}+0+0+0=\frac{\rho^2(1-\rho )}{2+\rho
-\rho^2-\rho^3}$. The steady-state probability that the memory is used (i.e. not available), called the {\em shared
memory utilization}, is $1-\frac{\rho^2(1-\rho )}{2+\rho -\rho^2-\rho^3}=\frac{2+\rho -2\rho^2}{2+\rho
-\rho^2-\rho^3}$.

\item After activation of the system, we leave the state $\tilde{s}_1$ for ever, and the common memory is either
requested or allocated in every remaining state, with exception of $\tilde{s}_2$.
The {\em rate with which the necessity of shared memory emerges} coincides with the rate of leaving $\tilde{s}_2$,
$\frac{\tilde{\varphi}_2}{\widetilde{\it SJ}_2}=\frac{\rho^2(1-\rho )}{2+\rho -\rho^2-\rho^3}\cdot\frac{\rho
(2-\rho )}{1}=\frac{\rho^3(1-\rho )(2-\rho )}{2+\rho -\rho^2-\rho^3}$.

\item The parallel common memory request of two processors $\{(\{r_1\},\rho ),(\{r_2\},\rho )\}$ is only possible
from the state $\tilde{s}_2$. In this state, the request probability is the sum of the execution probabilities for
all multisets of activities containing both $(\{r_1\},\rho )$ and $(\{r_2\},\rho )$.
The {\em steady-state probability of the shared memory request from two processors} is
$\tilde{\varphi}_2\sum_{\{\Upsilon\mid (\{(\{r_1\},\rho ),(\{r_2\},\rho )\}\subseteq\Upsilon\}}PT(\Upsilon
,\tilde{s}_2)=\frac{\rho^2(1-\rho )}{2+\rho -\rho^2-\rho^3}\rho^2=\frac{\rho^4(1-\rho )}{2+\rho -\rho^2-\rho^3}$.

\item The common memory request of the first processor $(\{r_1\},\rho )$ is only possible from the states
$\tilde{s}_2,\tilde{s}_7$. In each of the states, the request probability is the sum of the execution probabilities
for all sets of activities containing $(\{r_1\},\rho )$. The {\em steady-state probability of the shared memory
request from the first processor} is $\tilde{\varphi}_2\sum_{\{\Upsilon\mid (\{r_1\},\rho )\in\Upsilon\}}PT(\Upsilon
,\tilde{s}_2)+\tilde{\varphi}_7\sum_{\{\Upsilon\mid (\{r_1\},\rho )\in\Upsilon\}}PT(\Upsilon ,\tilde{s}_7)=\\
\frac{\rho^2(1-\rho )}{2+\rho -\rho^2-\rho^3}(\rho (1-\rho )+\rho^2)+
\frac{\rho (2-\rho )}{2(2+\rho -\rho^2-\rho^3)}(\rho (1-\rho^2)+\rho^3)=
\frac{\rho^2(2+\rho -2\rho^2)}{2(2+\rho -\rho^2-\rho^3)}$.

\end{itemize}

In Figure \ref{shmimgmdboxrw.fig}, the marked dtsi-boxes corresponding to the dynamic expressions of two processors,
shared memory and the generalized shared memory system are presented, i.e. $N_i=Box_{\rm dtsi}(\overline{K_i})\ (1\leq
i\leq 3)$ and $N=Box_{\rm dtsi}(\overline{K})$.

\begin{figure}
\begin{center}
\includegraphics[scale=0.9]{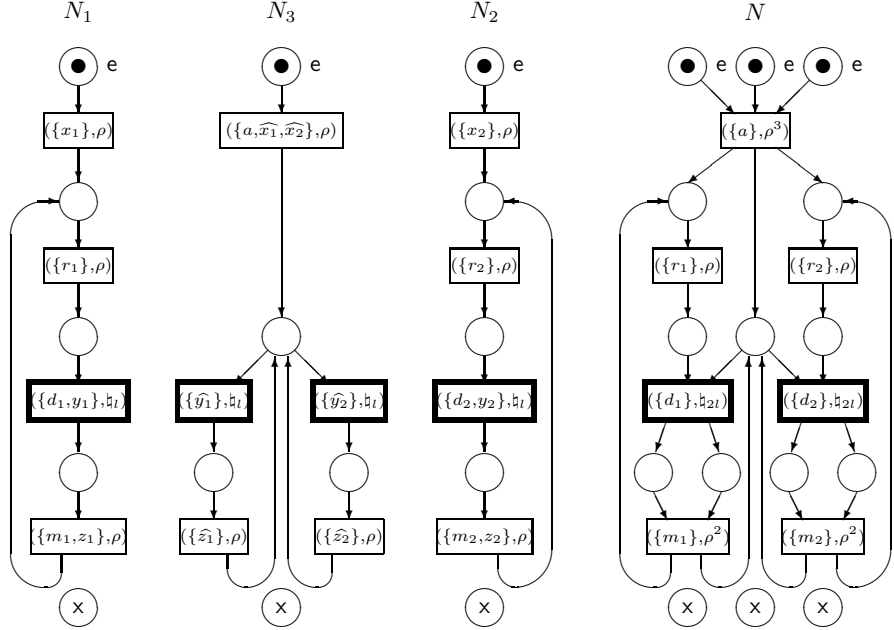}
\end{center}
\vspace{-6mm}
\caption{The marked dtsi-boxes of two processors, shared memory and the generalized shared memory system.}
\label{shmimgmdboxrw.fig}
\end{figure}

\subsection{The abstract system}

Consider a modification of the generalized shared memory system with abstraction from the identifiers of the
processors
that makes them indistinguishable, called the abstract generalized shared memory one. For the abstraction, we replace
the actions $r_i,d_i,m_i\ (1\leq i\leq 2)$ in the system specification by $r,d,m$.

The static expression of the first processor is\\
$L_1=[(\{x_1\},\rho )*((\{r\},\rho );(\{d,y_1\},\natural_l);(\{m,z_1\},\rho ))*{\sf Stop}]$.

The static expression of the second processor is\\
$L_2=[(\{x_2\},\rho )*((\{r\},\rho );(\{d,y_2\},\natural_l);(\{m,z_2\},\rho ))*{\sf Stop}]$.

The static expression of the shared memory is\\
$L_3=[(\{a,\widehat{x_1},\widehat{x_2}\},\rho )*(((\{\widehat{y_1}\},\natural_l);(\{\widehat{z_1}\},\rho ))
\cho ((\{\widehat{y_2}\},\natural_l);(\{\widehat{z_2}\},\rho )))*{\sf Stop}]$.

The static expression of the abstract generalized shared memory system is\\
$L=(L_1\| L_2\| L_3)\sr (x_1,x_2,y_1,y_2,z_1,z_2)$.

$DR(\overline{L})$ resembles $DR(\overline{K})$, and $TS(\overline{L})$ is similar to $TS(\overline{K})$. We have ${\it
SMC}(\overline{L})\simeq{\it SMC}(\overline{K})$. Thus, the average sojourn time vectors of $\overline{L}$ and
$\overline{K}$, as well as the TPMs and the steady-state PMFs for ${\it EDTMC}(\overline{L})$ and ${\it
EDTMC}(\overline{K})$, coincide.

The first, second and third performance indices are the same for the generalized system and its abstraction. The next
performance index is specific to the abstract system.
\begin{itemize}

\item The common memory request of a processor $(\{r\},\rho )$ is only possible from the states
$\tilde{s}_2,\tilde{s}_5,\tilde{s}_7$. In each of the states, the request probability is the sum of the execution
probabilities for all sets of activities containing $(\{r\},\rho )$. The {\em steady-state probability of the
shared memory request from a processor} is
$\tilde{\varphi}_2\sum_{\{\Upsilon\mid (\{r\},\rho )\in\Upsilon\}}PT(\Upsilon ,\tilde{s}_2)+
\tilde{\varphi}_5\sum_{\{\Upsilon\mid (\{r\},\rho )\in\Upsilon\}}PT(\Upsilon ,\tilde{s}_5)+\\
\tilde{\varphi}_7\sum_{\{\Upsilon\mid (\{r\},\rho )\in\Upsilon\}}PT(\Upsilon ,\tilde{s}_7)=
\frac{\rho^2(1-\rho )}{2+\rho -\rho^2-\rho^3}(\rho (1-\rho )+\rho (1-\rho )+\rho^2)+\\
\frac{\rho (2-\rho )}{2(2+\rho -\rho^2-\rho^3)}(\rho (1-\rho^2)+\rho^3)+
\frac{\rho (2-\rho )}{2(2+\rho -\rho^2-\rho^3)}(\rho (1-\rho^2)+\rho^3)=
\frac{\rho^2(2-\rho )(1+\rho -\rho^2)}{2+\rho -\rho^2-\rho^3}$.

\end{itemize}

We have $DR(\overline{L})/_{\mathcal{R}_{\rm ss}(\overline{L})}=\{\widetilde{\mathcal{K}}_1,\widetilde{\mathcal{K}}_2,
\widetilde{\mathcal{K}}_3,\widetilde{\mathcal{K}}_4,\widetilde{\mathcal{K}}_5,\widetilde{\mathcal{K}}_6\}$, where
$\widetilde{\mathcal{K}}_1=\{\tilde{s}_1\}$ (the initial state), $\widetilde{\mathcal{K}}_2=\{\tilde{s}_2\}$ (the
system is activated and the memory is not requested), $\widetilde{\mathcal{K}}_3=\{\tilde{s}_3,\tilde{s}_4\}$ (the
memory is requested by one processor), $\widetilde{\mathcal{K}}_4=\{\tilde{s}_5,\tilde{s}_7\}$ (the memory is allocated
to a processor), $\widetilde{\mathcal{K}}_5=\{\tilde{s}_6\}$ (the memory is requested by two processors),
$\widetilde{\mathcal{K}}_6=\{\tilde{s}_8,\tilde{s}_9\}$ (the memory is allocated to a processor and the memory is
requested by another processor). Further, $DR_{\rm T}(\overline{L})/_{\mathcal{R}_{\rm ss}(\overline{L})}=
\{\widetilde{\mathcal{K}}_1,\widetilde{\mathcal{K}}_2,\widetilde{\mathcal{K}}_4,\widetilde{\mathcal{K}}_6\}$ and
$DR_{\rm V}(\overline{L})/_{\mathcal{R}_{\rm ss}(\overline{L})}=
\{\widetilde{\mathcal{K}}_3,\widetilde{\mathcal{K}}_5\}$.

In Figure \ref{shmimgqtsm.fig}, the quotient transition system $TS_{\bis_{\rm ss}}(\overline{L})$ is presented. In
Figure \ref{shmimgqsmcm.fig}, the quotient underlying SMC ${\it SMC}_{\bis_{\rm ss}}(\overline{L})$ is depicted. Note
that, in step semantics, we may execute the following multiactions in parallel: $\{r\},\{r\}$, as well as
$\{r\},\{m\}$. Again, the state $\widetilde{\mathcal{K}}_5$ only exists in step semantics, since it is reachable
exclusively by executing $\{r\}$ and $\{r\}$ in parallel.

\begin{figure}
\begin{center}
\includegraphics[scale=0.85]{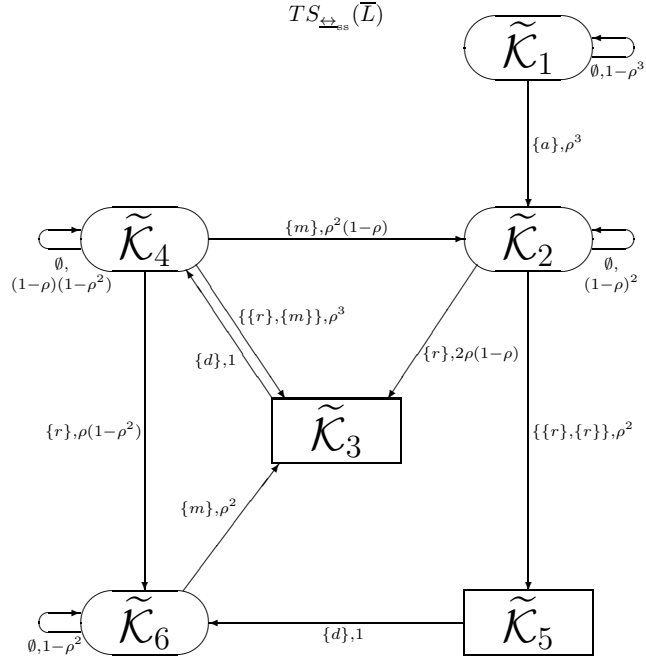}
\end{center}
\vspace{-6mm}
\caption{The quotient transition system of the abstract generalized shared memory system.}
\label{shmimgqtsm.fig}
\end{figure}

\begin{figure}
\begin{center}
\includegraphics[scale=0.85]{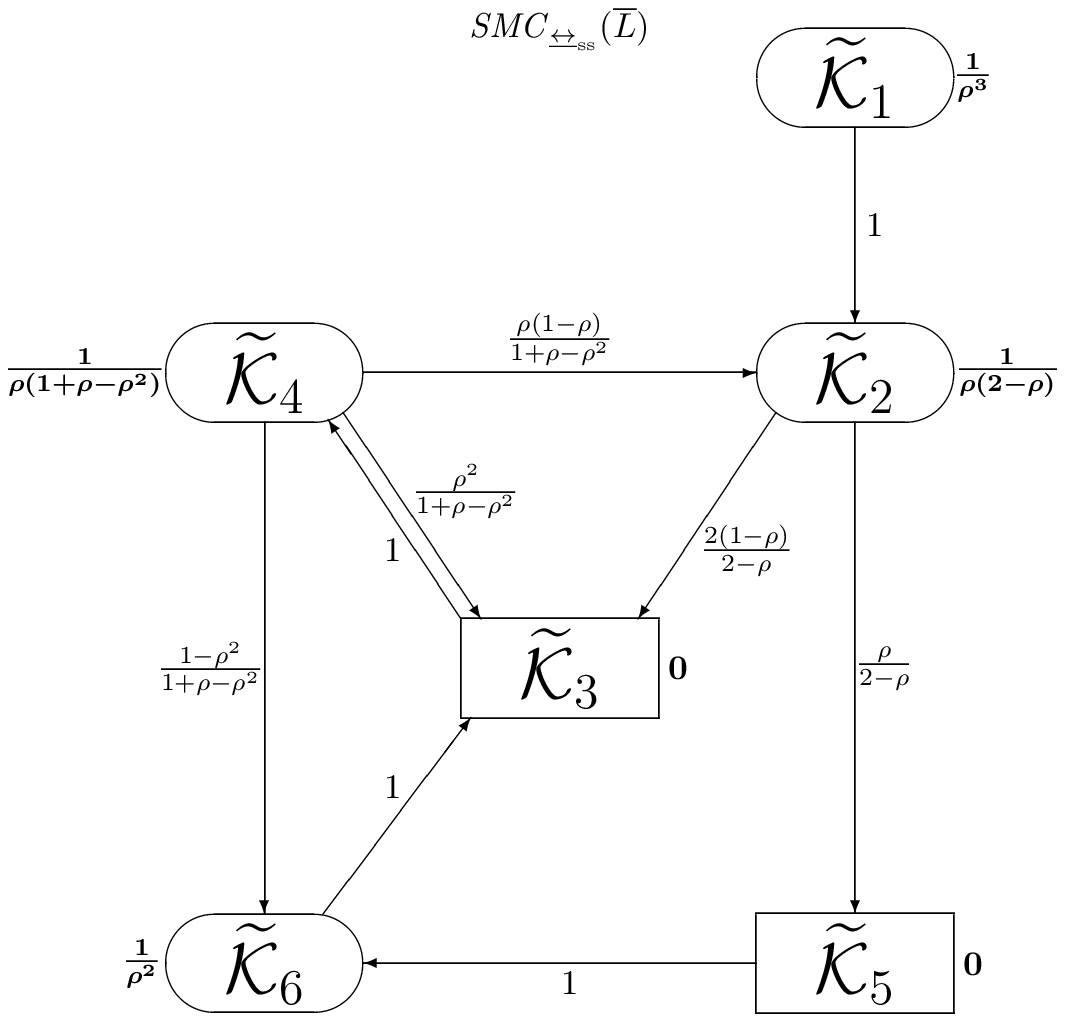}
\end{center}
\vspace{-6mm}
\caption{The quotient underlying SMC of the abstract generalized shared memory system.}
\label{shmimgqsmcm.fig}
\end{figure}

The quotient average sojourn time vector of $\overline{F}$ is
$\widetilde{\it SJ}'=\left(\frac{1}{\rho^3},\frac{1}{\rho (2-\rho )},0,\frac{1}{\rho (1+\rho -\rho^2)},0,
\frac{1}{\rho^2}\right)$.

The quotient sojourn time variance vector of $\overline{F}$ is
$\widetilde{\it VAR}'=\left(\frac{1-\rho^3}{\rho^6},\frac{(1-\rho )^2}{\rho^2(2-\rho )^2},0,\frac{(1-\rho
)^2(1+\rho )}{\rho^2(1+\rho -\rho^2)^2},0,\frac{1-\rho^2}{\rho^4}\right)$.

The TPM for ${\it EDTMC}_{\bis_{\rm ss}}(\overline{L})$ is
$\widetilde{\bf P}'^*=\left(\begin{array}{cccccc}
0 & 1 & 0 & 0 & 0 & 0\\
0 & 0 & \frac{2(1-\rho )}{2-\rho} & 0 & \frac{\rho}{2-\rho} & 0\\
0 & 0 & 0 & 1 & 0 & 0\\
0 & \frac{\rho (1-\rho )}{1+\rho -\rho^2} & \frac{\rho^2}{1+\rho -\rho^2} & 0 & 0 & \frac{1-\rho^2}{1+\rho -\rho^2}\\
0 & 0 & 0 & 0 & 0 & 1\\
0 & 0 & 1 & 0 & 0 & 0
\end{array}\right)$.

The steady-state PMF for ${\it EDTMC}_{\bis_{\rm ss}}(\overline{L})$ is\\
$\tilde{\psi}'^*=\frac{1}{6+3\rho -9\rho^2+2\rho^3}(0,\rho (2-3\rho +\rho^2),2+\rho -3\rho^2+\rho^3,
2+\rho -3\rho^2+\rho^3,\rho^2(1-\rho ),2-\rho -\rho^2)$.

The steady-state PMF $\tilde{\psi}'^*$ weighted by $\widetilde{\it SJ}'$ is
$\frac{1}{\rho^2(6+3\rho -9\rho^2+2\rho^3)}(0,\rho^2(1-\rho ),0,\rho (2-\rho ),0,2-\rho -\rho^2)$.

We normalize the steady-state weighted PMF, dividing it by the sum of its components\\
$\tilde{\psi}'^*\widetilde{\it SJ}'^{\rm T}=\frac{2+\rho -\rho^2-\rho^3}{\rho^2(6+3\rho -9\rho^2+2\rho^3)}$.

The steady-state PMF for ${\it SMC}_{\bis_{\rm ss}}(\overline{L})$ is
$\tilde{\varphi}'=\frac{1}{2+\rho -\rho^2-\rho^3}(0,\rho^2(1-\rho ),0,\rho (2-\rho ),0,2-\rho -\rho^2)$.

We can now calculate the main performance indices.
\begin{itemize}

\item The average recurrence time in the state $\widetilde{\mathcal{K}}_2$, where no processor requests the memory,
called the {\em average system run-through}, is $\frac{1}{\tilde{\varphi}_2'}=\frac{2+\rho
-\rho^2-\rho^3}{\rho^2(1-\rho )}$.

\item The common memory is available only in the states $\widetilde{\mathcal{K}}_2,\widetilde{\mathcal{K}}_3,
\widetilde{\mathcal{K}}_5$. The steady-state probability that the memory is available is
$\tilde{\varphi}_2'+\tilde{\varphi}_3'+\tilde{\varphi}_5'=\frac{\rho^2(1-\rho )}{2+\rho -\rho^2-\rho^3}+0+0=
\frac{\rho^2(1-\rho )}{2+\rho -\rho^2-\rho^3}$. The steady-state probability that the memory is used (i.e. not
available), called the {\em shared memory utilization}, is $1-\frac{\rho^2(1-\rho )}{2+\rho -\rho^2-\rho^3}=
\frac{2+\rho -2\rho^2}{2+\rho -\rho^2-\rho^3}$.

\item After activation of the system, we leave the state $\widetilde{\mathcal{K}}_1$ for ever, and the common memory is
either requested or allocated in every remaining state, with exception of $\widetilde{\mathcal{K}}_2$.
The {\em
rate with which the necessity of shared memory emerges} coincides with the rate of leaving $\widetilde{\mathcal{K}}_2$,
$\frac{\tilde{\varphi}_2'}{\widetilde{\it SJ}_2'}= \frac{\rho^2(1-\rho )}{2+\rho -\rho^2-\rho^3}\cdot\frac{\rho
(2-\rho )}{1}=\frac{\rho^3(1-\rho )(2-\rho )}{2+\rho -\rho^2-\rho^3}$.

\item The parallel common memory request of two processors $\{\{r\},\{r\}\}$ is only possible from the state
$\widetilde{\mathcal{K}_2}$. In this state, the request probability is the sum of the execution probabilities for
all multisets of multiactions containing $\{r\}$ twice.
The {\em steady-state probability of the shared memory request from two processors} is
$\tilde{\varphi}_2'\sum_{\{A,\widetilde{\mathcal{K}}\mid \{\{r\},\{r\}\}\subseteq A,\
\widetilde{\mathcal{K}}_2\stackrel{A}{\rightarrow}\widetilde{\mathcal{K}}\}} PM_A(\widetilde{\mathcal{K}}_2,
\widetilde{\mathcal{K}})=\frac{\rho^2(1-\rho )}{2+\rho -\rho^2-\rho^3}\rho^2=\frac{\rho^4(1-\rho )}{2+\rho
-\rho^2-\rho^3}$.

\item The common memory request of a processor $\{r\}$ is only possible from the states $\widetilde{\mathcal{K}}_2,
\widetilde{\mathcal{K}}_4$. In each of the states, the request probability is the sum of the execution
probabilities for all multisets of multiactions containing $\{r\}$. The {\em steady-state probability of the shared
memory request from a processor} is
$\tilde{\varphi}_2'\sum_{\{A,\widetilde{\mathcal{K}}\mid\{r\}\in A,\
\widetilde{\mathcal{K}}_2\stackrel{A}{\rightarrow}\widetilde{\mathcal{K}}\}}
PM_A(\widetilde{\mathcal{K}}_2,\widetilde{\mathcal{K}})+
\tilde{\varphi}_4'\sum_{\{A,\widetilde{\mathcal{K}}\mid\{r\}\in A,\
\widetilde{\mathcal{K}}_4\stackrel{A}{\rightarrow}\widetilde{\mathcal{K}}\}}
PM_A(\widetilde{\mathcal{K}}_4,\widetilde{\mathcal{K}})=\\
\frac{\rho^2(1-\rho )}{2+\rho -\rho^2-\rho^3}(2\rho (1-\rho )+\rho^2)+\frac{\rho (2-\rho )}{2+\rho
-\rho^2-\rho^3}(\rho (1-\rho^2)+\rho^3)= \frac{\rho^2(2-\rho )(1+\rho -\rho^2)}{2+\rho -\rho^2-\rho^3}$.

\end{itemize}

The performance indices are the same for the complete and the quotient abstract generalized shared memory systems. The
coincidence of the first, second and third performance indices obviously illustrates the results of Proposition
\ref{statprob.pro} and Proposition \ref{sjavevar.pro}. The coincidence of the fourth performance index is due to
Theorem \ref{stattrace.the}: one should just apply
this result to the derived step trace $\{\{r\},\{r\}\}$ of the expression $\overline{L}$ and itself. The coincidence of
the fifth performance index is due to Theorem \ref{stattrace.the}: one should just apply
this result to the derived step traces $\{\{r\}\},\ \{\{r\},\{r\}\},\ \{\{r\},\{m\}\}$ of the expression $\overline{L}$
and itself, and then sum the left and right parts of the three resulting equalities.

Let us consider what is the effect of quantitative changes of the parameter $\rho$ upon performance of the quotient
abstract generalized shared memory system in its steady state. Remember that $\rho\in (0;1)$ is the probability of
every stochastic multiaction in the specification of the system. The closer is $\rho$ to $0$, the less is the
probability to execute some activities at every discrete time tick, hence, the system will most probably {\em stand
idle}. The closer is $\rho$ to $1$, the greater is the probability to execute some activities at every discrete time
tact, hence, the system will most probably {\em operate}.

Since $\tilde{\varphi}_1'=\tilde{\varphi}_3'=\tilde{\varphi}_5'=0$, only $\tilde{\varphi}_2'=\frac{\rho^2(1-\rho
)}{2+\rho -\rho^2-\rho^3},\ \tilde{\varphi}_4'=\frac{\rho (2-\rho )}{2+\rho -\rho^2-\rho^3},\
\tilde{\varphi}_6'=\frac{2-\rho -\rho^2}{2+\rho -\rho^2-\rho^3}$ depend on $\rho$. In Figure \ref{shmimgqssp.fig}, the
plots of $\tilde{\varphi}_2',\ \tilde{\varphi}_4',\ \tilde{\varphi}_6'$ as functions of $\rho$ are depicted. Remember
that we do not allow $\rho =0$ or $\rho =1$.

One can see that $\tilde{\varphi}_2',\ \tilde{\varphi}_4'$ tend to $0$ and $\tilde{\varphi}_6'$ tends to $1$ when
$\rho$ approaches $0$. Thus, when $\rho$ is closer to $0$, the probability that the memory is allocated to a processor
and the memory is requested by another processor increases, hence, we have {\em more unsatisfied memory requests}.

Next, $\tilde{\varphi}_2',\ \tilde{\varphi}_6'$ tend to $0$ and $\tilde{\varphi}_4'$ tends to $1$ when $\rho$
approaches $1$. When $\rho$ is closer to $1$, the probability that the memory is allocated to a processor (and not
requested by another one) increases, hence, we have {\em less unsatisfied memory requests}.

The maximal value $0.0797$ of $\tilde{\varphi}_2'$ is reached when $\rho\approx 0.7433$. The probability that the
system is acti\-vated and the memory is not requested is maximal,
the {\em maximal shared memory availability}, is about~$8\%$.

\begin{figure}
\begin{center}
\includegraphics[width=0.8\textwidth]{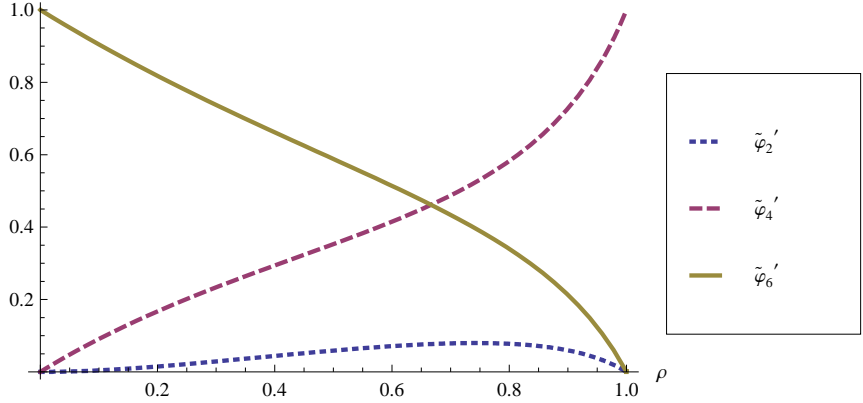}
\end{center}
\vspace{-8mm}
\caption{Steady-state probabilities $\tilde{\varphi}_2',\ \tilde{\varphi}_4',\ \tilde{\varphi}_6'$ as functions of the
parameter $\rho$.}
\label{shmimgqssp.fig}
\end{figure}

In Figure \ref{shmimgqart.fig}, the plot of the average system run-through, calculated as
$\frac{1}{\tilde{\varphi}_2'}$, as a function of $\rho$ is depicted.
The run-through tends to $\infty$ when $\rho$ approaches $0$ or $1$. Its minimal value $12.5516$ is reached when
$\rho\approx 0.7433$. To speed up operation of the system, one should take the parameter $\rho$ closer to $0.7433$.

\begin{figure}
\begin{center}
\includegraphics[width=0.8\textwidth]{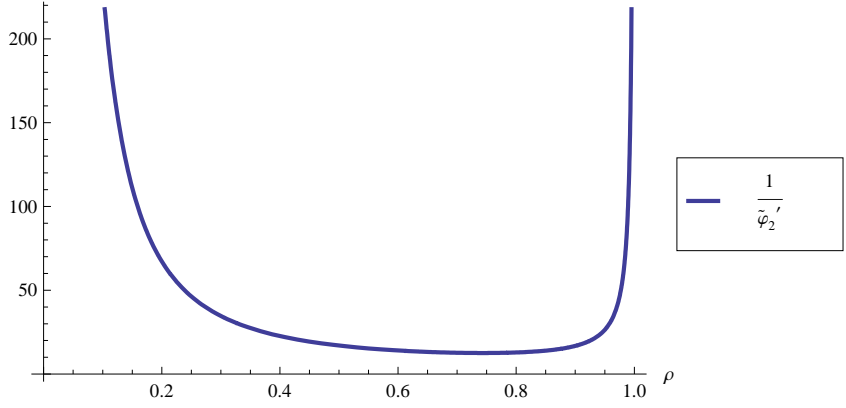}
\end{center}
\vspace{-8mm}
\caption{Average system run-through $\frac{1}{\tilde{\varphi}_2'}$ as a function of the parameter $\rho$.}
\label{shmimgqart.fig}
\end{figure}

The first curve in Figure \ref{shmimgqindn.fig} represents the shared memory utilization, calculated as
$1-\tilde{\varphi}_2'-\tilde{\varphi}_3'-\tilde{\varphi}_5'$, as a function of $\rho$. The utilization tends to $1$
both when $\rho$ approaches $0$ and when $\rho$ approaches $1$. The minimal value $0.9203$ of the utilization is
reached when $\rho\approx 0.7433$. Thus, the {\em minimal shared memory utilization} is about $92\%$. To increase the
utilization, one should take the parameter $\rho$ closer to $0$ or $1$.

The second curve in Figure \ref{shmimgqindn.fig} represents the rate with which the necessity of shared memory emerges,
calculated as $\frac{\tilde{\varphi}_2'}{\widetilde{\it SJ}_2'}$, as a function of $\rho$. The rate tends to $0$ both
when $\rho$ approaches $0$ and when $\rho$ approaches $1$. The maximal value $0.0751$ of the rate is reached when
$\rho\approx 0.7743$. The {\em maximal rate with which the necessity of shared memory emerges} is about $\frac{1}{13}$.
To decrease the rate, one must take the parameter $\rho$ closer to $0$ or $1$.

The third curve in Figure \ref{shmimgqindn.fig} represents the steady-state probability of the shared memory request
from two processors, calculated as $\tilde{\varphi}_2'\widetilde{\mathcal{P}}_{25}'$, where
$\widetilde{\mathcal{P}}_{25}'= \sum_{\{A,\widetilde{\mathcal{K}}\mid\{\{r\},\{r\}\}\subseteq A,\
\widetilde{\mathcal{K}}_2\stackrel{A}{\rightarrow}\widetilde{\mathcal{K}}\}}PM_A(\widetilde{\mathcal{K}}_2,
\widetilde{\mathcal{K}})=PM(\widetilde{\mathcal{K}}_2,\widetilde{\mathcal{K}}_5)$, as function of $\rho$. One can see
that the probability tends to $0$ both when $\rho$ approaches $0$ and when $\rho$ approaches $1$. The maximal value
$0.0517$ of the probability is reached when $\rho\approx 0.8484$. To decrease the mentioned probability, one should
take the parameter $\rho$ closer to $0$ or $1$.

The fourth curve in Figure \ref{shmimgqindn.fig} represents the steady-state probability of the shared memory request
from a processor, calculated as $\tilde{\varphi}_2'\widetilde{\Sigma}_2'+\tilde{\varphi}_4'\widetilde{\Sigma}_4'$, as
a function of $\rho$, where $\widetilde{\Sigma}_i'\!\!=\!\!\sum_{\{A,\widetilde{\mathcal{K}}\mid\{r\}\in A,\
\widetilde{\mathcal{K}}_i\stackrel{A}{\rightarrow}\widetilde{\mathcal{K}}\}}\!\!PM_A(\widetilde{\mathcal{K}}_i,
\widetilde{\mathcal{K}}),\\
i\in\{2,4\}$. One can see that the probability tends to $0$ when $\rho$ approaches $0$ and it tends to $1$ when $\rho$
approaches $1$. To increase the probability, one should take the parameter $\rho$ closer to $1$.

\begin{figure}
\begin{center}
\includegraphics[width=0.8\textwidth]{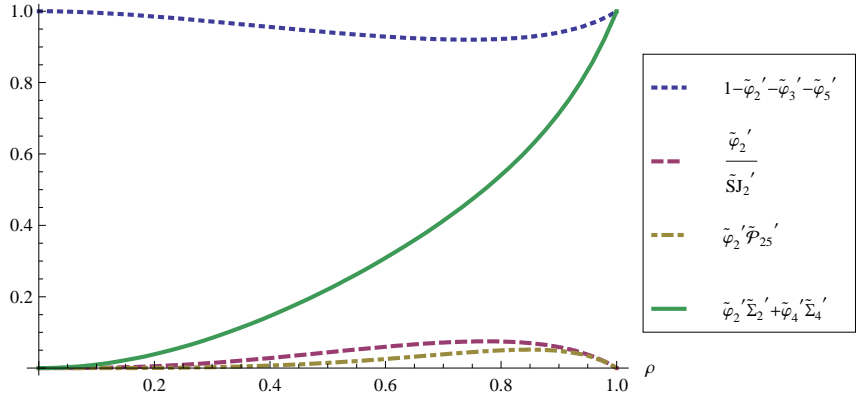}
\end{center}
\vspace{-8mm}
\caption{Some performance indices as functions of the parameter $\rho$.}
\label{shmimgqindn.fig}
\end{figure}

\section{Related work}
\label{relwork.sec}

Let us consider differences and similarities between dtsiPBC and other well-known SPAs.

\subsection{Continuous time and interleaving semantics}

Let us compare dtsiPBC with the classical interleaving SPAs.

Markovian Timed Processes for Performance Evaluation (MTIPP) \cite{HR94} specifies every activity as a pair consisting
of the action name (including the symbol $\tau$ for the {\em internal}, invisible action) and the parameter of
exponential distribution of the action delay (the {\em rate}). The interleaving operational semantics is defined on the
basis of Markovian (i.e. extended with the specification of rates) labeled transition systems. The interleaving
behaviour is here because the exponential PDF is a continuous one and simultaneous execution of any two activities has
zero probability according to the properties of continuous distributions. CTMCs can be derived from the
transition systems to analyze~performance.

Performance Evaluation Process Algebra (PEPA) \cite{Hil96} treats the activities as pairs consisting of action types
(including the {\em unknown} type $\tau$) and activity rates. The rate is either the parameter of exponential
distribution of the activity duration or it is {\em unspecified}. An activity with unspecified rate is {\em passive} by
its action type. The operational semantics is interleaving, it is defined via the extension of labeled transition
systems with a possibility to specify activity rates. Based on the transition systems, the continuous time Markov
processes (CTMPs) are generated which are used for performance evaluation with the help of the embedded continuous time
Markov chains (ECTMCs). In \cite{GHKR03}, a denotational semantics of PEPA has been proposed via PEPA nets that are
high-level CTSPNs with coloured tokens (coloured CTSPNs), from which the underlying CTMCs can be retrieved.

Extended Markovian Process Algebra (EMPA) \cite{BGo98} interprets each action as a pair consisting of its type and
rate. Actions can be {\em external} or {\em internal} (denoted by $\tau$) according to types. There are three kinds of
actions according to rates: {\em timed} ones with exponentially distributed durations (essentially, the actions from
MTIPP and PEPA), {\em immediate} ones with priorities and weights (the actions analogous to immediate transitions of
GSPNs) and {\em passive} ones (similar to passive actions of PEPA). The operational semantics is interleaving and based
on the labeled transition systems enriched with the information about action rates. For the exponentially timed kernel
of the algebra (the sublanguage including only exponentially timed and passive actions), it is possible to construct
CTMCs from the transition systems of the process terms to analyze the performance. In \cite{BDGo98,Bern99}, a
denotational semantics of EMPA based on GSPNs has been defined, from which one can also extract the underlying SMCs and
CTMCs (when both immediate and timed transitions are present) or DTMCs (but when there are only immediate transitions).

dtsiPBC considers every activity as a pair consisting of the multiaction (not just an action, as in the classical SPAs)
as a first element. The second element is either the probability (not the rate, as in the classical SPAs) to execute
the multiaction independently (the activity is called a stochastic multiaction in this case) or the weight expressing
how important is the execution of this multiaction (then the activity is called an immediate multiaction). Immediate
multiactions in dtsiPBC are similar to immediate actions in EMPA, but all the immediate multiactions in
dtsiPBC have the same high priority (with the goal to execute them always before stochastic multiactions, all having
the same low priority), whereas the immediate actions in EMPA can have different priorities. Associating the same
priority with all immediate multiactions in dtsiPBC results in the simplified specification and analysis, and such a
decision is appropriate to the calculus, since weights (assigned also to immediate actions in EMPA) are enough to
denote preferences among immediate multiactions and to produce the conformable probabilistic behaviours. There are no
immediate actions in MTIPP and PEPA. Immediate actions are available only in iPEPA \cite{HBC13}, where they are
analogous to immediate multiactions in dtsiPBC, and in a variant of TIPP \cite{GHR93} discussed while constructing the
calculus PM-TIPP \cite{Ret95}, but there immediate activities are used just to specify probabilistic branching and they
cannot be synchronized. dtsiPBC has a discrete time semantics, and residence time in the tangible states is
geometrically distributed, unlike the classical SPAs with continuous time semantics and exponentially distributed
activity delays. As a consequence, dtsiPBC has a step operational semantics in contrast to interleaving operational
semantics of the classical SPAs. The performance in dtsiPBC is analyzed via the underlying SMCs and (reduced) DTMCs
\cite{TMV15} extracted from the labeled probabilistic transition systems associated with the expressions. In the
classical SPAs, CTMCs are usually used for performance evaluation. dtsiPBC has a denotational semantics based on
LDTSIPNs from which the underlying SMCs and (reduced) DTMCs are derived, unlike (reduced) CTMCs in PEPA and EMPA. MTIPP
has no denotational semantics.

\subsection{Continuous time and non-interleaving semantics}

A few non-interleaving SPAs were considered among non-Markovian ones \cite{KA01,BA04}.

Generalized Stochastic Process Algebra (GSPA) \cite{BKLL95} is a stochastic extension of Simple Process Algebra
\cite{BKLL95}. GSPA has no operational semantics. GSPA has a true-concurrent denotational semantics via generalized
stochastic event structures (GSESs) with non-Markovian stochastic delays of events. In \cite{KBLL96}, generalized
semi-Markov processes (GSMPs) were extracted from GSESs to analyze performance.

Generalized Stochastic $\pi$-calculus (S$\pi$) \cite{Pri96,Pri02} extends $\pi$-calculus \cite{MPW92}. S$\pi$ allows
for general continuous distributions of activity delays. It has a proved operational semantics with transitions labeled
by encodings of their deduction trees. The transition labels encode the action causality information and allow one to
derive the enabling relations and the firing distributions of concurrent transitions from the transition sequences.
Nevertheless, abstracting from stochastic delays leads to the classical early interleaving semantics of $\pi$-calculus.
No well-established underlying performance model for this version of S$\pi$ exists.

Generalized Semi-Markovian Process Algebra (GSMPA) \cite{BBGo98,Bra02} is an enrichment of EMPA. GSMPA has an
ST-operational semantics and non-Markovian action delays. The ST-operational semantics of GSMPA is based on decorated
transition systems governed by transition rules with rather complex preconditions. There are two types of transitions:
the choice (action beginning) and the termination (action ending) ones. The choice transitions are labeled by weights
of single actions chosen for execution while the termination transitions have no labels. Only single actions can begin,
but several actions can end in parallel. Thus, the choice transitions happen just sequentially while the termination
transitions can happen simultaneously. As a result, the decorated interleaving / step transition systems are obtained.
The performance analysis in GSMPA is accomplished via GSMPs.

dtsiPBC has immediate multiactions while GSPA, S$\pi$ and GSMPA do not specify instantaneous events or activities.
Geometrically distributed or zero delays are associated with process states in dtsiPBC, unlike generally distributed
delays assigned to events in GSPA or to activities in S$\pi$ and GSMPA. dtsiPBC has a discrete time operational
semantics allowing for concurrent execution of activities in steps. GSPA has no operational semantics while S$\pi$ and
GSMPA have continuous time ones. In continuous time semantics, concurrency is simulated by interleaving, since
simultaneous occurrence of any two events has zero probability according to the properties of continuous probability
distributions. Therefore, interleaving transitions should be annotated with an additional information to keep the
concurrency. dtsiPBC has an SPN-based denotational semantics. In comparison with event structures, PNs are more
expressive and visually tractable formalism, capable of finitely specifying an infinite behaviour. Recursion in GSPA
produces infinite GSESs while dtsiPBC has iteration operation with a finite SPN semantics. Identification of infinite
GSESs that can be finitely represented in GSPA was left for a future research.

\subsection{Discrete time}

Much fewer SPAs with discrete time semantics were constructed.

Dts-nets \cite{AHR00} are a class of compositional DTSPNs with generally distributed discrete time transition delays.
The denotational semantics of a stochastic extension (we call it stochastic ACP or sACP) of a subset of Algebra of
Communicating Processes (ACP) \cite{BK85} can be constructed via dts-nets. There are two types of transitions:
immediate (timeless) ones, with zero delays, and time ones, whose delays are random variables with general discrete
distributions. The top-down synthesis of dts-nets consists in the substitution of their transitions by blocks
(dts-subnets) corresponding to some composition operators. It was explained how to calculate the throughput time of
dts-nets using the service time (holding time or delay) of their transitions. For this, the notions of service
distribution for the transitions and throughput distribution for the building blocks were defined. Since the throughput
time of the parallelism block was calculated as the maximal service time for its two constituting transitions, the
analogue of the step semantics was implemented.

Theory of Communicating Processes with discrete stochastic time ($TCP^{\rm dst}$) \cite{MVi08,MVi09}, later called
Theory of Communicating Processes with discrete real and stochastic time ($TCP^{\rm drst}$) \cite{MABV12}, is another
stochastic extension of ACP. $TCP^{\rm dst}$ has discrete real time (deterministic) delays (including zero time delays)
and discrete stochastic time delays. The algebra generalizes real time processes to discrete stochastic time ones by
applying real time properties to stochastic time and imposing race condition to real time semantics. $TCP^{\rm dst}$
has an interleaving operational semantics in terms of stochastic transition systems. The performance is analyzed via
discrete time probabilistic reward graphs which are essentially the reward transition systems with probabilistic states
having finite number of outgoing probabilistic transitions and timed states having a single outgoing timed transition.
The mentioned graphs can be transformed by unfolding or geometrization into discrete time Markov reward chains (DTMRCs)
appropriate for transient or stationary analysis.

dtsiPBC, sACP and $TCP^{\rm dst}$, all have zero delays. However, discrete time delays in dtsiPBC are zeros or
geometrically distributed and associated with process states. The zero delays are possible just in vanishing states
while geometrically distributed delays are possible only in tangible states. For each tangible state, the parameter of
geometric distribution governing the delay in the state is completely determined by the probabilities of all stochastic
multiactions executable from it. In sACP and $TCP^{\rm dst}$, delays are generally distributed, but they are assigned
to transitions in sACP and separated from actions (excepting zero delays) in $TCP^{\rm dst}$. Moreover, a special
attention is given to zero delays in sACP and deterministic delays in $TCP^{\rm dst}$. In sACP, immediate (timeless)
transitions with zero delays serve as source and sink transitions of the dts-subnets corresponding to the choice,
parallelism and iteration operators. In $TCP^{\rm dst}$, zero delays of actions are specified by undelayable action
prefixes while positive deterministic delays of processes are specified with timed delay prefixes. Neither formal
syntax nor operational semantics for sACP were defined and it was not explained how to derive Markov chains from the
algebraic expressions or the corresponding dts-nets to analyze performance. It was not stated explicitly, which type of
semantics (interleaving or step) is accommodated in sACP. In spite of the discrete time approach, operational semantics
of $TCP^{\rm dst}$ is still interleaving, unlike that of dtsiPBC.
$TCP^{\rm dst}$ has no denotational semantics.

Table \ref{spaclassim.tab} summarizes the SPAs comparison above and that from Section \ref{introduction.sec}, by
classifying the SPAs according to the concept of time, the presence of immediate (multi)actions and the type of
operational semantics. The names of SPAs, whose denotational semantics is based on SPNs, are printed in bold font. The
underlying stochastic process (if defined) is specified near the name of the corresponding SPA.

\begin{table}
\caption{Classification of stochastic process algebras.}
\vspace{-1mm}
\label{spaclassim.tab}
\begin{center}
\small\begin{tabular}{|c|c||c|c|}
\hline
Time & Immediate & Interleaving & Non-interleaving\\
 & (multi)actions & semantics & semantics \\
\hline
Continuous & No & MTIPP (CTMC), {\bf PEPA} (CTMP), & GSPA (GSMP), S$\pi$,\\
 & & {\bf sPBC} (CTMC) & GSMPA (GSMP)\\
\cline{2-4}
 & Yes & {\bf EMPA} (SMC, CTMC), & ---\\
 & & {\bf gsPBC} (SMC) & \\
\hline
Discrete & No & --- & {\bf dtsPBC} (DTMC)\\
\cline{2-4}
 & Yes & $TCP^{\rm dst}$ (DTMRC) & {\bf sACP},\\
 & & & {\bf dtsiPBC} (SMC, DTMC)\\
\hline
\end{tabular}
\end{center}
\end{table}

\section{Discussion}
\label{discussion.sec}

Let us now discuss which advantages has dtsiPBC in comparison with the SPAs described in Section \ref{relwork.sec}.

\subsection{Analytical solution}

An important aspect is the analytical tractability of the underlying stochastic process, used for performance
evaluation in SPAs. The underlying CTMCs in MTIPP and PEPA, as well as SMCs in EMPA, are treated analytically, but
these continuous time SPAs have interleaving semantics. GSPA, S$\pi$ and GSMPA are the continuous time models, for
which a non-interleaving semantics is constructed, but for the underlying GSMPs in GSPA and GSMPA, only simulation and
numerical methods are applied, whereas no performance model for S$\pi$ is defined. sACP and $TCP^{dst}$ are the
discrete time models with the associated analytical methods for the throughput calculation in sACP or for the
performance evaluation based on the underlying DTMRCs in $TCP^{dst}$, but both models have interleaving semantics.
dtsiPBC is a discrete time model with a non-interleaving semantics, where analytical methods are applied to the
underlying SMCs. Hence, if an interleaving model is appropriate as a framework for the analytical solution towards
performance evaluation then one has a choice between the continuous time SPAs MTIPP, PEPA, EMPA and the discrete time
ones sACP, $TCP^{dst}$. Otherwise, if one needs a non-interleaving model with the associated analytical methods for
performance evaluation and the discrete time approach is feasible then dtsiPBC is the right choice.

The existence of an analytical solution also permits to interpret quantitative values (rates, probabilities, weights
etc.) from the system specifications as parameters, which can be adjusted to optimize the system performance, like in
dtsPBC, dtsiPBC and parametric probabilistic transition systems (i.e. DTMCs whose transition probabilities may be
real-value parameters) \cite{LMST07}. Note that DTMCs whose transition probabilities are parameters were introduced in
\cite{Daw04}. Parametric CTMCs with the transition rates treated as parameters were investigated in \cite{HKM08}. On
the other hand, no parameters in formulas of SPAs were considered in the literature so far. In dtsiPBC we can easily
construct examples with more parameters than we did in our case study. The performance indices will be then interpreted
as functions of several variables. The advantage of our approach is that, unlike of the method from \cite{LMST07}, we
should not impose to the parameters any special conditions needed to guarantee that the real values, interpreted as the
transition probabilities, always lie in the interval $[0;1]$. To be convinced of this fact, just remember that, as we
have demonstrated, the positive probability functions $PF,\ PT,\ PM,\ PM^*$ define probability distributions, hence,
they always return
values belonging to $(0;1]$ for any probability parameters from $(0;1)$ and weight parameters from $\reals_{>0}$.
In addition, the transition constraints (their probabilities, rates and guards), calculated using the parameters, in
our case should not always be polynomials over variables-parameters, as often required in the mentioned papers, but
they may also be fractions of polynomials, like in our case study.

\subsection{Application area}

From the application viewpoint,
MTIPP and PEPA are well-suited for interleaving continuous time systems,
in which the activity rates or the average sojourn time in the states are known in advance and exponential distribution
approximates well the activity delay distributions.
EMPA, however, can be used to model the mentioned systems with the activity delays of different duration order or the
extended systems, in which purely probabilistic choices or urgent activities must be implemented. GSPA and GSMPA fit
well for modeling continuous time systems with a capability to keep the activity causality information, and with known
activity delay distributions, which cannot be approximated accurately by exponential
distributions.
S$\pi$ can additionally model mobility in such systems. $TCP^{dst}$ is a good choice for interleaving discrete time
systems with deterministic (fixed) and generalized stochastic delays, whereas sACP is capable to model non-interleaving
systems as well, but it offers not enough performance analysis methods. dtsiPBC is consistent for the step discrete
time systems such that the independent execution probabilities of activities are known and geometrical distribution
approximates well the state residence time distributions. In addition, dtsiPBC can model these systems featuring very
scattered activity delays or even more complex systems with instantaneous probabilistic choice or urgency, hence,
dtsiPBC can be taken as a non-interleaving discrete time counterpart of EMPA.

\subsection{Concurrency interpretation}

The stochastic process calculi proposed in the literature are based on interleaving, as a rule, and parallelism is
simulated by synchronous or asynchronous execution. As a semantic domain, the interle\-aving formalism of transition
systems is often used. However,
to properly support intuition of the behaviour of concurrent and distributed systems, their semantics should treat
parallelism as a primitive concept that cannot be reduced to nondeterminism. Moreover, in interleaving semantics, some
important properties of these systems cannot be expressed, such as simultaneous occurrence of concurrent transitions
\cite{DP99} or local deadlock in the spatially distributed processes \cite{MPY96}. Therefore, investigation of
stochastic extensions for more expressive and powerful algebraic calculi is an important issue. The development of step
or ``true concurrency'' (such that parallelism is considered as a causal independence) SPAs is an interesting and
nontrivial problem, which has attracted special attention in the last years. Nevertheless, not so many formal
stochastic models were defined whose underlying stochastic processes
were based on DTMCs. As mentioned in \cite{Fou10}, such models are more difficult to analyze, since a lot of events can
occur simultaneously in discrete time systems (the models have a step semantics) and the probability of a set of events
cannot be easily related to the probability of the single ones. As observed in \cite{HPRV12}, even for stochastic
models with generally distributed time delays, some restrictions on the concurrency degree were imposed to simplify
their analysis techniques. In particular, the enabling restriction requires that no two generally distributed
transitions are enabled in any reachable marking. Hence, their activity periods do not intersect and no two such
transitions can fire simultaneously, this results in interleaving semantics of the model.

Stochastic models with discrete time and step semantics have the following important advantage over those having just
an interleaving semantics. The underlying Markov chains of parallel stochastic processes have the additional
transitions corresponding to the simultaneous execution of concurrent (i.e. non-synchronized) activities.
These additional transitions allow us one to bypass a lot of intermediate states, which otherwise should be visited
when interleaving semantics is accommodated. When step semantics is used, the intermediate states can also be
visited with some probability (this is an advantage, since some alternative system's behaviour may start from these
states), but this probability is not greater than the corresponding one in case of interleaving semantics. While in
interleaving semantics, only the empty or singleton (multi)sets of activities can be executed, in step semantics,
generally, the (multi)sets of activities with more than one element can be executed as well. Hence, in step semantics,
there are more variants of execution from each state than in the interleaving case and the executions probabilities,
whose sum
must be equal to $1$, are distributed among more possibilities. Therefore, the systems with parallel stochastic
processes usually have smaller average run-through.
Thus, when the underlying Markov chains of the processes are ergodic, they will take less discrete time units to
stabilize the behaviour, since their TPMs will be denser because of additional non-zero elements outside the main
diagonal. Hence, both the first passage-time performance indices based on the transient probabilities and the
steady-state performance indices based on the stationary probabilities can be computed quicker, resulting in faster
quantitative analysis of the systems. On the other hand, step semantics, induced by simultaneous firing several
transitions at each step, is natural for Petri nets and allows one to exploit full power of the model. Therefore, it is
important to respect the probabilities of parallel executions of activities in discrete time SPAs, especially in those
with a Petri net denotational semantics.

\subsection{Advantages of dtsiPBC}

The advantages of dtsiPBC are the flexible multiaction labels, immediate multiactions, powerful operations, as well as
a step operational and a Petri net denotational semantics allowing for concurrent execution of activities
(transitions), together with an ability for analytical and parametric performance evaluation.

\section{Conclusion}
\label{conclusion.sec}

In this paper, we have proposed a discrete time stochastic extension dtsiPBC of a finite part of PBC enriched with
iteration and immediate multiactions. The calculus has a concurrent step operational semantics based on labeled
probabilistic transition systems and a denotational semantics in terms of a subclass of LDTSIPNs. A method of
performance evaluation in the framework of the calculus has been presented. Step stochastic bisimulation equivalence of
process expressions has been defined and its interrelations with other equivalences of the calculus have been
investigated. We have explained how to reduce transition systems and underlying SMCs of expressions w.r.t. the
introduced equivalence. We have proved that the mentioned equivalence guarantees identity of the stationary behaviour
and the sojourn time properties, and thus preserves performance measures. A case study of a generalization of the
shared memory system by allowing for variable probabilites in its specification has been presented. The case study is
an example of modeling, performance evaluation and performance preserving reduction within the calculus.

The advantage of our framework is twofold. First, one can specify in it concurrent composition and synchronization of
(multi)actions, whereas this is not possible in classical Markov chains. Second, algebraic formulas represent processes
in a more compact way than Petri nets and allow one to apply syntactic transformations and comparisons. Process
algebras are compositional by definition and their operations naturally correspond to operators of programming
languages. Hence, it is much easier to construct a complex model in the algebraic setting than in PNs. The complexity
of PNs generated for practical models in the literature demonstrates that it is not straightforward to construct such
PNs directly from the system specifications. dtsiPBC is well suited for the discrete time applications, such as
business processes, neural and transportation networks, computer and communication systems, web services, whose
discrete states change with a global time tick, and in which the distributed architecture or the concurrency level
should be preserved (remember that, in step semantics, we have additional transitions due to concurrent executions).

Future work will consist in constructing a congruence for dtsiPBC, i.e. the equivalence that withstands application of
all its operations. A possible candidate is a stronger version of $\bis_{\rm ss}$ defined via transition systems
equipped with two extra transitions {\sf skip} and {\sf redo}, like those from \cite{MVCF08}. We also plan to extend
the calculus with deterministically timed multiactions having a fixed discrete time delay (including the zero one which
is the case of immediate multiactions) to enhance expressiveness of the calculus and extend application area of the
associated analysis techniques. The resulting SPA will be a concurrent discrete time analogue of SM-PEPA \cite{Brad05},
whose underlying stochastic model is a semi-Markov chain. Finally, recursion could be added to dtsiPBC to increase its
specification power.

%
%
\nocite{*}
\bibliographystyle{abbrvnat}
\bibliography{dtsipbcdmtcs}

\appendix

\section{Proofs}
\label{proofs.sec}

\subsection{Proof of Proposition \protect\ref{largestbisim.pro}}
\label{largestbisim.ssc}

Like
for strong equivalence in Proposition 8.2.1 from \cite{Hil96}, we shall prove the following fact about step stochastic
bisimulation. Let us have $\forall j\in\mathcal{J},\ \mathcal{R}_j:G\bis_{\rm ss}G'$ for some index set $\mathcal{J}$.
Then the transitive closure of the union of all relations $\mathcal{R}=(\cup_{j\in\mathcal{J}}\mathcal{R}_j)^+$ is also
an equivalence and $\mathcal{R}:G\bis_{\rm ss}G'$.

Since $\forall j\in\mathcal{J},\ \mathcal{R}_j$ is an equivalence, by definition of $\mathcal{R}$, we get that
$\mathcal{R}$ is also an equivalence. Let $j\in\mathcal{J}$, then, by definition of $\mathcal{R},\
(s_1,s_2)\in\mathcal{R}_j$ implies $(s_1,s_2)\in\mathcal{R}$. Hence, $\forall\mathcal{H}_{jk}\in (DR(G)\cup
DR(G'))/_{\mathcal{R}_j},\ \exists\mathcal{H}\in (DR(G)\cup DR(G'))/_\mathcal{R},\
\mathcal{H}_{jk}\subseteq\mathcal{H}$. Moreover, $\exists\mathcal{J}',\
\mathcal{H}=\cup_{k\in\mathcal{J}'}\mathcal{H}_{jk}$.

We denote $\mathcal{R}(n)=(\cup_{j\in\mathcal{J}}\mathcal{R}_j)^n$. Let $(s_1,s_2)\in\mathcal{R}$, then, by definition
of $\mathcal{R},\ \exists n>0,\ (s_1,s_2)\in\mathcal{R}(n)$. We shall prove that $\mathcal{R}:G\bis_{\rm ss}G'$ by
induction on $n$. It is clear that $\forall j\in\mathcal{J},\ \mathcal{R}_j:G\bis_{\rm ss}G'$ implies $\forall
j\in\mathcal{J},\ ([G]_\approx ,[G']_\approx )\in\mathcal{R}_j$ and we have $([G]_\approx ,[G']_\approx
)\in\mathcal{R}$ by definition of $\mathcal{R}$. It remains to prove that $(s_1,s_2)\in\mathcal{R}$ implies
$\forall\mathcal{H}\in (DR(G)\cup DR(G'))/_\mathcal{R},\ \forall A\in\naturals_{\rm fin}^\mathcal{L},\
PM_A(s_1,\mathcal{H})=PM_A(s_2,\mathcal{H})$.
\begin{itemize}

\item $n=1$\\
In this case, $(s_1,s_2)\in\mathcal{R}$ implies $\exists j\in\mathcal{J},\ (s_1,s_2)\in\mathcal{R}_j$. Since
$\mathcal{R}_j:G\bis_{\rm ss}G'$, we get\\
$\forall\mathcal{H}\in (DR(G)\cup DR(G'))/_\mathcal{R},\ \forall A\in\naturals_{\rm fin}^\mathcal{L}$,
$$PM_A(s_1,\mathcal{H})=\sum_{k\in\mathcal{J}'}PM_A(s_1,\mathcal{H}_{jk})=
\sum_{k\in\mathcal{J}'}PM_A(s_2,\mathcal{H}_{jk})=PM_A(s_2,\mathcal{H}).$$

\item $n\rightarrow n+1$\\
Suppose that $\forall m\leq n,\ (s_1,s_2)\in\mathcal{R}(m)$ implies $\forall\mathcal{H}\in (DR(G)\cup
DR(G'))/_\mathcal{R},\ \forall A\in\naturals_{\rm fin}^\mathcal{L},\\
PM_A(s_1,\mathcal{H})=PM_A(s_2,\mathcal{H})$. Then $(s_1,s_2)\in\mathcal{R}(n+1)$ implies $\exists
j\in\mathcal{J},\ (s_1,s_2)\in\mathcal{R}_j\circ \mathcal{R}(n)$, i.e. $\exists s_3\in (DR(G)\cup DR(G'))$, such
that $(s_1,s_3)\in\mathcal{R}_j$ and $(s_3,s_2)\in\mathcal{R}(n)$. Then, like for the case $n=1$, we get
$PM_A(s_1,\mathcal{H})=PM_A(s_3,\mathcal{H})$. By the induction hypothesis, we get
$PM_A(s_3,\mathcal{H})=PM_A(s_2,\mathcal{H})$. Thus, $\forall\mathcal{H}\in (DR(G)\cup DR(G'))/_\mathcal{R},\
\forall A\in\naturals_{\rm fin}^\mathcal{L}$,
$$PM_A(s_1,\mathcal{H})=PM_A(s_3,\mathcal{H})=PM_A(s_2,\mathcal{H}).$$

\end{itemize}

By definition, $\mathcal{R}_{\rm ss}(G,G')$ is at least as large as the largest step stochastic bisimulation between
$G$ and $G'$. It follows from above that $\mathcal{R}_{\rm ss}(G,G')$ is an equivalence and $\mathcal{R}_{\rm
ss}(G,G'):G\bis_{\rm ss}G'$, hence, it is the largest step stochastic bisimulation between $G$ and $G'$. \qed

\subsection{Proof of Proposition \protect\ref{statprob.pro}}
\label{statprob.ssc}

By Proposition \ref{bissplit.pro}, $(DR(G)\cup DR(G'))/_\mathcal{R}\!\!=\!\!((DR_{\rm T}(G)\cup DR_{\rm
T}(G'))/_\mathcal{R})\uplus ((DR_{\rm V}(G)\cup DR_{\rm V}(G'))/_\mathcal{R})$. Hence, $\forall\mathcal{H}\in
(DR(G)\cup DR(G'))/_\mathcal{R}$, all states from $\mathcal{H}$ are tangible, when\\
$\mathcal{H}\in (DR_{\rm T}(G)\cup DR_{\rm T}(G'))/_\mathcal{R}$, or all of them are vanishing, when $\mathcal{H}\in
(DR_{\rm V}(G)\cup DR_{\rm V}(G'))/_\mathcal{R}$.

By definition of the steady-state PMFs for SMCs, $\forall s\in DR_{\rm V}(G),\ \varphi (s)=0$ and $\forall s'\in
DR_{\rm V}(G'),\\
\varphi '(s')=0$. Thus, $\forall\mathcal{H}\in (DR_{\rm V}(G)\cup DR_{\rm V}(G'))/_\mathcal{R},\
\sum_{s\in\mathcal{H}\cap DR(G)}\varphi (s)=\sum_{s\in\mathcal{H}\cap DR_{\rm V}(G)}\varphi (s)=0=\\
\sum_{s'\in\mathcal{H}\cap DR_{\rm V}(G')}\varphi '(s')=\sum_{s'\in\mathcal{H}\cap DR(G')}\varphi '(s')$.

By Proposition \ref{pmfsmc.pro}, $\forall s\in DR_{\rm T}(G),\ \varphi (s)=\frac{\psi (s)}{\sum_{\tilde{s}\in
DR_{\rm T}(G)}\psi (\tilde{s})}$ and $\forall s'\in DR_{\rm T}(G'),\\
\varphi '(s')=\frac{\psi '(s')}{\sum_{\tilde{s}'\in DR_{\rm T}(G')}\psi '(\tilde{s}')}$, where $\psi$ and $\psi '$ are
the steady-state PMFs for ${\it DTMC}(G)$ and ${\it DTMC}(G')$, respectively. Thus,
$\forall\mathcal{H},\widetilde{\mathcal{H}}\in (DR_{\rm T}(G)\cup DR_{\rm T}(G'))/_\mathcal{R},\
\sum_{s\in\mathcal{H}\cap DR(G)}\varphi (s)=\sum_{s\in\mathcal{H}\cap DR_{\rm T}(G)}\varphi
(s)=\sum_{s\in\mathcal{H}\cap DR_{\rm T}(G)}\left(\frac{\psi (s)}{\sum_{\tilde{s}\in DR_{\rm T}(G)}\psi
(\tilde{s})}\right)=\frac{\sum_{s\in\mathcal{H}\cap DR_{\rm T}(G)}\psi (s)}{\sum_{\tilde{s}\in DR_{\rm T}(G)}\psi
(\tilde{s})}= \frac{\sum_{s\in\mathcal{H}\cap DR_{\rm T}(G)}\psi (s)}{\sum_{\widetilde{\mathcal{H}}}
\sum_{\tilde{s}\in\widetilde{\mathcal{H}}\cap DR_{\rm T}(G)}\psi (\tilde{s})}$ and\\
$\sum_{s'\in\mathcal{H}\cap DR(G')}\varphi '(s')=\sum_{s'\in\mathcal{H}\cap DR_{\rm T}(G')}\varphi '(s')=
\sum_{s'\in\mathcal{H}\cap DR_{\rm T}(G')}\left(\frac{\psi '(s')}{\sum_{\tilde{s}'\in DR_{\rm T}(G')}\psi
'(\tilde{s}')}\right)=\\
\frac{\sum_{s'\in\mathcal{H}\cap DR_{\rm T}(G')}\psi '(s')}{\sum_{\tilde{s}'\in DR_{\rm T}(G')}\psi '(\tilde{s}')}=
\frac{\sum_{s'\in\mathcal{H}\cap DR_{\rm T}(G')}\psi '(s')}{\sum_{\widetilde{\mathcal{H}}}
\sum_{\tilde{s}'\in\widetilde{\mathcal{H}}\cap DR_{\rm T}(G')}\psi '(\tilde{s}')}$.

It remains to prove that $\forall\mathcal{H}\in (DR_{\rm T}(G)\cup DR_{\rm T}(G'))/_\mathcal{R},\
\sum_{s\in\mathcal{H}\cap DR_{\rm T}(G)}\psi (s)=\\
\sum_{s'\in\mathcal{H}\cap DR_{\rm T}(G')}\psi '(s')$. Since $(DR(G)\cup DR(G'))/_\mathcal{R}=((DR_{\rm T}(G)\cup
DR_{\rm T}(G'))/_\mathcal{R})\uplus ((DR_{\rm V}(G)\cup DR_{\rm V}(G'))/_\mathcal{R})$, the previous equality is a
consequence of the following: $\forall\mathcal{H}\in (DR(G)\cup DR(G'))/_\mathcal{R},\\
\sum_{s\in\mathcal{H}\cap DR(G)}\psi (s)=\sum_{s'\in\mathcal{H}\cap DR(G')}\psi '(s')$.

It is sufficient to prove the previous statement for transient PMFs only, since $\psi=\lim_{k\to\infty}\psi[k]$ and
$\psi '=\lim_{k\to\infty}\psi '[k]$. We proceed by induction on $k$.
\begin{itemize}

\item $k=0$\\
The only nonzero values of the initial PMFs of ${\it DTMC}(G)$ and ${\it DTMC}(G')$ are $\psi [0]([G]_\approx )$
and $\psi [0]([G']_\approx )$. Let $\mathcal{H}_0$ be the equivalence class containing $[G]_\approx$ and
$[G']_\approx$. Then\\
$\sum_{s\in\mathcal{H}_0\cap DR(G)}\psi [0](s)=\psi [0]([G]_\approx )=1=\psi '[0]([G']_\approx )=
\sum_{s'\in\mathcal{H}_0\cap DR(G')}\psi '[0](s')$. As for other equivalence classes, $\forall\mathcal{H}\in
((DR(G)\cup DR(G'))/_\mathcal{R})\setminus\mathcal{H}_0$, we have $\sum_{s\in\mathcal{H}\cap DR(G)}\psi
[0](s)=0=\sum_{s'\in\mathcal{H}\cap DR(G')}\psi '[0](s')$.

\item $k\rightarrow k+1$\\
Let $\mathcal{H}\!\in\!(DR(G)\cup DR(G'))/_\mathcal{R}$ and $s_1,s_2\!\in\!\mathcal{H}$. We have
$\forall\widetilde{\mathcal{H}}\!\in\!(DR(G)\cup DR(G'))/_\mathcal{R},\
\forall A\!\in\!\naturals_{\rm fin}^\mathcal{L},\\
s_1\stackrel{A}{\rightarrow}_\mathcal{P}\widetilde{\mathcal{H}}\ \Leftrightarrow\
s_2\stackrel{A}{\rightarrow}_\mathcal{P}\widetilde{\mathcal{H}}$. Therefore, $PM(s_1,\widetilde{\mathcal{H}})=
\sum_{\{\Upsilon\mid\exists\tilde{s}_1\in\widetilde{\mathcal{H}},\
s_1\stackrel{\Upsilon}{\rightarrow}\tilde{s}_1\}}PT(\Upsilon ,s_1)=\\
\sum_{A\in\naturals_{\rm fin}^\mathcal{L}}\sum_{\{\Upsilon\mid\exists\tilde{s}_1\in\widetilde{\mathcal{H}},\
s_1\stackrel{\Upsilon}{\rightarrow}\tilde{s}_1,\ \mathcal{L}(\Upsilon )=A\}}PT(\Upsilon ,s_1)=
\sum_{A\in\naturals_{\rm fin}^\mathcal{L}}PM_A(s_1,\widetilde{\mathcal{H}})=\\
\sum_{A\in\naturals_{\rm fin}^\mathcal{L}}PM_A(s_2,\widetilde{\mathcal{H}})=
\sum_{A\in\naturals_{\rm fin}^\mathcal{L}}\sum_{\{\Upsilon\mid\exists\tilde{s}_2\in\widetilde{\mathcal{H}},\
s_2\stackrel{\Upsilon}{\rightarrow}\tilde{s}_2,\ \mathcal{L}(\Upsilon )=A\}}PT(\Upsilon ,s_2)=\\
\sum_{\{\Upsilon\mid\exists\tilde{s}_2\in\widetilde{\mathcal{H}},\
s_2\stackrel{\Upsilon}{\rightarrow}\tilde{s}_2\}}PT(\Upsilon ,s_2)=PM(s_2,\widetilde{\mathcal{H}})$. Since this
equality is valid for all $s_1,s_2\in\mathcal{H}$, we can denote
$PM(\mathcal{H},\widetilde{\mathcal{H}})=PM(s_1,\widetilde{\mathcal{H}})=PM(s_2,\widetilde{\mathcal{H}})$.
Transitions from the states of $DR(G)$ always lead to those from the same set, hence, $\forall s\in DR(G),\
PM(s,\widetilde{\mathcal{H}})=PM(s,\widetilde{\mathcal{H}}\cap DR(G))$. The same holds for $DR(G')$.

By induction hypothesis, $\sum_{s\in\mathcal{H}\cap DR(G)}\psi [k](s)=\sum_{s'\in\mathcal{H}\cap DR(G')}\psi
'[k](s')$. Further,\\
$\sum_{\tilde{s}\in\widetilde{\mathcal{H}}\cap DR(G)}\psi [k+1](\tilde{s})=
\sum_{\tilde{s}\in\widetilde{\mathcal{H}}\cap DR(G)}\sum_{s\in DR(G)}\psi [k](s)PM(s,\tilde{s})=\\
\sum_{s\in DR(G)}\sum_{\tilde{s}\in\widetilde{\mathcal{H}}\cap DR(G)}\psi [k](s)PM(s,\tilde{s})=
\sum_{s\in DR(G)}\psi [k](s)\sum_{\tilde{s}\in\widetilde{\mathcal{H}}\cap DR(G)}PM(s,\tilde{s})=\\
\sum_\mathcal{H}\sum_{s\in\mathcal{H}\cap DR(G)}\psi [k](s)
\sum_{\tilde{s}\in\widetilde{\mathcal{H}}\cap DR(G)}PM(s,\tilde{s})=\\
\sum_\mathcal{H}\sum_{s\in\mathcal{H}\cap DR(G)}\psi [k](s)\sum_{\tilde{s}\in\widetilde{\mathcal{H}}\cap
DR(G)}\sum_{\{\Upsilon\mid s\stackrel{\Upsilon}{\rightarrow}\tilde{s}\}}PT(\Upsilon ,s)=\\
\sum_\mathcal{H}\sum_{s\in\mathcal{H}\cap DR(G)}\psi
[k](s)\sum_{\{\Upsilon\mid\exists\tilde{s}\in\widetilde{\mathcal{H}}\cap
DR(G),\ s\stackrel{\Upsilon}{\rightarrow}\tilde{s}\}}PT(\Upsilon ,s)=\\
\sum_\mathcal{H}\sum_{s\in\mathcal{H}\cap DR(G)}\psi [k](s)PM(s,\widetilde{\mathcal{H}})=
\sum_\mathcal{H}\sum_{s\in\mathcal{H}\cap DR(G)}\psi [k](s)PM(\mathcal{H},\widetilde{\mathcal{H}})=\\
\sum_\mathcal{H}PM(\mathcal{H},\widetilde{\mathcal{H}})\sum_{s\in\mathcal{H}\cap DR(G)}\psi [k](s)=
\sum_\mathcal{H}PM(\mathcal{H},\widetilde{\mathcal{H}})\sum_{s'\in\mathcal{H}\cap DR(G')}\psi '[k](s')=\\
\sum_\mathcal{H}\sum_{s'\in\mathcal{H}\cap DR(G')}\psi '[k](s')PM(\mathcal{H},\widetilde{\mathcal{H}})=
\sum_\mathcal{H}\sum_{s'\in\mathcal{H}'\cap DR(G')}\psi '[k](s')PM(s',\widetilde{\mathcal{H}})=\\
\sum_\mathcal{H}\sum_{s'\in\mathcal{H}\cap DR(G')}\psi '[k](s')\sum_{\{\Upsilon\mid\exists\tilde{s}'\in
\widetilde{\mathcal{H}}\cap DR(G'),\ s'\stackrel{\Upsilon}{\rightarrow}\tilde{s}'\}}PT(\Upsilon ,s')=\\
\sum_\mathcal{H}\sum_{s'\in\mathcal{H}\cap DR(G')}\psi '[k](s')\sum_{\tilde{s}'\in\widetilde{\mathcal{H}}\cap
DR(G')}\sum_{\{\Upsilon\mid\exists\tilde{s}',\ s'\stackrel{\Upsilon}{\rightarrow}\tilde{s}'\}}PT(\Upsilon ,s')=\\
\sum_\mathcal{H}\sum_{s'\in\mathcal{H}\cap DR(G')}\psi '[k](s')\sum_{\tilde{s}'\in\widetilde{\mathcal{H}}\cap
DR(G')}PM(s',\tilde{s}')=\\
\sum_{s'\in DR(G')}\psi '[k](s')\sum_{\tilde{s}'\in\widetilde{\mathcal{H}}\cap DR(G')}PM(s',\tilde{s}')=\\
\sum_{s'\in DR(G')}\sum_{\tilde{s}'\in\widetilde{\mathcal{H}}\cap DR(G')}\psi '[k](s')PM(s',\tilde{s}')=\\
\sum_{\tilde{s}'\in\widetilde{\mathcal{H}}\cap DR(G')}\sum_{s'\in DR(G')}\psi '[k](s')PM(s',\tilde{s}')=
\sum_{\tilde{s}'\in\widetilde{\mathcal{H}}\cap DR(G')}\psi '[k+1](\tilde{s}')$. \qed

\end{itemize}

\subsection{Proof of Theorem \protect\ref{stattrace.the}}
\label{stattrace.ssc}

Let $\mathcal{H}\in (DR(G)\cup DR(G'))/_\mathcal{R}$ and $s,\bar{s}\in\mathcal{H}$. We have
$\forall\widetilde{\mathcal{H}}\in (DR(G)\cup DR(G'))/_\mathcal{R},\ \forall A\in\naturals_{\rm fin}^\mathcal{L},\\
s\stackrel{A}{\rightarrow}_\mathcal{P}\widetilde{\mathcal{H}}\ \Leftrightarrow\
\bar{s}\stackrel{A}{\rightarrow}_\mathcal{P}\widetilde{\mathcal{H}}$. Since this equality is valid for all
$s,\bar{s}\in\mathcal{H}$, we can rewrite it as
$\mathcal{H}\stackrel{A}{\rightarrow}_\mathcal{P}\widetilde{\mathcal{H}}$ and denote
$PM_A(\mathcal{H},\widetilde{\mathcal{H}})= PM_A(s,\widetilde{\mathcal{H}})=PM_A(\bar{s},\widetilde{\mathcal{H}})$. The
transitions from the states of $DR(G)$ always lead to those from the same set, hence, $\forall s\in DR(G),\
PM_A(s,\widetilde{\mathcal{H}})= PM_A(s,\widetilde{\mathcal{H}}\cap DR(G))$. The same holds for $DR(G')$.

Let $\Sigma =A_1\cdots A_n$ be a derived step trace of $G$ and $G'$. Then $\exists\mathcal{H}_0,\ldots
,\mathcal{H}_n\in (DR(G)\cup DR(G'))/_\mathcal{R},\\
\mathcal{H}_0\stackrel{A_1}{\rightarrow}_{\mathcal{P}_1}\mathcal{H}_1\stackrel{A_2}{\rightarrow}_{\mathcal{P}_2}\cdots
\stackrel{A_n}{\rightarrow}_{\mathcal{P}_n}\mathcal{H}_n$.
Let us prove that the sum of probabilities of all the paths starting in every $s_0\in\mathcal{H}_0$ and going through
the states from $\mathcal{H}_1,\ldots ,\mathcal{H}_n$ is equal to the product of $\mathcal{P}_1,\ldots ,\mathcal{P}_n$:
$$\sum_{\{\Upsilon_1,\ldots ,\Upsilon_n\mid s_0\stackrel{\Upsilon_1}{\rightarrow}\cdots
\stackrel{\Upsilon_n}{\rightarrow}s_n,\ \mathcal{L}(\Upsilon_i)=A_i,\ s_i\in\mathcal{H}_i\ (1\leq i\leq
n)\}}\prod_{i=1}^n PT(\Upsilon_i,s_{i-1})=\prod_{i=1}^n PM_{A_i}(\mathcal{H}_{i-1},\mathcal{H}_i).$$

We prove this equality by induction on the derived step trace length $n$.
\begin{itemize}

\item $n=1$\\
$\sum_{\{\Upsilon_1\mid s_0\stackrel{\Upsilon_1}{\rightarrow}s_1,\ \mathcal{L}(\Upsilon_1)=A_1,\
s_1\in\mathcal{H}_1\}}PT(\Upsilon_1,s_0)=PM_{A_1}(s_0,\mathcal{H}_1)=PM_{A_1}(\mathcal{H}_0,\mathcal{H}_1)$.

\item $n\rightarrow n+1$\\
$\sum_{\{\Upsilon_1,\ldots ,\Upsilon_n,\Upsilon_{n+1}\mid s_0\stackrel{\Upsilon_1}{\rightarrow}\cdots
\stackrel{\Upsilon_n}{\rightarrow}s_n\stackrel{\Upsilon_{n+1}}{\rightarrow}s_{n+1},\ \mathcal{L}(\Upsilon_i)=A_i,\
s_i\in\mathcal{H}_i\ (1\leq i\leq n+1)\}}\prod_{i=1}^{n+1}PT(\Upsilon_i,s_{i-1})=\\
\sum_{\{\Upsilon_1,\ldots ,\Upsilon_n\mid s_0\stackrel{\Upsilon_1}{\rightarrow}\cdots
\stackrel{\Upsilon_n}{\rightarrow}s_n,\ \mathcal{L}(\Upsilon_i)=A_i,\ s_i\in\mathcal{H}_i\ (1\leq i\leq n)\}}\\
\sum_{\{\Upsilon_{n+1}\mid s_n\stackrel{\Upsilon_{n+1}}{\rightarrow}s_{n+1},\ \mathcal{L}(\Upsilon_{n+1})=A_{n+1},\
s_n\in\mathcal{H}_n,\ s_{n+1}\in\mathcal{H}_{n+1}\}}\prod_{i=1}^n PT(\Upsilon_i,s_{i-1})
PT(\Upsilon_{n+1},s_n)=\\
\sum_{\{\Upsilon_1,\ldots ,\Upsilon_n\mid s_0\stackrel{\Upsilon_1}{\rightarrow}\cdots
\stackrel{\Upsilon_n}{\rightarrow}s_n,\ \mathcal{L}(\Upsilon_i)=A_i,\ s_i\in\mathcal{H}_i\ (1\leq i\leq n)\}}
\Bigl[\prod_{i=1}^n PT(\Upsilon_i,s_{i-1})\\
\left.\sum_{\{\Upsilon_{n+1}\mid s_n\stackrel{\Upsilon_{n+1}}{\rightarrow}s_{n+1},\
\mathcal{L}(\Upsilon_{n+1})=A_{n+1},\
s_n\in\mathcal{H}_n,\ s_{n+1}\in\mathcal{H}_{n+1}\}}PT(\Upsilon_{n+1},s_n)\right]=\\
\sum_{\{\Upsilon_1,\ldots ,\Upsilon_n\mid s_0\stackrel{\Upsilon_1}{\rightarrow}\cdots
\stackrel{\Upsilon_n}{\rightarrow}s_n,\ \mathcal{L}(\Upsilon_i)=A_i,\ s_i\in\mathcal{H}_i\ (1\leq i\leq n)\}}
\prod_{i=1}^n PT(\Upsilon_i,s_{i-1})PM_{A_{n+1}}(s_n,\mathcal{H}_{n+1})=\\
\sum_{\{\Upsilon_1,\ldots ,\Upsilon_n\mid s_0\stackrel{\Upsilon_1}{\rightarrow}\cdots
\stackrel{\Upsilon_n}{\rightarrow}s_n,\ \mathcal{L}(\Upsilon_i)=A_i,\ s_i\in\mathcal{H}_i\ (1\leq i\leq n)\}}
\prod_{i=1}^n PT(\Upsilon_i,s_{i-1})PM_{A_{n+1}}(\mathcal{H}_n,\mathcal{H}_{n+1})=\\
PM_{A_{n+1}}(\mathcal{H}_n,\mathcal{H}_{n+1})\sum_{\{\Upsilon_1,\ldots ,\Upsilon_n\mid
s_0\stackrel{\Upsilon_1}{\rightarrow}\cdots\stackrel{\Upsilon_n}{\rightarrow}s_n,\ \mathcal{L}(\Upsilon_i)=A_i,\
s_i\in\mathcal{H}_i\ (1\leq i\leq n)\}}\prod_{i=1}^n PT(\Upsilon_i,s_{i-1})=\\
PM_{A_{n+1}}(\mathcal{H}_n,\mathcal{H}_{n+1})\prod_{i=1}^n PM_{A_i}(\mathcal{H}_{i-1},\mathcal{H}_i)=\prod_{i=1}^{n+1}
PM_{A_i}(\mathcal{H}_{i-1},\mathcal{H}_i)$.

\end{itemize}

Let $s_0,\bar{s}_0\in\mathcal{H}_0$. We have\\
$PT(A_1\cdots A_n,s_0)=\sum_{\{\Upsilon_1,\ldots ,\Upsilon_n\mid s_0\stackrel{\Upsilon_1}{\rightarrow}\cdots
\stackrel{\Upsilon_n}{\rightarrow}s_n,\ \mathcal{L}(\Upsilon_i)=A_i,\ (1\leq i\leq n)\}}\prod_{i=1}^n
PT(\Upsilon_i,s_{i-1})=\\
\sum_{\mathcal{H}_1,\ldots ,\mathcal{H}_n}\sum_{\{\Upsilon_1,\ldots ,\Upsilon_n\mid
s_0\stackrel{\Upsilon_1}{\rightarrow} \cdots\stackrel{\Upsilon_n}{\rightarrow}s_n,\ \mathcal{L}(\Upsilon_i)=A_i,\
s_i\in\mathcal{H}_i\ (1\leq i\leq n)\}}\prod_{i=1}^n PT(\Upsilon_i,s_{i-1})=\\
\sum_{\mathcal{H}_1,\ldots ,\mathcal{H}_n}\prod_{i=1}^n PM_{A_i}(\mathcal{H}_{i-1},\mathcal{H}_i)=\\
\sum_{\mathcal{H}_1,\ldots ,\mathcal{H}_n}\sum_{\{\overline{\Upsilon}_1,\ldots ,\overline{\Upsilon}_n\mid
\bar{s}_0\stackrel{\overline{\Upsilon}_1}{\rightarrow}\cdots\stackrel{\overline{\Upsilon}_n}{\rightarrow}\bar{s}_n,\
\mathcal{L}(\overline{\Upsilon}_i)=A_i,\ \bar{s}_i\in\mathcal{H}_i\ (1\leq i\leq n)\}}\prod_{i=1}^n
PT(\overline{\Upsilon}_i,\bar{s}_{i-1})=\\
\sum_{\{\overline{\Upsilon}_1,\ldots ,\overline{\Upsilon}_n\mid\bar{s}_0\stackrel{\overline{\Upsilon}_1}{\rightarrow}
\cdots\stackrel{\overline{\Upsilon}_n}{\rightarrow}\bar{s}_n,\ \mathcal{L}(\overline{\Upsilon}_i)=\\
A_i,\ (1\leq i\leq n)\}}\prod_{i=1}^n PT(\overline{\Upsilon}_i,\bar{s}_{i-1})=PT(A_1\cdots A_n,\bar{s}_0)$.\\
Since we have the previous equality for all $s_0,\bar{s}_0\in\mathcal{H}_0$, we can denote
$PT(A_1\cdots A_n,\mathcal{H}_0)=\\
PT(A_1\cdots A_n,s_0)=PT(A_1\cdots A_n,\bar{s}_0)$.

By Proposition \ref{statprob.pro}, $\sum_{s\in\mathcal{H}\cap DR(G)}\varphi (s)=\sum_{s'\in\mathcal{H}\cap
DR(G')}\varphi '(s')$. We now can complete the proof:\\
$\sum_{s\in\mathcal{H}\cap DR(G)}\varphi (s)PT(\Sigma ,s)=\sum_{s\in\mathcal{H}\cap DR(G)}\varphi (s)PT(\Sigma
,\mathcal{H})=PT(\Sigma ,\mathcal{H})\sum_{s\in\mathcal{H}\cap DR(G)}\varphi (s)=\\
PT(\Sigma ,\mathcal{H})\sum_{s'\in\mathcal{H}\cap DR(G')}\varphi '(s')= \sum_{s'\in\mathcal{H}\cap DR(G')}\varphi
'(s')PT(\Sigma ,\mathcal{H})=\\
\sum_{s'\in\mathcal{H}\cap DR(G')}\varphi '(s')PT(\Sigma ,s')$. \qed

\subsection{Proof of Proposition \protect\ref{sjavevar.pro}}
\label{sjavevar.ssc}

Let us present two facts, which will be used in the proof.
\begin{enumerate}

\item By Proposition \ref{bissplit.pro}, $(DR(G)\cup DR(G'))/_\mathcal{R}=
((DR_{\rm T}(G)\cup DR_{\rm T}(G'))/_\mathcal{R})\uplus ((DR_{\rm V}(G)\cup DR_{\rm V}(G'))/_\mathcal{R})$. Hence,
$\forall\mathcal{H}\in (DR(G)\cup DR(G'))/_\mathcal{R}$, all states from $\mathcal{H}$ are tangible, when
$\mathcal{H}\in (DR_{\rm T}(G)\cup DR_{\rm T}(G'))/_\mathcal{R}$, or all of them are vanishing, when
$\mathcal{H}\in (DR_{\rm V}(G)\cup DR_{\rm V}(G'))/_\mathcal{R}$.

\item Let $\mathcal{H}\in (DR(G)\cup DR(G'))/_\mathcal{R}$ and $s_1,s_2\in\mathcal{H}$. We have
$\forall\widetilde{\mathcal{H}}\in (DR(G)\cup DR(G'))/_\mathcal{R},\\
\forall A\in\naturals_{\rm fin}^\mathcal{L},\ s_1\stackrel{A}{\rightarrow}_\mathcal{P}\widetilde{\mathcal{H}}\
\Leftrightarrow\ s_2\stackrel{A}{\rightarrow}_\mathcal{P}\widetilde{\mathcal{H}}$.
Hence, $PM(s_1,\widetilde{\mathcal{H}})= \sum_{\{\Upsilon\mid\exists\tilde{s}_1\in\widetilde{\mathcal{H}},\
s_1\stackrel{\Upsilon}{\rightarrow}\tilde{s}_1\}}PT(\Upsilon ,s_1)=\\
\sum_{A\in\naturals_{\rm fin}^\mathcal{L}}\sum_{\{\Upsilon\mid\exists\tilde{s}_1\in\widetilde{\mathcal{H}},\
s_1\stackrel{\Upsilon}{\rightarrow}\tilde{s}_1,\ \mathcal{L}(\Upsilon )=A\}}PT(\Upsilon ,s_1)=
\sum_{A\in\naturals_{\rm fin}^\mathcal{L}}PM_A(s_1,\widetilde{\mathcal{H}})=\\
\sum_{A\in\naturals_{\rm fin}^\mathcal{L}}PM_A(s_2,\widetilde{\mathcal{H}})=
\sum_{A\in\naturals_{\rm fin}^\mathcal{L}}\sum_{\{\Upsilon\mid\exists\tilde{s}_2\in\widetilde{\mathcal{H}},\
s_2\stackrel{\Upsilon}{\rightarrow}\tilde{s}_2,\ \mathcal{L}(\Upsilon )=A\}}PT(\Upsilon ,s_2)=\\
\sum_{\{\Upsilon\mid\exists\tilde{s}_2\in\widetilde{\mathcal{H}},\
s_2\stackrel{\Upsilon}{\rightarrow}\tilde{s}_2\}}PT(\Upsilon ,s_2)=PM(s_2,\widetilde{\mathcal{H}})$. Since we have
the previous equality for all\\
$s_1,s_2\in\mathcal{H}$, we can denote $PM(\mathcal{H},\widetilde{\mathcal{H}})=PM(s_1,\widetilde{\mathcal{H}})=
PM(s_2,\widetilde{\mathcal{H}})$. The transitions from the states of $DR(G)$ always lead to those from the same set,
hence, $\forall s\in DR(G),\\
PM(s,\widetilde{\mathcal{H}})=PM(s,\widetilde{\mathcal{H}}\cap DR(G))$. The same is true for $DR(G')$. Hence, for
all $s\in\mathcal{H}\cap DR(G)$, we obtain
$PM(\mathcal{H},\widetilde{\mathcal{H}})=PM(s,\widetilde{\mathcal{H}})=PM(s,\widetilde{\mathcal{H}}\cap
DR(G))=PM(\mathcal{H}\cap DR(G),\widetilde{\mathcal{H}}\cap DR(G))$. The same is true for $DR(G')$. Finally,
$PM(\mathcal{H}\cap DR(G),\widetilde{\mathcal{H}}\cap
DR(G))=PM(\mathcal{H},\widetilde{\mathcal{H}})=\\
PM(\mathcal{H}\cap DR(G'),\widetilde{\mathcal{H}}\cap DR(G'))$.

\end{enumerate}

Let us now prove the proposition statement for the sojourn time averages.
\begin{itemize}

\item Let $\mathcal{H}\in (DR_{\rm V}(G)\cup DR_{\rm V}(G'))/_\mathcal{R}$.\\
We have $\mathcal{H}\cap DR(G)=\mathcal{H}\cap DR_{\rm V}(G)\in DR_{\rm V}(G)/_\mathcal{R}$ and
$\mathcal{H}\cap DR(G')=\mathcal{H}\cap DR_{\rm V}(G')\in DR_{\rm V}(G')/_\mathcal{R}$.
By definition of the average sojourn time in an equivalence class of states, we get ${\it SJ}_{\mathcal{R}\cap
(DR(G))^2}(\mathcal{H}\cap DR(G))={\it SJ}_{\mathcal{R}\cap (DR(G))^2}(\mathcal{H}\cap DR_{\rm V}(G))=0=\\
{\it SJ}_{\mathcal{R}\cap (DR(G'))^2}(\mathcal{H}\cap DR_{\rm V}(G'))={\it SJ}_{\mathcal{R}\cap
(DR(G'))^2}(\mathcal{H}\cap DR(G'))$.

\item Let $\mathcal{H}\in (DR_{\rm T}(G)\cup DR_{\rm T}(G'))/_\mathcal{R}$.\\
We have $\mathcal{H}\cap DR(G)=\mathcal{H}\cap DR_{\rm T}(G)\in DR_{\rm T}(G)/_\mathcal{R}$ and
$\mathcal{H}\cap DR(G')=\mathcal{H}\cap DR_{\rm T}(G')\in DR_{\rm T}(G')/_\mathcal{R}$.
By definition of the average sojourn time in an equivalence class of states, we get\\
${\it SJ}_{\mathcal{R}\cap (DR(G))^2}(\mathcal{H}\cap DR(G))\!=\!{\it SJ}_{\mathcal{R}\cap (DR(G))^2}(\mathcal{H}\cap
DR_{\rm T}(G))\!=\!\frac{1}{1-PM(\mathcal{H}\cap DR_{\rm T}(G),\mathcal{H}\cap DR_{\rm T}(G))}\!=\\
\frac{1}{1-PM(\mathcal{H}\cap DR(G),\mathcal{H}\cap DR(G))}=\frac{1}{1-PM(\mathcal{H},\mathcal{H})}=
\frac{1}{1-PM(\mathcal{H}\cap DR(G'),\mathcal{H}\cap DR(G'))}=\\
\frac{1}{1-PM(\mathcal{H}\cap DR_{\rm T}(G'),\mathcal{H}\cap DR_{\rm T}(G'))}\!\!=\!\!
{\it SJ}_{\mathcal{R}\cap (DR(G'))^2}(\mathcal{H}\cap DR_{\rm T}(G'))\!\!=\!\!
{\it SJ}_{\mathcal{R}\cap (DR(G'))^2}(\mathcal{H}\cap DR(G'))$.

\end{itemize}

Thus, $\forall\mathcal{H}\in (DR(G)\cup DR(G'))/_\mathcal{R}$ we have ${\it SJ}_{\mathcal{R}\cap
(DR(G))^2}(\mathcal{H}\cap DR(G))=\\
{\it SJ}_{\mathcal{R}\cap (DR(G'))^2}(\mathcal{H}\cap DR(G'))$.

The proposition statement for the sojourn time variances is proved similarly. \qed

\end{document}